\pdfoutput=1
\PassOptionsToPackage{bookmarksnumbered,unicode,hypertexnames=false}{hyperref}
\RequirePackage{hyperref}
\documentclass[acmtog]{acmart}
\makeatletter
\g@addto@macro\@authornotes{\footnotetext[1]{Both authors contributed equally to this work.}}
\makeatother

\usepackage{arydshln}
\usepackage{multirow}
\usepackage{threeparttable}
\usepackage[percent]{overpic}
\usepackage{booktabs}
\usepackage{colortbl}
\definecolor{rowgray}{gray}{0.93}
\usepackage{amsmath}
\usepackage{xspace}
\usepackage{wrapfig}
\usepackage[ruled]{algorithm2e}
\usepackage{bm}
\usepackage{makecell}
\usepackage{float}
\usepackage{caption}
\usepackage{subcaption}
\usepackage{amsthm}

\newtheorem{proposition}{Proposition}
\newtheorem{lemma}{Lemma}
\usepackage{enumitem}

\acmJournal{TOG}
\setcopyright{cc}
\setcctype{by}
\acmJournal{TOG}
\acmYear{2026} \acmVolume{45} \acmNumber{6} \acmArticle{197}
\acmMonth{12} \acmDOI{10.1145/3842510}

\newcommand{\name}{S4R\xspace}

\begin{document}

\title{S4R: Scaling for Rigid-Body Interpenetration Resolution}

\author{Zhiyang Dou}
\authornotemark[1]
\email{frankdou@mit.edu}
\orcid{0000-0003-0186-8269}
\affiliation{%
  \institution{MIT CSAIL}
  \city{Cambridge}
  \country{USA}
}

\author{Ang Zhao}
\authornotemark[1]
\email{heskeyangzhao@gmail.com}
\orcid{0009-0004-4661-0530}
\affiliation{%
  \institution{Xiamen University}
  \city{Xiamen}
  \country{China}
}

\author{Chen Peng}
\email{millyapeng@gmail.com}
\orcid{0000-0002-3102-2574}
\affiliation{%
  \institution{The University of Hong Kong}
  \city{Hong Kong}
  \country{China}
}

\author{Minghao Guo}
\email{guomh2014@gmail.com}
\orcid{0000-0003-3408-4997}
\affiliation{%
  \institution{MIT CSAIL}
  \city{Cambridge}
  \country{USA}
}

\author{Haixu Wu}
\email{wuhaixu@mit.edu}
\orcid{0009-0007-5715-0527}
\affiliation{%
  \institution{MIT CSAIL}
  \city{Cambridge}
  \country{USA}
}

\author{Cheng Lin}
\email{chenglin@must.edu.mo}
\orcid{0000-0002-3335-6623}
\affiliation{%
  \institution{Macau University of Science and Technology}
  \city{Macau}
  \country{China}
}

\author{Yuan Liu}
\email{yuanly@ust.hk}
\orcid{0000-0003-2933-5667}
\affiliation{%
  \institution{The Hong Kong University of Science and Technology}
  \city{Hong Kong}
  \country{China}
}

\author{Junfeng Yao}
\email{yao0010@xmu.edu.cn}
\orcid{0000-0002-2330-7406}
\affiliation{%
  \institution{Xiamen University}
  \city{Xiamen}
  \country{China}
}

\author{Xiaohu Guo}
\email{xguo@utdallas.edu}
\orcid{0000-0002-7610-6923}
\affiliation{%
  \institution{The University of Texas at Dallas}
  \city{Dallas}
  \country{USA}
}

\author{Wenping Wang}
\email{wenping@tamu.edu}
\orcid{0000-0002-2284-3952}
\affiliation{%
  \institution{Texas A\&M University}
  \city{College Station}
  \country{USA}
}

\author{Wojciech Matusik}
\email{wojciech@csail.mit.edu}
\orcid{0000-0003-0212-5643}
\affiliation{%
  \institution{MIT CSAIL}
  \city{Cambridge}
  \country{USA}
}

\renewcommand{\shortauthors}{Dou et al.}


\begin{teaserfigure}
  \includegraphics[width=\textwidth]{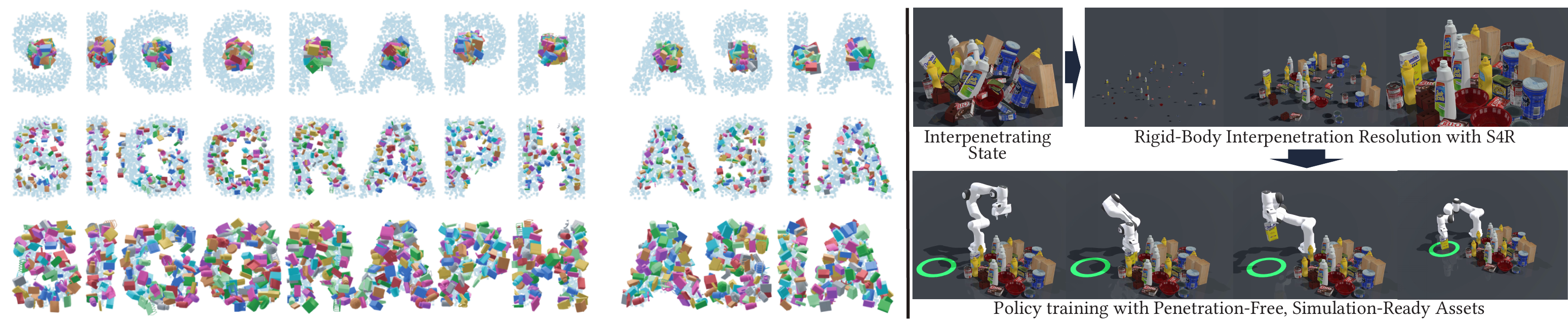}
  \caption{\textbf{S4R resolves rigid-body interpenetrations at scale and produces simulation-ready scenes.}
\textbf{Left:} Progressive scaling of 1,000 object instances from the Kubric~\cite{greff2022kubric} pool, arranged to spell ``SIGGRAPH ASIA,'' from an interpenetrating initialization to a penetration-free configuration.
\textbf{Right:} S4R efficiently converts cluttered 3D assets into simulation-ready scenes for downstream robot-policy training.
}
  \Description{Two-part teaser. On the left, three rows show one thousand small mesh assets arranged in the shape of the words SIGGRAPH ASIA: an interpenetrating initial state, an intermediate scaled state, and the resolved full-scale layout. On the right, a cluttered pile of household objects is shrunk and restored into a separated tabletop arrangement, which a robot arm then manipulates, placing an object into a marked circular target region on the ground.}
  \label{fig:teaser}
\end{teaserfigure}

\begin{abstract}
Rigid-body interpenetration frequently occurs in procedurally assembled and
generated scenes and must be removed before downstream applications such as physical simulation. We present
\name\ (\emph{S}caling \emph{for} Rigid-Body Interpenetration \emph{R}esolution), a scale-continuation method for static interpenetration repair. \name
first uniformly shrinks each body about a fixed reference center to a small
initial scale, at which the layout is penetration-free, and then restores full scale through a sequence
of minimum-norm convex contact quadratic programs (QPs) that target the linearized separation
margin during continuation. Resolution thereby replaces one deep correction with a sequence of shallow-contact subproblems. A conservative scale-event bound and frozen-witness
gap predictions cut the number of exact mesh queries; the continuation then ends
with a full-scale evaluator check and bounded tail refinement. We evaluate \name on
Kubric~\cite{greff2022kubric}, HY3D-Bench~\cite{tencent2026hy3dbench}, and
Thingi10K~\cite{zhou2016thingi10k} using a shared mesh-level evaluator and a
unified per-scene timing protocol.  In the main comparisons on all three benchmarks, up to $N{=}5000$ bodies,
\name reaches zero reported penetration with displacement that stays small and
nearly independent of scene size, and at the lowest wall time within each hardware tier among the compared methods. A GPU implementation extends these results to large-scale scenes. Our code and data can be found on our project page: \textcolor{ACMDarkBlue}{\url{https://frank-zy-dou.github.io/projects/S4R/index.html}}.
\end{abstract}

\maketitle

\providecommand{\vx}{}\renewcommand{\vx}{\mathbf{x}}
\providecommand{\vp}{}\renewcommand{\vp}{\mathbf{p}}
\providecommand{\vc}{}\renewcommand{\vc}{\mathbf{c}}
\providecommand{\vq}{}\renewcommand{\vq}{\mathbf{q}}
\providecommand{\vr}{}\renewcommand{\vr}{\mathbf{r}}
\providecommand{\vv}{}\renewcommand{\vv}{\mathbf{v}}
\providecommand{\vn}{}\renewcommand{\vn}{\mathbf{n}}
\providecommand{\vw}{}\renewcommand{\vw}{\mathbf{w}}
\providecommand{\vu}{}\renewcommand{\vu}{\mathbf{u}}
\providecommand{\vf}{}\renewcommand{\vf}{\mathbf{f}}
\providecommand{\vy}{}\renewcommand{\vy}{\mathbf{y}}
\providecommand{\vd}{}\renewcommand{\vd}{\mathbf{d}}
\providecommand{\vb}{}\renewcommand{\vb}{\mathbf{b}}
\providecommand{\vomega}{}\renewcommand{\vomega}{\boldsymbol{\omega}}
\providecommand{\vlambda}{}\renewcommand{\vlambda}{\boldsymbol{\lambda}}

\providecommand{\mA}{}\renewcommand{\mA}{\mathbf{A}}
\providecommand{\mM}{}\renewcommand{\mM}{\mathbf{M}}
\providecommand{\mI}{}\renewcommand{\mI}{\mathbf{I}}
\providecommand{\mH}{}\renewcommand{\mH}{\mathbf{H}}
\providecommand{\mR}{}\renewcommand{\mR}{\mathbf{R}}

\providecommand{\R}{}\renewcommand{\R}{\mathbb{R}}

\providecommand{\dhat}{}\renewcommand{\dhat}{\hat{d}}
\providecommand{\qtilde}{}\renewcommand{\qtilde}{\tilde{\vq}}

\providecommand{\Tr}{}\renewcommand{\Tr}{\operatorname{Tr}}
\providecommand{\argmin}{}\renewcommand{\argmin}{\operatorname*{arg\,min}}

\section{Introduction}

\label{sec:intro}

Object-rich 3D content is now produced at scale, from text-to-3D assets~\cite{poole2023dreamfusion,lin2023magic3d} to large synthetic indoor environments --- procedurally generated~\cite{deitke2022procthor}, professionally or human-authored~\cite{fu20213dfront,khanna2024hssd}, and synthesized by generative scene models~\cite{paschalidou2021atiss,zhai2023commonscenes,tang2024diffuscene,yang2024holodeck,yang2024physcene,yang2024scenecraft,wang2024architect,li2024discene}. In robotics, tabletop and embodied-simulation pipelines likewise instantiate and rearrange large collections of rigid assets for interaction, manipulation, and policy learning~\cite{greff2022kubric,szot2021habitat2,gu2023maniskill2,wang2024gensim,wang2024robogen,lee2025dynscene}. However, assembling independently modeled assets into a shared layout does not by itself guarantee mesh-level non-penetration: interpenetration and implausible placement are among the reported failure modes of generated scenes~\cite{zhai2023commonscenes}, and physics- and collision-aware guidance in scene synthesis sets out to reduce them~\cite{yang2024physcene}. Resolving those overlaps with minimal layout change is therefore the step that turns generated content into simulation-ready scenes for physics simulation~\cite{todorov2012mujoco, coumans2021pybullet, isaacgym, Li2020IPC, Lan2022ABD}, motion planning~\cite{schulman2014trajopt}, and robot-policy training~\cite{gu2023maniskill2, wang2024robogen}. In our deep-overlap initialization tests (App.~\ref{sec:exp_engine_frontend_app}), standard physics engines~\cite{todorov2012mujoco, coumans2021pybullet, isaacgym} respond with large corrective impulses that rapidly displace the bodies, and barrier-based solvers such as IPC presuppose an intersection-free state~\cite{Li2020IPC}.

Existing approaches to penetration resolution fall into three broad categories, each with characteristic trade-offs.
\emph{Projection methods} detect overlapping pairs and iteratively push them apart along estimated contact normals; their behavior depends on the quality of the contact points and normals~\cite{Erleben2018Methodology}, and in our dense non-convex tests our projection baseline (PD-PGS, Sec.~\ref{sec:exp_setup}) does not always reach feasibility within its budget (App.~\ref{sec:exp_sphere}).
\emph{Dynamics-based methods}, including recent augmented-penalty solvers such as AVBD~\cite{Giles2025AVBD} and intersection-repair schemes such as ISIR~\cite{Jang2025ISIR}, use forward simulation or gradient-flow energies to push bodies apart. However, in static deep-overlap benchmarks, the resulting motion can move bodies far from the intended layout. \emph{Optimization-based formulations}~\cite{Otaduy2009Implicit, schulman2014trajopt, Li2020IPC, Lan2022ABD} treat contact feasibility as constrained optimization --- originally for contact dynamics or trajectory planning --- and can recover low-displacement solutions when adapted to static repair, but the adapted formulations either rely on expensive nonlinear programming (Drake with Ipopt~\cite{wachter2006ipopt}), global contact quadratic programs (QPs) whose single linearization is unreliable at deep initial penetration and must be repaired by repeated outer re-linearization, or IPC log-barrier solvers that require a full Newton stack --- contact barrier, continuous collision detection (CCD) line search, backtracking, and a stiff linear solve (ABD~\cite{Lan2022ABD}).
Where these methods do handle deep or hard contact, they do so by substepping in time under full dynamics or by barrier-guarded Newton steps; in this paper, we instead propose to substep along a \emph{scale} homotopy, which needs no dynamics and keeps every subproblem a shallow-contact convex QP (Sec.~\ref{sec:progressive_scaling}--\ref{sec:per_step_qp}).
Signed-distance-field approaches~\cite{Macklin2019SDF} provide useful distance and gradient queries, but require object-level field construction with memory--resolution trade-offs or precomputation; we instead use on-the-fly mesh proximity.

In this paper, we address the infeasible initial state by introducing a scale continuation from an evaluator-checked separated configuration: we propose \name: \emph{S}caling \emph{for} Rigid-Body Interpenetration \emph{R}esolution. Specifically, \name first scales all candidate bodies down to a small initial scale, at which the layout becomes \emph{penetration-free}, and then gradually scales them back to their original size while optimizing their reference-center positions (see Fig.~\ref{fig:method_overview}); the optional variant of Sec.~\ref{sec:exp_rotation} also updates orientation. By shrinking each rigid body to a sufficiently small fraction of its intended size, most practical layouts become intersection-free; degenerate reference-center coincidences are handled by the perturbation rule of Sec.~\ref{sec:progressive_scaling}.
We then progressively scale the bodies back to full size, at each increment solving a small sparse QP that computes the minimum-norm incremental displacement needed to resolve newly activated contacts.
The continuation replaces one deep-overlap linearization with a sequence of mostly shallow-contact QPs; stale or newly activated contacts are recovered by periodic re-detection and the final refinement.
The per-step minimum-norm correction --- rather than a previous-step proximal term or a deep-penetration impulse --- is what empirically keeps \name's root-mean-square displacement (RMSD) relatively low: every QP pushes bodies only as far as the current shallow-contact constraints require, and the active-set graph at each scaling step stays sparse in our benchmarks.

Exposed in this form, progressive scaling admits an analytical acceleration: under pure uniform scaling at fixed poses, every surface-vertex trajectory is \emph{affine} in the uniform scale parameter $s\in(0,1]$.
For a contact pair $(i,j)$ with unit contact direction $\mathbf{n}_{ij}$ (from body $i$ toward body $j$), the gap between its current closest points can therefore close by at most $e_i(\mathbf{n}_{ij})+e_j(-\mathbf{n}_{ij})$ per unit of scale, i.e.\ its rate of change is bounded below by $-e_i(\mathbf{n}_{ij})-e_j(-\mathbf{n}_{ij})$, where $e_i(\mathbf{n})$ is how far body $i$ reaches from its reference center in the direction $\mathbf{n}$ (its clamped one-sided support, defined in Sec.~\ref{sec:linearized_constraints}; the bound is derived in Sec.~\ref{sec:conservativeness}).
This observation has two practical consequences.
First, a closed-form lower bound on the first-contact scale between any pair lets us prune well-separated pairs before any narrow-phase query runs.
Second, between full narrow-phase invocations the frozen-witness gap prediction for every known pair can be updated analytically by projecting the QP-predicted displacement onto the cached contact normal, so the expensive mesh-proximity primitive runs only every few scale steps.
Combined with a coarse uniform scale schedule, adaptive skipping at contact-free scales, and a sparse QP that only introduces variables for bodies in the active-contact set, these accelerations reduce both the number of QP variables and the frequency of expensive mesh-proximity queries in the query-dominated regime analyzed in App.~\ref{sec:theory}.
For tabletop generated-asset repair, the same QP can also absorb hard scene-structure constraints: we return every object to upright, attach its lowest point to a common support plane, and let \name resolve interpenetration through admissible in-plane translations and yaw rotations.
We summarize our contributions below:

\begin{itemize}[leftmargin=1.2em, itemsep=0.2em, topsep=0.2em, parsep=0pt, partopsep=0pt]
  \item A scale-continuation formulation for static rigid-body interpenetration repair, in which full-scale recovery is decomposed into a sequence of minimum-norm convex contact QPs with a conservative one-sided closure model.
  \item A practical path-following algorithm that reduces exact mesh queries using a conservative scale-event bound and frozen-witness gap predictions, then closes with a full-scale evaluator check and bounded refinement. 
  \item A broad evaluation across public mesh collections and different application settings, together with ablations and analyses. 
\end{itemize}

\section{Related Work}
\label{sec:related}
\paragraph{Penetration Resolution as Optimization.}
Resolving an infeasible configuration of interpenetrating rigid bodies into a feasible one is an old problem with many partial solutions.
\emph{Projection methods} iteratively separate overlapping pairs along estimated contact normals; they are fast, but their behavior depends on the quality of the local contact points and normals, which is itself hard to guarantee on complex meshes~\cite{Erleben2018Methodology}; in our experiments the iteration fails to converge within its budget on dense non-convex entanglement. Velocity-level projection formulations such as Staggered Projections~\cite{Kaufman2008SCF} enforce frictional contact constraints within dynamic simulation and target a different setting from static repair. Inside a running simulator, position-level error is classically kept small by constraint stabilization~\cite{baumgarte1972} and post-stabilization projections~\cite{cline2003post}, which remove per-step drift; the deep static overlap addressed here lies far outside that regime.
\emph{Hard-constraint nonlinear programming (NLP)} formulations such as Drake's~\cite{drake} minimum-distance constraints under Ipopt~\cite{wachter2006ipopt} satisfy the constraints to tolerance when they converge, but scale poorly in our benchmark.
The \emph{global contact QP/linear complementarity problem (LCP)} baseline linearizes the distance constraints at the current pose and re-linearizes one large QP until feasible, so its accuracy is limited by the error of that linearization at depth. We solve those QPs with OSQP~\cite{stellato2020osqp}.
\emph{Soft-penalty continuation} ramps a penalty on penetration depth; our baseline follows TrajOpt's hinge collision loss and outer-loop penalty continuation~\cite{schulman2014trajopt}, with a squared hinge as our static adaptation. It converges locally in practice, but every objective evaluation of its L-BFGS-B inner solve re-queries all $N(N{-}1)/2$ pairs, which is what limits it at scale.
Hard-contact solvers built on barriers advance in \emph{time} under full dynamics with safeguarded steps~\cite{Li2020IPC, lan2022pfpd}; in this paper we instead substep in \emph{scale}: penetration resolution becomes path-following continuation along a scale parameter.

\paragraph{Penalty-Force and Gradient-Flow Methods.}
AVBD~\cite{Giles2025AVBD} extends vertex block descent with an augmented-Lagrangian formulation that supports hard constraints, stiff systems, and stable rigid-body stacking; ISIR~\cite{Jang2025ISIR} repairs self-intersections of a static surface mesh by flowing its vertex positions along local signed tangent-point energies.
The AVBD implementation we benchmark is the official 3D demo with oriented-bounding-box (OBB)--OBB separating-axis (SAT) collision; in our static non-convex layouts this proxy can inflate apparent overlap and produce substantially larger displacements than mesh-level optimization baselines.
Our benchmark uses a rigid-pose adaptation of ISIR that maps its vertex-level updates to per-body rigid poses; Sec.~\ref{sec:exp_setup} provides the implementation details and evaluation results.
The two target different settings --- dynamic contact simulation for AVBD, deformable surface repair for ISIR --- and we report how they behave once adapted to static multi-body resolution.

\paragraph{Reaching Feasibility by Flows and Shrink-and-Recover.}
A complementary family reaches an intersection-free state by \emph{flowing} the geometry there and preserving feasibility once attained.
Cloth untangling recovers from intersecting states via global intersection analysis~\cite{baraff2003untangling} or intersection-contour minimization~\cite{volino2006icm}; history-free collision response separates initially penetrating oriented surface pairs without trajectory history~\cite{Ye2014HistoryFree}; and air meshes~\cite{Muller2015PBD} restore separation through an auxiliary volume mesh.
In geometry processing, one line preserves local~\cite{schuller2013lim} or global~\cite{smith2015bijective, fang2021injective} injectivity from a feasible initialization and never leaves the feasible set, while another explicitly untangles non-injective or foldover inputs~\cite{du2020lifting, garanzha2021foldover}. Tangent-point energies drive self-avoiding flows of curves and surfaces~\cite{yu2021repulsivecurves, yu2021repulsivesurfaces}; Repulsive Shells builds the same repulsion into a shape-space metric for collision-aware shell deformation~\cite{sassen2024repulsiveshells}; \citet{Chen2023ShortestPath} compute shortest paths to the boundary of self-intersecting meshes; and, most recently, shape-and-mesh repulsion energies enforce embeddedness for surfaces and multi-object configurations~\cite{Minarcik2026Untangling}. Our Rigid-ISIR baseline (Sec.~\ref{sec:exp_setup}) adapts ISIR's~\cite{Jang2025ISIR} vertex-level flow to rigid poses.
\name\ shares the shrink-to-feasible strategy, but follows a homotopy in a single scale parameter: vertex trajectories are affine in $s$ under pure inflation at fixed poses, so each step is a small convex QP and the next fixed-pose inflation event admits a closed-form conservative bound rather than being discovered by integrating a stiff flow.

\paragraph{Incremental Potential Contact.}
IPC~\cite{Li2020IPC} provides a log-barrier contact energy and a CCD-filtered line search that together make intersection-free \emph{trajectories} a built-in invariant of Newton-style solvers.
Extensions to codimensional geometry~\cite{Li2021CIPC} and rigid/affine bodies~\cite{Ferguson2021RigidIPC, Lan2022ABD} have made IPC a standard tool for robust contact in deformable and rigid-body simulation.
IPC also presupposes an intersection-free state --- its barrier acts on unsigned primitive distances and its CCD-filtered line search only preserves separation --- so it can maintain an intersection-free state but cannot establish one from an interpenetrating input; in the engine-initialization experiment of Sec.~\ref{sec:exp_engine_frontend}, \name\ produces this precondition from interpenetrating inputs.
The full Newton and feasibility-preserving line-search stack is heavy for static repair (the linear solver itself is an implementation choice), however: once bodies start separated at small scale and the schedule keeps modeled contacts shallow, one linearization per step suffices in our benchmarks.
When a step is hard to resolve, these solvers shorten it: within an implicit time step, continuous collision detection bounds the nonlinear line search so the iterates stay intersection-free~\cite{Li2020IPC}, and recent barrier variants set the barrier stiffness dynamically from the elasticity~\cite{ando2024cubic}. \name\ subdivides along the scale axis instead, where Sec.~\ref{sec:soi} provides a conservative event bound for pure inflation at fixed poses.

\paragraph{Affine Body Dynamics.}
\citet{Lan2022ABD} represent each body by its 12 affine DOFs $(\mA\in\R^{3\times3}, \vp\in\R^3)$ and softly enforce $\mA^\top\mA=\mI$ via an orthogonality energy.
The representation is smooth and differentiable and couples well with IPC.
Our work does \emph{not} require the full ABD machinery: the schedule is designed to keep modeled contacts shallow (unmodeled contacts are caught by exact refreshes and final verification), and the small per-step QP operates directly on body reference centers; the default solver fixes orientations, and the optional rotational variant updates them on $\mathrm{SO}(3)$ (App.~\ref{ap:6dof}).

\paragraph{Mesh-Level Collision.}
Fast mesh-level distance and intersection queries are provided by several libraries, including the Flexible Collision Library (FCL)~\cite{pan2012fcl} for bounding-volume-hierarchy (BVH) distance and collision, libigl~\cite{jacobson2018libigl} for computational-geometry primitives, and CGAL~\cite{cgal} for robust polygon-mesh processing. We use FCL, though the algorithm only requires mesh-distance, collision, witness, and normal queries; it does not assume the backend supplies a globally smooth signed-distance function.

\paragraph{Irregular Packing and Layout.}
A related line of work studies non-overlapping layouts of irregular shapes for mosaics, collage, artistic packing, and decorative surface design.
\citet{reinert2013interactive} generate artistic 2D packing layouts from user examples; \citet{hu2015surface} synthesize surface mosaics by starting from shrunk, non-overlapping irregular tiles and alternating continuous configuration optimization with combinatorial tile selection to improve coverage.
Later work extends irregular packing to 3D object containers~\cite{ma2018packing} and learning-based scalable packing~\cite{xue2023learning}.
Scaling there is a layout device: Hu et al.\ start tiles shrunk and optimize their scales, whose final values need not equal one, while the other systems place fixed-size shapes; here the bodies are deeply interpenetrating in 3D and must return to their given full scale.

\paragraph{Continuation and Homotopy Methods.}
Progressive scaling is an instance of \emph{numerical continuation}~\cite{allgower2003introduction}, where a difficult problem is connected to an easy one through a parameterized family.
Under fixed-pose uniform scaling, vertex trajectories are affine in $s$, yielding analytically tractable event bounds. Between exact detections, the frozen-witness gap model is also affine, enabling the cache updates of Sec.~\ref{sec:cache} and the local QP warm starts of Sec.~\ref{sec:schedule}.
\section{Method}

\name resolves deep rigid-body interpenetration by first shrinking the bodies to a small initial scale at which the layout is penetration-free, and then progressively restoring them to full scale while applying minimum-norm incremental contact corrections. Sec.~\ref{sec:problem} states the reference problem and explains why linearizing it directly is unreliable under deep penetration. Sec.~\ref{sec:progressive_scaling} introduces the scale continuation that replaces one deep correction with a sequence of shallow-contact steps. Sec.~\ref{sec:per_step_qp} derives the minimum-norm contact QP solved at each step, whose frozen-witness closure term is conservative by construction. Sec.~\ref{sec:accelerated} accelerates the continuation with scale-event skipping, analytical distance updates, and warm-starting, and Sec.~\ref{sec:tail} handles the residual contacts of the full-scale state with bounded tail refinement.

\subsection{Problem Formulation}
\label{sec:problem}

We consider a system of $N$ rigid bodies with a user-prescribed initial state $\mathbf{x}^0 = (\mathbf{p}^0, \mathbf{R}^0)$, where $\mathbf{p}^0 \in \mathbb{R}^{3N}$ and $\mathbf{R}^0 \in \mathrm{SO}(3)^N$ denote the initial positions and orientations, respectively. In practical scenarios---such as procedural scene generation or physics initialization---this input configuration $\mathbf{x}^0$ is often physically invalid, featuring severe mutual interpenetrations. A natural reference objective is to eliminate all penetrations while changing the intended layout as little as possible:
\begin{align}
  \min_{\mathbf{p}^\star, \mathbf{R}^\star} \quad & \sum_{i=1}^N \left( \frac{1}{2} \|\mathbf{p}_i^\star - \mathbf{p}_i^0\|^2 + \frac{\beta}{2} d_{\mathrm{SO}(3)}(\mathbf{R}_i^\star, \mathbf{R}_i^0)^2 \right) \label{eq:reference_problem} \\
  \text{s.t.} \quad & d_{ij}(\mathbf{x}^\star) \ge \hat{d}, \quad \forall i \neq j, \notag
\end{align}
where $d_{ij}(\mathbf{x}^\star)$ denotes the signed clearance between bodies $i$ and $j$ at state $\mathbf{x}^\star$ --- positive when the two closed bodies are separated, negative when their interiors overlap --- $\hat{d} > 0$ is a clearance margin, $d_{\mathrm{SO}(3)}$ is the geodesic distance on the rotation group, and $\beta$ is a weighting factor. $d_{ij}$ is a conceptual quantity, and Eq.~\ref{eq:reference_problem} serves only as the reference goal; no solver evaluates it globally. One option is to linearize its constraints at the current state and solve the resulting QP directly; Sec.~\ref{sec:direct} explains why this is unreliable under deep penetration. \name\ instead follows a scale continuation (Sec.~\ref{sec:progressive_scaling}) and solves a sequence of minimum-norm \emph{incremental} QPs.

\subsubsection{Limits of Direct Linearization.}
\label{sec:direct}
\label{sec:theory_linearization}

Direct QP methods resolve overlaps by repeatedly linearizing the non-convex
distance constraint at the current state $\mathbf{x}$. For this local argument let
$\boldsymbol{\xi}$ be Euclidean pose coordinates in a fixed local chart --- translation alone
in the deployed solver, plus a tangent-space rotation increment in the optional
6-DOF variant --- and write $\mathbf{x}\oplus\Delta\boldsymbol{\xi}$ for the pose reached from $\mathbf{x}$ by
that increment. Then
\begin{equation*}
  d_{ij}(\mathbf{x}\oplus\Delta \boldsymbol{\xi})
  \approx
  d_{ij}(\mathbf{x})+ \mathbf{J}_{ij}(\mathbf{x})\,\Delta \boldsymbol{\xi}
  \geq \hat d ,
\end{equation*}
with $\mathbf{J}_{ij}$ the Jacobian of the pair distance in those coordinates.
When $d_{ij}(\mathbf{x})\ll 0$ --- that is,
when the penetration depth $\delta_{ij}(\mathbf{x}):=-d_{ij}(\mathbf{x})$ is large relative to
the local geometric scale and the pair is deeply rather than shallowly
interpenetrating --- this first-order proxy can be unreliable for two reasons. \textbf{Geometric discontinuity.}
On a stable smooth witness branch, the gradient of $d_{ij}$ w.r.t.\ the relative translation $\mathbf{p}_j-\mathbf{p}_i$ is the unit direction
defined by the closest-point pair. Under deep penetration, the distance query
can approach a medial-axis or feature-switch region, where the closest feature
is non-unique and the normal can change abruptly under small perturbations
(see the detailed discussion in Fig.~\ref{fig:medial_axis}). \textbf{Linearization residual.}
Even on a stable $C^2$ distance branch --- for instance, when the closest feature of a vertex of body $i$ remains one spherical patch of body $j$ throughout the update, so that the distance is a smooth function of the pose --- the tangent approximation discards
the nonlinear remainder (written here for a translation-only setting,
where $\oplus$ is ordinary addition and $\Delta\boldsymbol{\xi}=\Delta \mathbf{x}$)
\[
R_{ij}(\Delta \mathbf{x})
=
d_{ij}(\mathbf{x}+\Delta \mathbf{x})-d_{ij}(\mathbf{x})
-\nabla d_{ij}(\mathbf{x})^{\top}\Delta \mathbf{x} .
\]
By Taylor's theorem, its magnitude is locally bounded by
$\tfrac12 L_{\mathrm{seg}}\|\Delta \mathbf{x}\|^{2}$, provided that the update
segment remains on the same branch, where $L_{\mathrm{seg}}$ uniformly
bounds the Hessian norm of the distance function along that segment
(Lemma~\ref{lem:taylor_residual}, App.~\ref{ap:theory_linearization}).
Deep penetration often requires a large corrective update and may also place
the query near a medial-axis or feature-switch region, increasing the residual
through $\|\Delta \mathbf{x}\|$ and $L_{\mathrm{seg}}$, respectively. Scale
continuation instead replaces one deep-overlap linearization with a sequence
of mostly shallow-contact subproblems, for which the local residual is
quadratic in the realized per-step update magnitude (App.~\ref{ap:theory_linearization}).

\begin{figure*}[t]
  \centering
  \includegraphics[width=\textwidth]{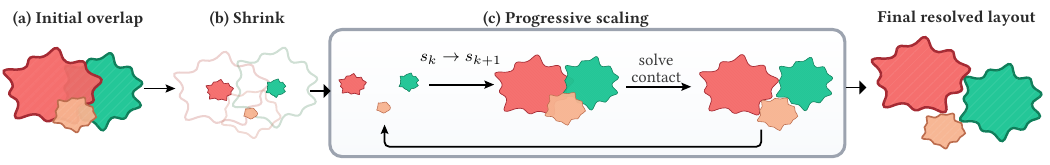}
  \caption{Method overview. Given an initially interpenetrating layout (a), \name\ shrinks each body about its reference center to the fixed initial scale $s_{\min}$ (b), at which the layout is penetration-free, then restores full scale along a continuation path (c) in which every step inflates the geometry and solves a minimum-norm contact QP. A final full-scale evaluator pass identifies residual contacts, which are handled by bounded tail refinement.}
  \Description{A pipeline diagram showing an initial overlapping set of three irregular rigid bodies, a shrink step that makes them small and separated, a progressive scaling loop with scale advancement and contact solving, and a final separated full-scale layout.}
  \label{fig:method_overview}
\end{figure*}

\subsection{Progressive Scaling}
\label{sec:progressive_scaling}
\name\ approaches the reference problem by continuation in a single scale parameter; Fig.~\ref{fig:method_overview} summarizes the construction. Our key insight is that direct first-order resolution from deep, global interpenetration is poorly conditioned by the aforementioned errors, whereas shallow-contact resolution is much more stable. To bridge this gap, we reformulate the problem using numerical continuation along a uniform scale parameter $s \in (0, 1]$. Let $\bar{\mathbf{v}}$ denote a rest-shape vertex of body $i$ in its local frame. We define its world-space trajectory during scaling as:
\begin{equation*}
  \mathbf{v}_i(s) = \mathbf{p}_i + s\,\mathbf{R}_i \bar{\mathbf{v}}.
\end{equation*}

Here $\mathbf{p}_i$ is the world-space position of the same reference center used for scaling and for displacement measurement. Under uniform scaling about this center, body $i$ at scale $s$ is contained in a ball of radius $s r_i$, centered at $\mathbf{c}_i:=\mathbf{p}_i$, where $r_i:=\max_{\bar{\mathbf{v}}}\|\bar{\mathbf{v}}\|$ is its full-scale enclosing radius about that center (a scalar, distinct from the rotation matrix $\mathbf{R}_i$).
A sufficient condition for the bounding spheres of every pair to be separated by at least the internal margin $\hat d$ at the initial scale $s_{\min}$ is
\begin{equation}
  \label{eq:smin_safe}
  s_{\min}
  \;<\;
  \min_{i<j}
  \frac{\|\mathbf{c}_i-\mathbf{c}_j\|-\hat d}{r_i+r_j}.
\end{equation}
We recenter every mesh at the center of its normalization box (App.~\ref{ap:scene_gen}) and use that fixed point as the body reference center throughout.

We fix $s_{\min}{=}0.01$ in all reported experiments and use Eq.~\ref{eq:smin_safe} as a check. Any pair with
\[
  \|\mathbf{c}_i-\mathbf{c}_j\| \;<\; \hat d+s_{\min}(r_i+r_j)
\]
is perturbed apart deterministically along its pair axis to the distance $\hat d+s_{\min}(r_i+r_j)+\varepsilon_{\mathrm{sep}}$, with a separation padding $\varepsilon_{\mathrm{sep}}=10^{-6}$ (App.~\ref{ap:hyperparams}) and a fixed axis for exactly coincident centers (App.~\ref{ap:scene_gen}). The implementation visits the offending pairs once rather than iterating the sweep to a fixed point, and the resulting displacement is charged to the reported RMSD.

Near-coincident reference centers remain admissible but make the pair direction $(\mathbf{c}_j-\mathbf{c}_i)/\|\mathbf{c}_j-\mathbf{c}_i\|$ poorly conditioned; see our discussion on this limitation in Sec.~\ref{sec:conclusion}.

Progressive scaling is not an unconditional speedup: it pays off when the exact mesh-query time it saves by evaluating contacts only at scale events exceeds the extra continuation work (per-scale QPs, caching, tail refinement); the cost decomposition in App.~\ref{sec:theory_speed} (Eq.~\ref{eq:exact_speed_condition_app}) makes this precise, and explains why the one-shot Direct-QP ablation can be faster at small scale while leaving residual penetration (Tab.~\ref{tab:abl_scaling}).

\subsection{Per-Step Contact QP}
\label{sec:per_step_qp}

With progressive scaling in place, we formulate the local subproblem for the transition from $s_k$ to $s_{k+1}=s_k+\mathit{ds}$, where $\mathit{ds}>0$ is the scale increment. Let $\widetilde d_{ij}^{\,k}$ denote the scalar the deployed FCL~\cite{pan2012fcl} detector returns for the pair at the current scale --- a \emph{full detection}, i.e.\ a fresh geometry query: the positive boundary distance when FCL reports separation, the negative of the largest reported local contact depth when it returns collision contacts, and zero for an empty contact set. This piecewise score assembles the QP.

Sec.~\ref{sec:linearized_constraints} derives the linear constraint that one contact pair contributes, Sec.~\ref{sec:active_set} fixes which pairs contribute, Sec.~\ref{sec:objective} states the resulting QP, and Sec.~\ref{sec:conservativeness} records what the linearization does and does not guarantee.

\begin{figure}[t]
  \centering
  \includegraphics[width=0.62\columnwidth]{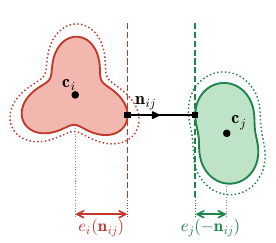}
  \caption{One-sided supports of a contact pair. With $\mathbf{n}_{ij}$ the contact direction from body $i$ to body $j$, $e_i(\mathbf{n}_{ij})$ is how far body $i$ reaches from its reference center toward body $j$, and $e_j(-\mathbf{n}_{ij})$ how far body $j$ reaches toward body $i$, i.e.\ along $-\mathbf{n}_{ij}$; inflation by $\mathit{ds}$ (dotted) closes the gap between the frozen witnesses (squares) by at most $\mathit{ds}\,[e_i(\mathbf{n}_{ij})+e_j(-\mathbf{n}_{ij})]$.}
  \Description{Two irregular 2D bodies with their reference centers, the contact normal between their closest points, and dashed support lines marking how far each body reaches toward the other; dotted outlines show the inflated bodies.}
  \label{fig:support_2d}
\end{figure}
\subsubsection{Linearized Contact Constraints}
\label{sec:linearized_constraints}
To formulate a convex subproblem, we linearize the contact condition w.r.t.\ the translational corrections $\Delta \mathbf{p}_i, \Delta \mathbf{p}_j \in \mathbb{R}^3$. For a contact pair $(i,j)$, let $\mathbf{n}_{ij}\in\mathbb{R}^3$ be the unit contact direction returned by the detector, pointing from body $i$ toward body $j$:
\begin{itemize}[leftmargin=*,nosep]
  \item \emph{Separated pair.} $\mathbf{n}_{ij}$ is the direction of the closest-point segment, $(\mathbf{w}_j-\mathbf{w}_i)/\|\mathbf{w}_j-\mathbf{w}_i\|$ for the closest points $\mathbf{w}_i$, $\mathbf{w}_j$ (the \emph{witness points}, or witnesses, of the pair); where that segment meets a face interior it agrees with the face normal, and at edges and vertices it is simply the closest-point direction.
  \item \emph{Penetrating pair.} $\mathbf{n}_{ij}$ is the separating direction the detector reports for the deepest contact, oriented toward body $j$.
  \item \emph{Frozen across the step.} $\mathbf{n}_{ij}$ is held fixed over the increment and until the next full detection (Sec.~\ref{sec:cache}), so it need not remain normal to either surface as the bodies move; it is a direction between the bodies, not a mesh normal, although we follow common usage in contact solvers and call it the contact normal.
\end{itemize}
Reversing the roles only reverses the arrow, $\mathbf{n}_{ji}=-\mathbf{n}_{ij}$, so $e_j(-\mathbf{n}_{ij})=e_j(\mathbf{n}_{ji})$ below, every pair quantity is symmetric in $(i,j)$, and each unordered pair carries one direction and one constraint row. For a normal $\mathbf{n}$, define the clamped one-sided support
\begin{equation*}
  e_i(\mathbf{n}) \;:=\; \max\Bigl\{0,\; \max_{\bar{\mathbf{v}}}\, \mathbf{n}^{\top} \mathbf{R}_i \bar{\mathbf{v}}\Bigr\},
\end{equation*}
where $\bar{\mathbf{v}}$ ranges over the body's vertex offsets from its scaling reference center, $\mathbf{R}_i$ is its rotation, and $e_j$ is defined in the same way for body $j$. Fig.~\ref{fig:support_2d} illustrates the construction. Geometrically, $e_i(\mathbf{n})$ is how far body $i$ reaches from its reference center in the direction $\mathbf{n}$: under uniform inflation by $\mathit{ds}$, no surface point of body $i$ advances along $\mathbf{n}$ by more than $\mathit{ds}\,e_i(\mathbf{n})$. Only the side facing the other body matters, which is why the pair uses $e_i(\mathbf{n}_{ij})$ and $e_j(-\mathbf{n}_{ij})$; the clamp at zero covers the case in which the reference center lies beyond the body's extent along $\mathbf{n}$ and only makes the bound more conservative. At a frozen witness pair and normal, the increment $\mathit{ds}$ closes the projected gap by at most $\mathit{ds}\,[e_i(\mathbf{n}_{ij})+e_j(-\mathbf{n}_{ij})]$, which gives the frozen-witness predictor
\begin{equation}
  \label{eq:linearized_gap}
  \ell_{ij}^{\,k+1}
  \;:=\;
  \widetilde d_{ij}^{\,k}
  + \mathbf{n}_{ij}^{\top}(\Delta \mathbf{p}_j - \Delta \mathbf{p}_i)
  - \mathit{ds}\,[e_i(\mathbf{n}_{ij}) + e_j(-\mathbf{n}_{ij})].
\end{equation}
This is a conservative projected-gap prediction at the frozen witness, not an identity for the post-step detector score: the supports upper-bound the true closure rate, and the witness features may switch during the step. The prediction bounds the projected gap between the frozen material witnesses from below whenever its intercept $\widetilde d_{ij}^{\,k}$ does not exceed that gap. This holds with equality for a fresh closest-point pair of a separated body pair; for penetrating detector rows it is a modeling convention, and in neither case is it a bound on the post-step mesh clearance. Enforcing $\ell_{ij}^{\,k+1} \ge \hat{d}$ (Eq.~\ref{eq:linearized_gap}) yields, for every active pair $(i,j)$ (the active set is defined next), a linear inequality on the displacements:
\begin{equation*}
  \mathbf{n}_{ij}^{\top} (\Delta \mathbf{p}_j - \Delta \mathbf{p}_i) \ge b_{ij},
\end{equation*}
where $b_{ij} := \hat{d} - \widetilde d_{ij}^{\,k} + \mathit{ds}\,[e_i(\mathbf{n}_{ij}) + e_j(-\mathbf{n}_{ij})]$ is the required separation gap plus the anticipated inflation closure.

\subsubsection{Active Set and Witness Conventions}
\label{sec:active_set}
With the closure coefficient $E_{ij}=e_i(\mathbf{n}_{ij})+e_j(-\mathbf{n}_{ij})$ of Eq.~\ref{eq:linearized_gap} available, the active-contact set of step $k$ is defined by the frozen-witness predictor rather than by an unknown post-step score,
\begin{equation*}
\begin{aligned}
  \mathcal{A}_k \;=\; \bigl\{(i,j) \;:\; {}& \widetilde d_{ij}^{\,k} - \mathit{ds}\,E_{ij} < \hat d \\
  & \text{and a valid witness/normal is available}\bigr\},
\end{aligned}
\end{equation*}
For each active pair we maintain \emph{one representative witness constraint per body pair}: when the pair is separated, the closest witness pair reported by the detector; when it penetrates, the deepest of the contacts the detector returns (up to $16$ per query). A single QP therefore does not jointly resolve all simultaneous primitive contacts of a non-convex pair; the remaining branches surface at the next full detection --- across scale steps, at refresh boundaries, and during tail refinement --- and are re-linearized then. Convergence is assessed empirically in Sec.~\ref{sec:exp}. The normal is fixed at the last full detection and held constant across the increment; it is refreshed at the next full detection (every $M$-th scale step, with $M$ the refresh interval of Sec.~\ref{sec:cache}). A pair whose collision query reports contact but returns an empty contact set is scored as touching by the evaluator and generates no QP row; it is reconsidered at the next full detection.

\subsubsection{Optimization Objective}
\label{sec:objective}
At each step we seek the minimum-norm translational correction that resolves the currently active constraints. Because every step solves a minimum-norm incremental update (not an absolute proximal term toward the initial pose), the cumulative displacement $\sum_k \Delta \mathbf{p}^{\star,k}$ is controlled only indirectly by the sequence of local corrections. For newly activated near-margin rows, the right-hand side $b_{ij}$ is typically of the order of the modeled inflation closure $\mathcal{O}(\mathit{ds}\,r_{\max})$, $r_{\max}=\max_i r_i$; already-penetrating or cached rows need not obey this scaling, and a poorly conditioned active set can amplify the minimum-norm solution, so a small $\mathit{ds}$ does not by itself guarantee a small correction. Empirically, the combination of per-step minimum-norm correction and progressive scaling keeps the total displacement RMSD nearly invariant over the tested scene sizes (Tab.~\ref{tab:scaling}). The translation-only QP is:
\begin{equation}
\label{eq:qp_translation}
\begin{aligned}
  \min_{\Delta \mathbf{p}^\star} \quad & \frac{1}{2} \sum_{i=1}^N \|\Delta \mathbf{p}_i^\star\|^2 \\
  \text{s.t.} \quad & \mathbf{n}_{ij}^{\top} (\Delta \mathbf{p}_j^\star - \Delta \mathbf{p}_i^\star) \ge b_{ij}, \quad \forall (i,j) \in \mathcal{A}_k.
\end{aligned}
\end{equation}
Eq.~\ref{eq:qp_translation} is the form used by every default translation-only \name\ run reported in this paper. Two variants are used elsewhere: the 6-DOF extension of App.~\ref{ap:6dof} adds an angular variable per body, weighted by a body length scale that bounds the surface displacement a rotation can cause, under a small-angle box constraint (the ablation in Tab.~\ref{tab:abl_rotation} shows a ${\sim}15\times$ wall-time overhead for a $13\%$ RMSD reduction on Kubric, so translation-only remains the default); and for the tabletop generated-asset application of Sec.~\ref{sec:exp_generated_asset}, the upright-on-plane variant of App.~\ref{sec:upright_plane} keeps every object on a common support plane and brings it upright at full scale, optimizing only in-plane translation and yaw.

\subsubsection{Conservativeness}
\label{sec:conservativeness}
The closure term is \emph{one-sided}. With $E_{ij}$ as above, the exact closure rate of the projected gap of a stable witness pair $(\bar{\mathbf{v}}_i^{\mathrm{w}},\bar{\mathbf{v}}_j^{\mathrm{w}})$ under uniform inflation is $c_{ij}=\mathbf{n}_{ij}^{\top}\mathbf{R}_i\bar{\mathbf{v}}_i^{\mathrm{w}}+(-\mathbf{n}_{ij})^{\top}\mathbf{R}_j\bar{\mathbf{v}}_j^{\mathrm{w}}$. Each witness lies inside a triangle and is therefore a convex combination of its vertices, so a linear functional over the body attains its maximum at a vertex and $c_{ij}\leq E_{ij}$; the clamp at zero can only enlarge the supports further. The pure-inflation closure term is therefore conservative at the frozen witness and normal, at the cost of a first-order slack $\mathit{ds}\,(E_{ij}-c_{ij})\geq0$ between the model and the geometry.

Separately, on a fixed $C^2$ distance branch, Taylor's theorem bounds the local remainder by $\tfrac12 L_{\mathrm{seg}}h_{ij}^2$ in the per-step relative motion $h_{ij}=\|\Delta \mathbf{p}_j^\star-\Delta \mathbf{p}_i^\star\|+\mathit{ds}\,(r_i+r_j)$, rather than in the initial penetration depth (Lemma~\ref{lem:taylor_residual}, App.~\ref{ap:theory_linearization}), with the branch written in the joint coordinates $\mathbf{z}=(\mathbf{p}_j-\mathbf{p}_i,\,(r_i+r_j)\,s)$, along which $\|\Delta \mathbf{z}\|\le h_{ij}$, and $L_{\mathrm{seg}}$ bounding its Hessian on the update segment. Notice that neither statement is a per-step feasibility guarantee: the first bounds the modeled closure at a frozen witness and normal, the second bounds the linearization error on one smooth branch, and witness switches within a step, drift of the cached normals, and the piecewise detector score lie outside both. Small increments keep the modeled per-step motion small, which is what the continuation is designed for; violations that survive a step are caught afterward --- re-detection at each refresh re-linearizes them (Sec.~\ref{sec:cache}), the tail pass corrects the full-scale state (Sec.~\ref{sec:tail}), and the shared evaluator scores every returned pose.

\subsection{Accelerated Path Following}
\label{sec:accelerated}

While the per-step quadratic program of Eq.~\ref{eq:qp_translation} is convex and locally robust, assembling its constraints requires evaluating the mesh-proximity score $\widetilde d_{ij}$ and the associated witness/contact normal $\mathbf{n}_{ij}$ for proximate mesh pairs. Computing these geometric primitives relies on BVH traversals and narrow-phase closest-point queries whose cost depends on the backend, triangle count, BVH quality, and concavity. In dense scenes the candidate graph can approach all $\mathcal{O}(N^2)$ body pairs, so these exact mesh queries dominate the CPU solver's wall time (App.~\ref{sec:exp_breakdown}). We introduce path-following accelerations that reduce full detections and solver work in the favorable regime of App.~\ref{sec:theory_speed}.

\subsubsection{Scale of Impact (SOI) Event Jumps}
\label{sec:soi}
The event rule models pure uniform inflation at fixed poses, under which each surface-vertex trajectory is affine in the scale parameter. In principle a candidate vertex--face or edge--edge primitive pair admits a polynomial-in-$s$ (degree at most three) continuous-collision condition~\cite{brochu2012ccd} whose smallest valid root above the current scale gives the exact scale of impact, with candidate roots subject to geometric filtering and persistent coplanarity handled separately. For robustness on watertight non-convex triangle meshes --- where the primitive-level cubic has many spurious roots from coplanar or degenerate triangles --- the deployed solver instead uses a bounding-sphere argument. The scaled meshes of pair $(i,j)$ are contained in the balls $B(\mathbf{c}_i,s r_i)$ and $B(\mathbf{c}_j,s r_j)$, so they cannot enter the $\hat d$ margin while $\|\mathbf{c}_j-\mathbf{c}_i\| - s\,(r_i+r_j) > \hat d$, which gives the conservative fixed-pose event scale
\begin{equation*}
  s_{\mathrm{soi}}^{\mathrm{cons}} \;=\; \frac{\|\mathbf{c}_j-\mathbf{c}_i\| - \hat d}{r_i+r_j}.
\end{equation*}
This is the per-pair form of Eq.~\ref{eq:smin_safe}, now applied at each step rather than once at $s_{\min}$; like Eq.~\ref{eq:smin_safe} it is a statement about the geometry, independent of the detector score. Two mechanisms use it. At every full detection, the culling test recomputes it from the current centers and discards any pair whose event scale lies beyond the upcoming step ($s+\mathit{ds}<0.9\,s_{\mathrm{soi}}^{\mathrm{cons}}$). Separately, ranges that are event-free under the precomputed fixed-pose model are skipped by extending the stride toward the earliest such scale, capped at $3\,\mathit{ds}_{\max}$ per jump, where $\mathit{ds}_{\max}{=}0.05$ is the schedule's base stride; after a contact-free step the schedule takes $2\,\mathit{ds}_{\max}$ instead, which can cross an event in the list. The bound underestimates the true scale of impact, so the schedule may visit a few scale points the exact primitive root would have skipped, and it is conservative only for pure inflation at fixed poses: the QP translations can bring a previously distant pair into contact. A translation-induced contact that persists to a refresh is found by the next full detection (the $0.9$ above is a $10\%$ event-scale safety factor, and the broad-phase axis-aligned bounding-box (AABB) margin is padded by $0.2\,\mathit{ds}\,\max(\mathrm{nf}_i,\mathrm{nf}_j)$, with $\mathrm{nf}_i$ the per-object normalization factor of App.~\ref{ap:scene_gen}; no additional constant is added in any reported run. Neither heuristic certifies the translation the upcoming QP will produce, since it is not yet known when detection runs) and by the tail-refinement pass. The component ablation in App.~\ref{sec:exp_soi} compares this event jump against fixed-stride schedules: relative to the tuned coarse schedule, it removes the residual contact-free steps at essentially unchanged wall time and RMSD.

\subsubsection{Analytical Distance Caching}
\label{sec:cache}
\label{sec:contact_detection}
Between full detections, the surface motion induced by uniform inflation and by the preceding QP correction is explicit, so we can propagate a frozen-witness gap prediction instead of re-querying the geometry. Let the integer $M \ge 1$ be the refresh interval. At a full detection the solver stores the measured score $\widetilde d_{ij}^{\,k}$, the witness/contact normal $\mathbf{n}_{ij}$, and the support coefficients $e_i, e_j$, and initializes $\breve d_{ij}^{\,k} := \widetilde d_{ij}^{\,k}$. For the next $M-1$ steps it bypasses BVH traversal and applies
\begin{equation}
\label{eq:cache_update}
  \breve d_{ij}^{\,k+1} = \breve d_{ij}^{\,k} + \mathbf{n}_{ij}^{\top} (\Delta \mathbf{p}_j^\star - \Delta \mathbf{p}_i^\star) - \mathit{ds}\,[e_i(\mathbf{n}_{ij}) + e_j(-\mathbf{n}_{ij})],
\end{equation}
using the scale increment of the preceding \emph{accepted} transition and the most recent QP correction $\Delta \mathbf{p}^\star$, which is carried across contact-free transitions; the $s_{\min}$ detection serves as refresh step zero, and no propagation is performed before the first transition. On these cached steps, $\breve d_{ij}^{\,k}$ takes the place of $\widetilde d_{ij}^{\,k}$ in the predictor of Eq.~\ref{eq:linearized_gap} and in the active-set test of Sec.~\ref{sec:active_set}; the QP itself is unchanged. The normal stays fixed at the value returned by the last full detection; we do not attempt to approximate a new one. Nothing in the frozen-witness model controls the normal, so we treat its drift as a heuristic: the normal moves as the bodies do, and across a refresh window that drift accumulates with the relative motion $h_{ij}$. Accordingly $\breve d_{ij}^{\,k}$ is a first-order prediction, not a newly measured score, and we flush the cache every $M$ steps to keep stale-witness drift bounded in practice; Fig.~\ref{fig:cache_rhythm} sketches this rhythm, with the residual $|\breve d_{ij}-\widetilde d_{ij}|$ growing inside each window and resetting at every refresh. App.~\ref{ap:cache_audit} reports the final outcome at each refresh interval. The default CPU-FCL configuration of \name\ uses $M{=}3$ --- main tables, cross-dataset tables, runtime breakdown, and component ablations alike. The cache ablation of Tab.~\ref{tab:abl_cache} sweeps $M$ around this default, and the upright variant (App.~\ref{sec:upright_plane}) uses no cache. The GPU-native \name-Warp pipeline (Sec.~\ref{sec:exp_setup}) re-detects at every scale step and uses no cache.

\subsubsection{Active-Set Warm-Starting}
\label{sec:schedule}
The path-following nature of our method implies that the active-contact set $\mathcal{A}_k$ typically remains constant across several consecutive scale steps. Within a cache interval with an unchanged row set and frozen normals, the QP matrix is fixed and only the right-hand side moves; at a fixed accepted state, the corresponding trial QP is piecewise affine in its right-hand-side parameter as the binding set changes (App.~\ref{ap:proof_piecewise}). This local property motivates a sensitivity-based warm start --- it is not a statement that the accepted solution along the deployed trajectory is piecewise affine in the absolute scale $s$. We exploit it by warm-starting the primal variables of the solver (e.g., OSQP) with $\mathit{ds}$ times a regularized sensitivity of the previous solve when the active-body set is unchanged (halved when it overlaps by more than half), and otherwise with the cached per-body increments. Seeding across consecutive scale steps helps keep the solver's inner iteration count roughly stable along the path.

\subsection{Tail Refinement}
\label{sec:tail}

The frozen-witness model of Sec.~\ref{sec:per_step_qp} describes one exact-query step on a stable local branch. The deployed path is not that: cached predictions and witness changes can leave a small residual by the time the continuation reaches full scale. We therefore run a bounded tail-refinement phase --- scale is locked at $s{=}1$ (setting $\mathit{ds}=0$) and we iterate correction QPs, each followed by an FCL re-evaluation that reports the pairs the detector currently scores as penetrating. The phase terminates when the detector reports zero penetrating pairs, after $K_{\mathrm{tail}}{=}20$ correction QPs, when neither the penetrating-pair count nor the deepest violation has improved for $3$ consecutive iterations, or when a correction QP returns no accepted solution. In the main comparison, scaling, and dataset benchmarks the detector reaches zero within that budget; residual pairs appear in the densest stress tests (App.~\ref{sec:exp_density}, App.~\ref{sec:exp_sphere}) and in the constrained-confinement study, where the monotone continuation fails to find a feasible path within the prescribed constraints. Zero reported penetration is thus an empirical outcome of a bounded solver, not a theorem.

\subsubsection{Full-Active-Set Correction}
A local-only residual correction would isolate the penetrating pairs and push them apart. However, in dense configurations, resolving one pair in isolation frequently moves one of its bodies into another neighboring object, creating a cascading local conflict that prevents convergence.

The tail QP therefore includes both the penetrating pairs and the separated pairs in their neighborhood. We define the tail active set as $\mathcal{A}_{\mathrm{tail}}=\{(i,j):\widetilde d_{ij} < \hat{d}\ \text{and a valid witness/normal is available}\}$, enumerated over the pairs whose bounding boxes overlap at full scale; an empty-contact zero-score pair is re-evaluated at subsequent exact detections but generates no tail-QP row, consistent with Sec.~\ref{sec:per_step_qp}. This set contains both the currently penetrating pairs ($\widetilde d_{ij} < 0$) and the separated near-contact pairs ($0 \le \widetilde d_{ij} < \hat{d}$). By imposing the non-penetration constraints collectively across $\mathcal{A}_{\mathrm{tail}}$, the QP anticipates the domino effect: it computes a unified displacement $\Delta \mathbf{p}^\star$ that resolves the penetrations while accounting for separated near-contact neighbors that could otherwise be squeezed into overlap.

\subsubsection{The Correction QP and Step Damping}
With $\mathit{ds}=0$, the anticipated inflation closure vanishes. The kinematic constraint for the QP simplifies to bridging the current gap directly:
\begin{equation*}
  \mathbf{n}_{ij}^{\top} (\Delta \mathbf{p}_j^\star - \Delta \mathbf{p}_i^\star) \ge \hat{d} - \widetilde d_{ij}, \quad \forall (i, j) \in \mathcal{A}_{\mathrm{tail}}.
\end{equation*}
While the optimization objective remains identical to the primary path-following QP (minimizing $\|\Delta \mathbf{p}^\star\|^2$), the raw output displacement $\Delta \mathbf{p}^\star$ might occasionally exceed the valid linearization radius if the local residual is exceptionally stiff.

As a heuristic safeguard on large corrections (the margin $\hat d$ is not a certified linearization radius), we apply a damping factor $\alpha \in (0, 1]$ to the QP update before modifying the body states: $\mathbf{p}_i \leftarrow \mathbf{p}_i + \alpha \Delta \mathbf{p}_i^\star$. The parameter $\alpha$ is dynamically determined as:
\begin{equation*}
  \alpha = \min \left( 1, \, \frac{\hat{d}}{\max_i \|\Delta \mathbf{p}_i^\star\|} \right).
\end{equation*}
This step clipping ensures that no single body is displaced by more than the internal margin $\hat{d}$ in one iteration, which empirically keeps updates in the regime where the first-order model is reliable. When $\alpha<1$ the applied update $\alpha\Delta \mathbf{p}^\star$ generally does not satisfy the undamped rows above, so the clipped step is an iterative correction rather than a one-step enforcement of the full margin; the next exact re-detection forms a new QP. We set $\alpha=1$ when $\max_i\|\Delta \mathbf{p}_i^\star\|=0$. Algorithm~\ref{alg:s4r} summarizes the deployed solver end to end.

\begin{algorithm}[t]
\SetKwInOut{Input}{input}\SetKwInOut{Output}{output}
\SetKwComment{Where}{$\triangleright$\ }{}
\Input{meshes and poses; $\hat d,s_{\min},\mathit{ds}_{\max},\mathit{ds}_{\min}$; $M,K_{\max},K_{\mathrm{tail}}$; optional walls $\mathcal W$}
\Output{poses $\mathbf p$; evaluator report $\mathcal R$; status $\sigma$}
normalize meshes, compute $r_i$, and perturb the pairs that violate Eq.~\ref{eq:smin_safe}\;
$s\leftarrow s_{\min}$; $\mathit{ds}_{\mathrm{cap}}\leftarrow+\infty$; $a\leftarrow0$; build BVHs and the event list\;
full detection; $\breve d_{ij}\leftarrow\widetilde d_{ij}$\;
$(\Delta\mathbf p_{\mathrm{prev}},\mathit{ds}_{\mathrm{prev}})\leftarrow(\mathbf0,0)$\;
\While{$s<1$ \textbf{and} $a<K_{\max}$}{
  $a\leftarrow a+1$\;
  $\mathit{ds}_{\mathrm{event}}\leftarrow$ event-aware increment\Where*[r]{Sec.~\ref{sec:soi}}
  $\mathit{ds}\leftarrow\min\{\mathit{ds}_{\mathrm{event}},\mathit{ds}_{\mathrm{cap}},1-s\}$\;
  \eIf{an exact refresh is scheduled or forced}{
    full detection; $\breve d_{ij}\leftarrow\widetilde d_{ij}$; $\bar d_{ij}^{\,k}\leftarrow\widetilde d_{ij}^{\,k}$\;
  }{
    $\breve d_{ij}^{\,k}\leftarrow\textsc{Predict}(\breve d_{ij}^{\,k-1},\Delta\mathbf p_{\mathrm{prev}},\mathit{ds}_{\mathrm{prev}})$\Where*[r]{Eq.~\ref{eq:cache_update}}
    $\bar d_{ij}^{\,k}\leftarrow\breve d_{ij}^{\,k}$\;
  }
  form $\mathcal A_k$ from $\bar d_{ij}^{\,k}$\;
  $(\Delta\mathbf p^\star,\sigma)\leftarrow\textsc{SolveQP}(\mathcal A_k,\mathcal W)$\Where*[r]{Eq.~\ref{eq:qp_translation}}
  \lIf{$\sigma=\textsc{container-infeasible}$}{\Return $(\mathbf p,\varnothing,\sigma)$}
  \If{$\sigma\neq\textsc{solved}$}{
    $\mathit{ds}_{\mathrm{cap}}\leftarrow\mathit{ds}/2$\;
    \lIf{$\mathit{ds}_{\mathrm{cap}}<\mathit{ds}_{\min}$}{\Return $(\mathbf p,\varnothing,\textsc{qp-failure})$}
    mark refresh due; \textbf{continue}\;
  }
  $\mathbf p\leftarrow\mathbf p+\Delta\mathbf p^\star$; $s\leftarrow s+\mathit{ds}$\;
  $(\Delta\mathbf p_{\mathrm{prev}},\mathit{ds}_{\mathrm{prev}})\leftarrow(\Delta\mathbf p^\star,\mathit{ds})$; $\mathit{ds}_{\mathrm{cap}}\leftarrow+\infty$\;
}
\lIf{$s<1$}{evaluate $(\mathbf p,s)$ and \Return $(\mathbf p,\mathcal R,\textsc{incomplete})$}
run tail refinement at $s=1$ to obtain $\sigma$\Where*[r]{Sec.~\ref{sec:tail}}
evaluate $(\mathbf p,1)$ to obtain $\mathcal R$\Where*[r]{Alg.~\ref{alg:evaluator}}
\Return $(\mathbf p,\mathcal R,\sigma)$\;
\caption{\name\ (translation-only default). $\bar d_{ij}^{\,k}$ is the score that forms the rows of step $k$: measured at a refresh, predicted otherwise.\label{alg:s4r}}
\end{algorithm}

\section{Experiments}
\label{sec:exp}

We evaluate \name on rigid-body interpenetration resolution across scene sizes from $40$ to $30{,}000$ bodies (the full multi-baseline tables run through $5000$; Fig.~\ref{fig:scaling} plots the two largest sizes) and geometric complexity (convex, mildly non-convex, highly non-convex), against a broad set of baselines.
The evaluation is designed to isolate the impact of each technical contribution of \name---progressive scaling, QP displacement correction, scale of impact (SOI) analytical event handling, collision caching, adaptive scheduling, and contact sparsity---through controlled ablations.
All methods are scored by the \emph{same} mesh-level evaluator so the reported penetration counts and RMSDs are directly comparable.
Before detailing the protocol, we show in Fig.~\ref{fig:dataset_complexity} samples from the three source pools, illustrating the concavities, thin handles, perforations, and multi-component structures that appear before random placement and scaling.

\subsection{Experimental Setup}
\label{sec:exp_setup}

\begin{figure}[t]
  \centering
  \includegraphics[width=\columnwidth]{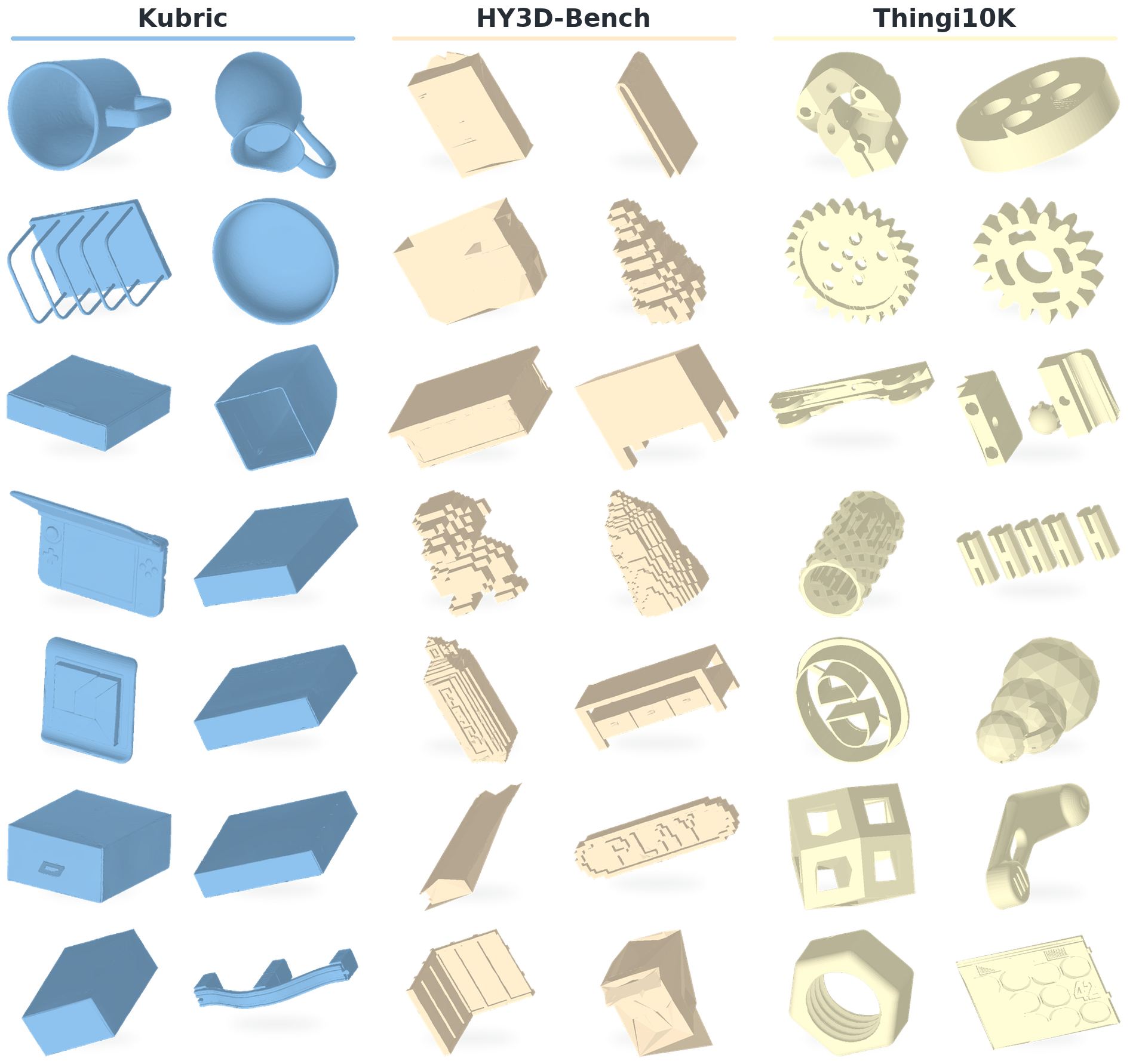}
  \caption{Mesh samples from the three source pools. Kubric, HY3D-Bench, and Thingi10K each contribute fourteen samples arranged as two columns in the single-column panel. The HY3D-Bench column shows assets of the source collection; the benchmark scenes draw on the filtered templates listed in App.~\ref{ap:scene_gen}.}
  \Description{A seven-row by six-column gallery of mesh assets on a white background. The first two columns show Kubric objects, the middle two columns show HY3D-Bench objects, and the last two columns show Thingi10K objects.}
  \label{fig:dataset_complexity}
\end{figure}

\paragraph{Scenes and metrics.}
We use three pools of varying geometric difficulty: \emph{Kubric}, $40$ watertight household meshes from Google Scanned Objects~\cite{downs2022gso} with mild-to-moderate non-convexity (area-based concavity $\kappa=1-\operatorname{area}(\text{hull})/\operatorname{area}(\text{mesh})$ up to $0.54$, App.~\ref{sec:exp_objaverse}), the collection the Kubric generator~\cite{greff2022kubric} draws on; \emph{HY3D-Bench}~\cite{tencent2026hy3dbench}, generative-3D meshes ($13$ templates pass the benchmark filters in the decimated pool and $12$ in the full-resolution pool; App.~\ref{ap:scene_gen}); and \emph{Thingi10K}~\cite{zhou2016thingi10k}, $430$ artist-authored printable meshes. Each scene places $N$ objects at random positions and orientations in a box that scales as $(N/40)^{1/3}$, so contact density stays roughly constant with $N$; three seeds ($42$, $123$, $456$) are used unless a study states otherwise. Initial penetrating pairs grow with $N$ ($31$, $772$, and $4026$ on average at $N{=}40$, $1000$, and $5000$), and every method starts from the same state. Whatever a solver returns, the same triangle-mesh evaluator (App.~\ref{ap:evaluator}) scores it: we report the number of negative-score body pairs under the shared evaluator (\emph{Pen.}), the maximum penetration depth (\emph{maxPen}) in the stress tests where residual penetration occurs, the reference-center RMSD from the initial layout, and wall-clock time. AVBD and ISIR also rotate bodies, so their center RMSD can understate surface motion. App.~\ref{ap:benchmark_limits} gives pool statistics, spawn parameters, and metric definitions.

\paragraph{Scoring protocol.}
Every method is scored against the same final penetration threshold $\tau{=}0$: the online evaluator counts the pairs with $\widetilde d_{ij}<0$ under the piecewise FCL convention, with no tolerance band, and a separate offline containment audit of the retained final states finds no missed nested watertight pair (App.~\ref{ap:evaluator}). A solver's internal margin is a solver-shaping parameter rather than the reported target: \textsc{Soft-Penalty} terminates at $\tau{=}0$, and \textsc{Drake} retains its stated $10^{-4}$ minimum-distance lower bound on its convex hulls. Every returned pose is then scored by the same strict $\tau{=}0$ evaluator. Solvers that work on proxies are scored on the meshes their poses imply.  Each cell is attempted on all three seeds within $1800$\,s per seed, with the exceptions noted for \textsc{Drake} and \textsc{AVBD-OBB}; the $N{=}30000$ cells of Fig.~\ref{fig:scaling} complete within it; cells that do not finish carry their completion count. App.~\ref{ap:benchmark_limits} has the full protocol.

\paragraph{Baselines.}
We compare against six baselines spanning the dominant families: \textbf{AVBD}~\cite{Giles2025AVBD} (official 3D implementation, \textsc{AVBD-OBB}, with OBB--OBB SAT collision); \textbf{Rigid-ISIR}~\cite{Jang2025ISIR} (ISIR in tables), our rigid-pose adaptation of a vertex-level self-intersection repair; \textbf{Drake}~\cite{drake}, inverse kinematics with minimum-distance lower bounds on convex hulls, in two solver configurations (Ipopt~\cite{wachter2006ipopt} and the SNOPT build bundled with Drake~\cite{gill2005snopt}); \textbf{Global QP/LCP}, an iterated global linearized contact QP solved by OSQP~\cite{stellato2020osqp}; \textbf{PD-PGS}, our own projected Gauss--Seidel sweeps over the overlap graph; and \textbf{Soft-Penalty} (Soft-Pen.\ in tables), our static penalty-continuation baseline after TrajOpt's collision-penalty continuation~\cite{schulman2014trajopt}, with a squared-hinge penetration penalty minimized by L-BFGS-B in our implementation. App.~\ref{ap:benchmark_limits} gives each one's geometry handling, adaptation, and tuning.
\name, QP/LCP and Soft-Penalty use the common internal clearance $\dhat{=}0.02$, PD-PGS its default clearance of $10^{-3}$, while \textsc{Drake} enforces a $10^{-4}$ minimum-distance lower bound on its convex hulls, whereas the shared evaluator tests $\tau{=}0$; methods with native stopping criteria (\textsc{AVBD-OBB}, \textsc{ISIR}) retain their reference settings and are evaluated, like every method, by the common final mesh-level evaluator, whose strict sign test is independent of any solver margin.

\paragraph{Implementation.}
\name is implemented in Python with NumPy~\cite{harris2020array} and OSQP~\cite{stellato2020osqp} for the inner quadratic programs.
Pairwise mesh proximity in our implementation is handled by FCL~\cite{pan2012fcl}.
All reported CPU \name\ numbers use the same FCL-backed oracle, except the early stages of the implementation-stack ablation (Tab.~\ref{tab:perf_stack}), which use the Python backend that FCL replaced; \name-Warp uses a separate GPU-native oracle. Both are checked by the shared evaluator of App.~\ref{ap:evaluator}.
We report all baselines in two hardware tiers and compare baseline wall times only within a tier; the two \name\ implementations are compared with each other only as an implementation diagnostic.
The CPU tier (one socket, 14 physical cores, of a dual-socket Intel Xeon E5-2680\,v4 @ 2.4\,GHz, 377\,GiB RAM; each method runs sequentially on the same 14 cores, so within-tier wall times are directly comparable) holds \textsc{Soft-Penalty}, \textsc{PD-PGS}, \textsc{QP/LCP}, \textsc{Drake (hull)}, \textsc{AVBD-OBB} (the published C++ demo) and our reference \name; the GPU tier (one RTX 2080 Ti plus 14 physical host cores as host, likewise sequential) holds \textsc{ISIR} (PyTorch CUDA) and \name-Warp (\name's full solve on device --- the per-step min-norm contact QP runs as a CUDA-graph-fused accelerated projected gradient descent on the dual (dual-APGD) on the GPU, paired with an FCL-free NVIDIA Warp~\cite{warp} contact oracle).
The mesh-level CPU methods (\textsc{Soft-Penalty}, \textsc{PD-PGS}, \textsc{QP/LCP}, \name) share an FCL-backed mesh closest-point query at solve time. \name-Warp uses an FCL-free NVIDIA Warp winding-number signed-distance oracle with a deterministic packed-argmin contact-normal reduction, and all rows are scored by the same final evaluator at test time. The GPU QPs are solved in float32 by a dual-APGD with warm starts and a fixed 60-iteration CUDA-graph schedule per QP in both the main loop and the tail, with the step size set from a power-iteration estimate of the dual Lipschitz constant inflated by $15.5\%$ (a step with at most four rows uses the Gershgorin bound with a $5\%$ margin) rather than from a certified bound; post-solve residuals are recorded as diagnostics rather than used as an acceptance test, so the CPU OSQP tolerances do not describe the GPU QP accuracy. The oracle samples the vertices and five interior points per edge of each body against the other body's mesh.

\subsection{Main Comparison}

\makeatletter\typeout{FLOATCFG topnumber=\the\c@topnumber\space totalnumber=\the\c@totalnumber\space topfraction=\topfraction\space textfraction=\textfraction\space floatpagefraction=\floatpagefraction\space colwidth=\the\columnwidth\space textheight=\the\textheight}\makeatother

\label{sec:exp_main}

Tab.~\ref{tab:main} reports the full comparison at $N{=}40$ and $N{=}100$ on Kubric, averaged over three random seeds unless noted otherwise.
All methods are evaluated with the unified mesh-level metric of Sec.~\ref{sec:exp_setup}; \name, \textsc{QP/LCP}, \textsc{PD-PGS} and \textsc{Soft-Penalty} share the same FCL-backed mesh oracle on the host during solving, while \name-Warp uses a separate GPU-native oracle and \textsc{AVBD-OBB}, \textsc{ISIR} and \textsc{Drake} keep their native geometry representations; every row is scored by the same final evaluator.
\textsc{AVBD-OBB} and \textsc{ISIR} retain their native collision backends, and \textsc{Drake} operates on convex hulls of the meshes.

\begin{table}[!tbp]
 \caption{Main comparison on Kubric ($N{=}40,100$; 3-seed means). \name\ reaches zero evaluator-reported penetration with low displacement and the shortest wall time in both hardware tiers; symbols mark budget or reporting exceptions (table note).}
 \label{tab:main}
 \centering
 \small
 \setlength{\tabcolsep}{6pt}
 \resizebox{\columnwidth}{!}{%
 \begin{tabular}{@{}lcccccc@{}}
 \toprule
 & \multicolumn{3}{c}{$N{=}40$} & \multicolumn{3}{c}{$N{=}100$} \\
 \cmidrule(lr){2-4}\cmidrule(lr){5-7}
 Method & Pen. & RMSD & Time & Pen. & RMSD & Time \\
 \midrule
 \multicolumn{7}{@{}l}{\emph{CPU methods (one 14-core socket)}} \\
 AVBD-OBB (tuned)$^{\natural}$ & \textbf{0} & 1.33 & 2.2 s & \textbf{0} & 2.34 & 8.0 s \\
 AVBD-OBB (official)$^{\flat}$ & \textbf{0}$^{\,2/3}$ & 4.46 & 1.2 s & \textbf{0} & 1.15 & 3.3 s \\
 Drake-Ipopt (hull)$^\dagger$ & \textbf{0} & \textbf{0.022} & 597 s & \textbf{0}$^{\,1/3}$ & \textbf{0.025}$^{\,1/3}$ & 6518 s$^{\,1/3}$ \\
 Drake-SNOPT (hull)$^\dagger$ & \textbf{0} & \underline{0.023} & 22.6 s & \textbf{0} & \textbf{0.025} & 553 s \\
 QP/LCP & \textbf{0} & 0.038 & \underline{0.5 s} & \textbf{0} & 0.039 & \underline{1.7 s} \\
 PD-PGS & \textbf{0} & 0.028 & 1.1 s & \textbf{0} & \underline{0.027} & 3.4 s \\
 Soft-Pen. & \textbf{0} & 0.032 & 19.6 s & \textbf{0} & 0.034 & 114 s \\
 \textbf{\name (CPU, ours)} & \textbf{0} & 0.036 & \textbf{0.2 s} & \textbf{0} & 0.035 & \textbf{0.4 s} \\
 \midrule
 \multicolumn{7}{@{}l}{\emph{GPU methods (1$\times$ RTX 2080 Ti)}} \\
 ISIR & \textbf{0} & 0.175 & 14 s & \textbf{0} & 0.175 & 14 s \\
 \textbf{\name-Warp (ours)} & \textbf{0} & \textbf{0.035} & \textbf{1.5 s} & \textbf{0} & \textbf{0.035} & \textbf{1.6 s} \\
 \bottomrule
 \end{tabular}%
 }
 \par
 \raggedright\footnotesize{Pen.\ = negative-score body pairs under the shared piecewise FCL evaluator; RMSD = reference-center displacement. Bold names mark our variants; within a tier, bold marks the column best and underline the runner-up, at the displayed precision (ties at pen${=}0$ are all bold, so no runner-up is marked there; the GPU tier has two methods, so only the best is marked). Superscripts: $\dagger$ \textsc{Drake} under a tightened bound and a $7200$\,s budget --- with Ipopt, $1/3$ seeds finish at $N{=}100$; \textsc{Drake-SNOPT} solves the same program with the SNOPT build bundled in Drake and finishes $3/3$; $\natural$ tuned and $\flat$ released \textsc{AVBD-OBB} configurations, the latter without a static-repair termination. App.~\ref{ap:benchmark_limits} gives each configuration.}
\end{table}

\emph{CPU block.} The optimization-based methods and \name\ reach zero residual penetration at $N{=}40$ with RMSD between $0.022$ and $0.038$. \name\ is the fastest method in the block ($0.2$\,s versus $0.5$\,s for QP/LCP, $1.1$\,s for PD-PGS, $19.6$\,s for Soft-Pen., and $22.6$--$597$\,s for the two Drake configurations) with RMSD about $0.014$ above the smallest (Drake's $0.022$--$0.023$, obtained at roughly $130$ and $3500$ times \name's wall time with SNOPT and Ipopt, respectively). The absolute gap widens with $N$ (Sec.~\ref{sec:exp_scaling}): QP/LCP and PD-PGS take $3.2\times$/$5.9\times$ \name's wall time at $N{=}5000$ ($4.6\times$/$8.6\times$ at $N{=}2000$), Drake-Ipopt completes only one of three seeds within the $7200$\,s budget already at $N{=}100$ while Drake-SNOPT completes all three in $553$\,s --- still roughly $1300\times$ \name's $0.4$\,s --- and Soft-Penalty --- evaluated at the common $\tau{=}0$ target --- completes $N{=}500$ in $1409$\,s and times out at $N{=}1000$. \textsc{AVBD-OBB} also reaches pen${=}0$ (given a convergence-binding step budget, as for \textsc{ISIR}/\textsc{Drake}), but sits apart from the optimization baselines on displacement: RMSD $1.33$ at $N{=}40$ ($37\times$ \name's), consistent with its momentum-carrying penalty update applied at static initialization (App.~\ref{ap:benchmark_limits}).

\emph{GPU block.} \name-Warp matches its CPU sibling: pen${=}0$ at both $N$, RMSD within $0.001$ of the CPU run, with the entire solve (QP and contact oracle) moved to the GPU. Its wall time at these sizes is almost entirely startup: bringing up the Warp runtime --- kernel loading and device allocation --- costs a fixed ${\sim}1.0$\,s per run, while the solve itself takes $0.20$\,s at $N{=}40$ and $0.28$\,s at $N{=}100$. That fixed cost is why the GPU port trails its CPU sibling on small scenes and overtakes it once the scene is large enough to amortize it (Sec.~\ref{sec:exp_scaling}). Its GPU peer is \textsc{ISIR}: with its iteration cap raised and framework initialization charged on both sides, ISIR reaches pen${=}0$ at both $N$ in ${\sim}14$\,s, but with ${\sim}5\times$ larger displacement ($0.17$--$0.18$ vs.\ $0.035$--$0.036$). \name-Warp is faster ($1.5$--$1.6$\,s), keeps the smaller displacement, and at scale runs far beyond ISIR's practical reach (Sec.~\ref{sec:exp_scaling}).

\subsection{Scaling Study}
\label{sec:exp_scaling}

To study how each method scales with scene size, we run \name, ISIR, QP/LCP, and PD-PGS on Kubric across seven sizes, from $N{=}40$ to $N{=}5000$ in Tab.~\ref{tab:scaling}, with Fig.~\ref{fig:scaling} additionally plotting the two largest measured sizes.
Baselines are reported until they time out or run out of device memory.

\begin{table*}[!tbp]
 \caption{Scaling on Kubric ($N{=}40$--$5000$; 3-seed means; setup$+$solve time, App.~\ref{ap:benchmark_limits}). Within each hardware tier \name\ is the fastest at every $N$, with essentially flat RMSD.}
 \label{tab:scaling}
 \centering
 \small
 \setlength{\tabcolsep}{14pt}
 \resizebox{\textwidth}{!}{%
 \begin{tabular}{@{}lcccccccccc@{}}
 \toprule
 & \multicolumn{6}{c}{\textbf{CPU} (single socket)} & \multicolumn{4}{c}{\textbf{GPU} (1$\times$ RTX 2080 Ti)} \\
 \cmidrule(lr){2-7}\cmidrule(lr){8-11}
 & \multicolumn{2}{c}{\textbf{\name} (ours)} & \multicolumn{2}{c}{QP/LCP} & \multicolumn{2}{c}{PD-PGS}
 & \multicolumn{2}{c}{\textbf{\name-Warp} (ours)} & \multicolumn{2}{c}{ISIR} \\
 \cmidrule(lr){2-3}\cmidrule(lr){4-5}\cmidrule(lr){6-7}\cmidrule(lr){8-9}\cmidrule(lr){10-11}
 $N$ & RMSD & Time & RMSD & Time & RMSD & Time & RMSD & Time & RMSD & Time \\
 \midrule
40 & \underline{0.036} & \textbf{0.2 s} & 0.038 & \underline{0.5 s} & \textbf{0.028} & 1.1 s & \textbf{0.035} & \textbf{1.5 s} & 0.175 & 14 s \\
100 & \underline{0.035} & \textbf{0.4 s} & 0.039 & \underline{1.7 s} & \textbf{0.027} & 3.4 s & \textbf{0.035} & \textbf{1.6 s} & 0.175 & 14 s \\
200 & \underline{0.035} & \textbf{1.0 s} & 0.039 & \underline{3.6 s} & \textbf{0.026} & 9.2 s & \textbf{0.034} & \textbf{1.7 s} & 0.241 & 51 s \\
500 & \underline{0.034} & \textbf{3.0 s} & 0.040 & \underline{11.7 s} & \textbf{0.024} & 36.5 s & \textbf{0.033} & \textbf{2.1 s} & 0.298 & 135 s \\
 1000 & \underline{0.035} & \textbf{7.7 s} & 0.040 & \underline{29.8 s} & \textbf{0.025} & 69.2 s & \textbf{0.034} & \textbf{2.9 s} & 0.356 & 427 s \\
 2000 & \underline{0.037} & \textbf{17.5 s} & 0.042 & \underline{79.8 s} & \textbf{0.025} & 150 s & \textbf{0.036} & \textbf{4.6 s} & 0.394 & 1305 s \\
 5000 & \underline{0.037} & \textbf{68.1 s} & 0.042 & \underline{220 s} & \textbf{0.025} & 404 s & \textbf{0.036} & \textbf{9.7 s} & 0.445$^\ddagger$ & 1775 s$^\ddagger$ \\
 \bottomrule
 \end{tabular}%
 }
 \par
 \raggedright\footnotesize{Bold names mark our variants; within a tier, bold marks the best RMSD and time per row and underline the runner-up (the GPU tier has two methods, so only the best is marked). Pen.\ is omitted: every cell shown reports pen${=}0$ except the marked one. $\ddagger$ \textsc{ISIR} stops at its budget ($39$ residual pairs on average). Budgets are in Sec.~\ref{sec:exp_setup}.}
\end{table*}

Across the sweep, \name\ (CPU) keeps RMSD between $0.034$ and $0.037$ from $N{=}40$ to $N{=}5000$ and \name-Warp stays in $[0.033,0.036]$ over the same range, a $125\times$ span in scene size; every \name/\name-Warp cell resolves to zero residual penetration. Each body moves in response to the shallow-contact constraints activated along the scale path, the initial pairwise perturbation, and the tail corrections, consistent with the per-step minimum-norm objective and the sparse active-contact graph observed along the path.

Under the unified setup$+$solve timing, \name\ is also the fastest CPU solver at every $N$: $3.9\times$ faster than QP/LCP ($7.7$\,s vs.\ $29.8$\,s) and $9.0\times$ faster than PD-PGS ($69.2$\,s) at $N{=}1000$, and $3.2\times$/$5.9\times$ faster at $N{=}5000$ ($68.1$\,s vs.\ $220$\,s and $404$\,s). On the GPU, \name-Warp's small-$N$ totals are dominated by one-time initialization (${\sim}1$\,s of runtime startup and allocation; the solve itself is $0.2$--$3.1$\,s through $N{=}2000$), while at $N{=}5000$ the total is $9.7$\,s, of which the solve is $8.1$\,s. It is faster than its GPU peer ISIR at every completed $N$ ($4.6$\,s vs.\ $1305$\,s at $N{=}2000$) and reaches pen${=}0$ at $N{=}5000$, where ISIR stops at its budget with $39$ residual pairs. ISIR reaches displacements of $0.17$--$0.39$ that grow with $N$ and reaches pen${=}0$ through $N{=}2000$, with its iteration cap and collision buffers non-binding on those completed cells; at $N{=}5000$ its budget, not memory, is the binding limit; AVBD's momentum-driven displacement is reported with \textsc{AVBD-OBB} in Tab.~\ref{tab:main}.

Fig.~\ref{fig:scaling} plots wall time on a log--log axis with reference slopes $\propto N$ and $\propto N^{2}$ overlaid. Empirically, every curve grows super-linearly at large $N$: the local log--log slope of \name\ (CPU) rises from about $1$ below $N{=}500$ to $1.2$--$1.6$ between $N{=}500$ and $N{=}10000$, while QP/LCP and PD-PGS stay between $0.9$ and $1.5$ over the same range, so the advantage of \name\ is a lower absolute cost at every measured $N$ rather than a lower growth order; \name-Warp is startup-bound below $N{\approx}1000$ and its slope rises to $1.75$ between $N{=}10000$ and $N{=}30000$. We make no asymptotic claim: the candidate-generation pass is a vectorized sweep over all $N(N{-}1)/2$ pairs (worst case $\Theta(N^2)$, with small constants), and it is the narrow-phase mesh-query count --- which follows the active-contact set at fixed density --- that dominates the measured wall time at these scales (App.~\ref{sec:exp_breakdown}).

\begin{figure}[t]
  \centering
  \includegraphics[width=\columnwidth]{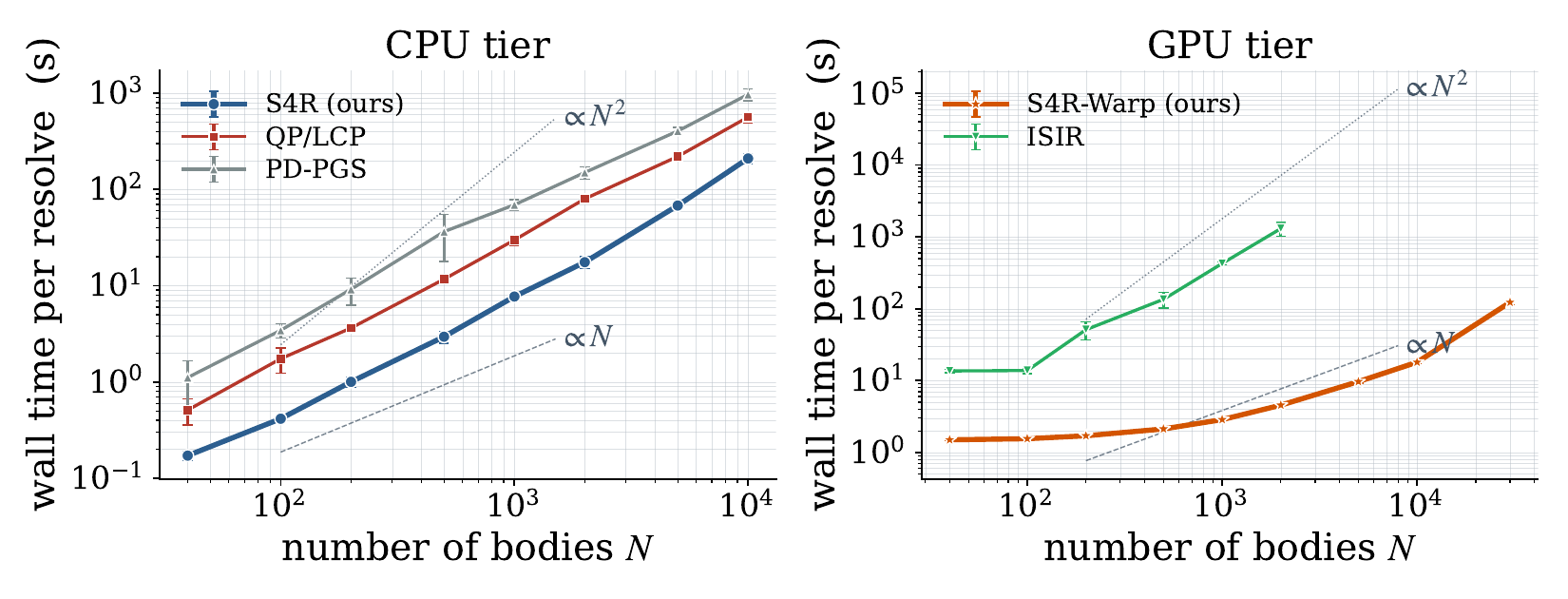}
  \caption{Wall time vs.\ scene size $N$ on Kubric (log--log, 3-seed mean $\pm$ s.d.; CPU tier on one 14-core Xeon E5-2680\,v4 socket, GPU tier on one RTX 2080 Ti). Solid lines connect the cells that reach pen${=}0$, of which Tab.~\ref{tab:scaling} lists the $N{\leq}5000$ subset; reference slopes $\propto N$ (dashed) and $\propto N^{2}$ (dotted) are guides to the eye. All curves grow super-linearly at large $N$ over $N{=}40$--$30{,}000$ (a $750\times$ span); \name\ has the lowest wall time in its tier at every measured $N$.}
  \Description{A log-log scaling plot of wall time versus body count. All wall times grow super-linearly at large body counts; the method's curve lies below the baselines of its tier at every measured size.}
  \label{fig:scaling}
\end{figure}

\subsection{Cross-Dataset Breadth at Scale}
\label{sec:exp_datasets}

The scaling study in Sec.~\ref{sec:exp_scaling} sweeps $N$ on a curated 40-object Kubric pool; to test that the wall-time and penetration advantages of \name hold on a broader mesh population, we run the CPU baselines (QP/LCP, PD-PGS, AVBD-OBB) and the GPU-tier Rigid-ISIR, alongside \name, on two larger mesh datasets:
\emph{HY3D-Bench}~\cite{tencent2026hy3dbench} --- templates drawn from a $300$-mesh subset filtered by their source metadata to winding-consistent meshes with concavity $\kappa \geq 0.1$ and face count $\leq 5000$ ($13$ pass the watertight filter, App.~\ref{ap:scene_gen}); and
\emph{Thingi10K}~\cite{zhou2016thingi10k} --- a $430$-mesh subset of artist-authored printable meshes loaded via the public HuggingFace dataset, restricted to closed, manifold, non-self-intersecting models with $1000$--$1500$ facets (the binding constraint; the secondary $500$--$5000$-vertex filter rarely fires under Euler's formula).
Both pools carry higher face counts than the curated 40-object Kubric pool of Sec.~\ref{sec:exp_scaling}, so mesh-query cost dominates solve time for the mesh-based methods.
We evaluate at $N \in \{500, 1000, 2000\}$, 3 seeds per cell, under the shared mesh-level evaluator (Sec.~\ref{sec:exp_setup}) and a 30-minute wall-time cap ($60$ minutes for \textsc{AVBD-OBB}, App.~\ref{ap:benchmark_limits}).
Scene generation is identical across methods, so every solver starts from the same initial scenes.

\begin{table*}[!tbp]
 \caption{Cross-dataset comparison on HY3D-Bench and Thingi10K. \name, QP/LCP, and PD-PGS share FCL-based mesh queries during solving; \textsc{AVBD-OBB} retains its OBB-SAT backend; every method is scored by the same evaluator. \name\ is the fastest CPU-tier solver in every cell ($2.5$--$3.8\times$ over QP/LCP, $5.4$--$12.9\times$ over PD-PGS) at pen${=}0$ with comparable displacement; \textsc{ISIR}, the GPU-tier entry, is listed for completeness; \textsc{AVBD-OBB} leaves growing residual penetration.}
 \label{tab:datasets}
 \centering
 \small
 \setlength{\tabcolsep}{14pt}
 \resizebox{\textwidth}{!}{%
 \begin{tabular}{@{}llccccccccc@{}}
 \toprule
 & & \multicolumn{3}{c}{$N{=}500$} & \multicolumn{3}{c}{$N{=}1000$} & \multicolumn{3}{c}{$N{=}2000$} \\
 \cmidrule(lr){3-5}\cmidrule(lr){6-8}\cmidrule(lr){9-11}
 Dataset & Method & Pen. & RMSD & Time & Pen. & RMSD & Time & Pen. & RMSD & Time \\
 \midrule
 \multirow{5}{*}{HY3D-Bench}
 & \textbf{\name} & \textbf{0} & \underline{0.021} & \textbf{5.2 s} & \textbf{0} & \underline{0.020} & \textbf{10.2 s} & \textbf{0} & \underline{0.021} & \textbf{22.2 s} \\
 & QP/LCP & \textbf{0} & 0.022 & \underline{13.2 s} & \textbf{0} & 0.022 & \underline{31.9 s} & \textbf{0} & 0.022 & \underline{65.4 s} \\
 & PD-PGS & \textbf{0} & \textbf{0.013} & 28.1 s & \textbf{0} & \textbf{0.013} & 73.3 s & \textbf{0} & \textbf{0.013} & 172 s \\
 & AVBD-OBB & 8.7 & 1.123 & 70 s & 46.7 & 0.935 & 210 s & 104.7 & 0.935 & 730 s \\
 & ISIR & \textbf{0} & 0.196 & 226 s & 0.3 & 0.197 & 949 s & 0.3 & 0.233 & 1472 s \\
 \midrule
 \multirow{5}{*}{Thingi10K}
 & \textbf{\name} & \textbf{0} & \underline{0.025} & \textbf{4.5 s} & \textbf{0} & \underline{0.026} & \textbf{8.7 s} & \textbf{0} & \underline{0.026} & \textbf{21.4 s} \\
 & QP/LCP & \textbf{0} & 0.026 & \underline{14.3 s} & \textbf{0} & 0.027 & \underline{32.7 s} & \textbf{0} & 0.028 & \underline{72.4 s} \\
 & PD-PGS & \textbf{0} & \textbf{0.016} & 43.1 s & \textbf{0} & \textbf{0.017} & 112 s & \textbf{0} & \textbf{0.017} & 250 s \\
 & AVBD-OBB & 14.0 & 1.332 & 72 s & 42.3 & 1.265 & 221 s & 134.0 & 1.234 & 872 s \\
 & ISIR & \textbf{0} & 0.305 & 230 s & 0.3 & 0.318 & 929 s & 0.3 & 0.348 & 1471 s \\
 \bottomrule
 \end{tabular}%
 }
 \par
 \raggedright\footnotesize{Bold names mark our variants; bold marks the best entry per column within a dataset block and underline the runner-up. Ties at pen${=}0$ are all bold, so no runner-up is marked there. \textsc{AVBD-OBB} figures are mesh-level scores of its OBB-proxy poses. \textsc{ISIR} runs in the GPU tier and the other rows in the CPU tier; wall times are compared within a tier (Sec.~\ref{sec:exp_setup}).}
\end{table*}

Across both datasets (Tab.~\ref{tab:datasets}) the optimization-based methods --- \name, QP/LCP, and PD-PGS --- reach pen${=}0$ in every cell, so runtime is the main differentiator: \name\ is the fastest in every cell; QP/LCP and PD-PGS take $2.5\times$ and $5.4\times$ its wall time at $N{=}500$ ($5.2$\,s vs.\ $13.2$/$28.1$\,s) and $2.9\times$/$7.7\times$ at $N{=}2000$ ($22.2$\,s vs.\ $65.4$/$172$\,s), with RMSD within $0.010$ of the lowest. On the harder \textsc{HY3D-Full} variant (the non-decimated $12$-template pool of App.~\ref{ap:scene_gen}, $\leq\!5000$-face meshes, initial penetrating pairs $311\to1312$ over $N{=}500\to2000$) the ranking is unchanged and the lead widens ($3.1\times$ over QP/LCP and $12.9\times$ over PD-PGS at $N{=}2000$); ISIR reaches pen${=}0$ on every $N{=}500$ cell and on two of three seeds at $N{=}1000$ and $N{=}2000$, where the third seed stops at its budget with one residual pair, and on every cell of the full-resolution pool; its wall times lie one to two orders of magnitude above \name's, a cross-tier comparison given for completeness rather than as a ranking. Thingi10K shows the same picture at lower geometric difficulty (\name\ RMSD $0.025$--$0.026$, $3.4\times$ faster than QP/LCP at $N{=}2000$). \textsc{AVBD-OBB} is the one method that stays far from feasibility: it retains $9$--$134$ mesh-level penetrating pairs at its step budget on these pools, consistent with its OBB proxy and momentum-carrying update. \name's own failure regimes (near-coincident reference centers, nested cavities, and tight confinement that blocks the deployed continuation) are discussed in Sec.~\ref{sec:conclusion}.

\paragraph{Non-convex robustness.} On a high-concavity HY3D-Bench sub-pool ($\kappa\in[0.3,0.95]$), \name\ produces RMSD ${\sim}126\times$ lower than the OBB-based AVBD demo while clearing every residual penetration (pen${=}0$ in all 5 seeds); details in App.~\ref{sec:exp_objaverse}.
\subsection{Qualitative Behavior Along the Scale Path}
\label{sec:exp_progression}

Fig.~\ref{fig:s4r_progression_gallery} shows the per-stage qualitative behavior of \name\ on four scenes drawn from all three pools. The leftmost cell of each row is the spawn configuration: $N$ randomly placed and rotated meshes whose mesh-level evaluator counts tens (small scenes) to hundreds (Kubric $N{=}300$: $221$) of penetrating body pairs. The middle three columns are snapshots of the same scene during \name's progressive-scaling sweep: at $s{=}0.31$ ($0.26$ on the Thingi10K row) the bodies are uniformly shrunk so their initial overlaps lie in the shallow-contact regime the per-step linear model is built for (Sec.~\ref{sec:progressive_scaling}); at $s{=}0.61$ and $s{=}0.91$ the QP-corrected layout is dilated back toward full size, and the rightmost cell is the resolved full-scale state with pen${=}0$ on all four scenes. Because the camera and lighting are locked within each row, every visual difference between the init and final cells is attributable to the reference-center trajectory produced by the solver.

\begin{figure*}
  \centering
  \includegraphics[width=\textwidth]{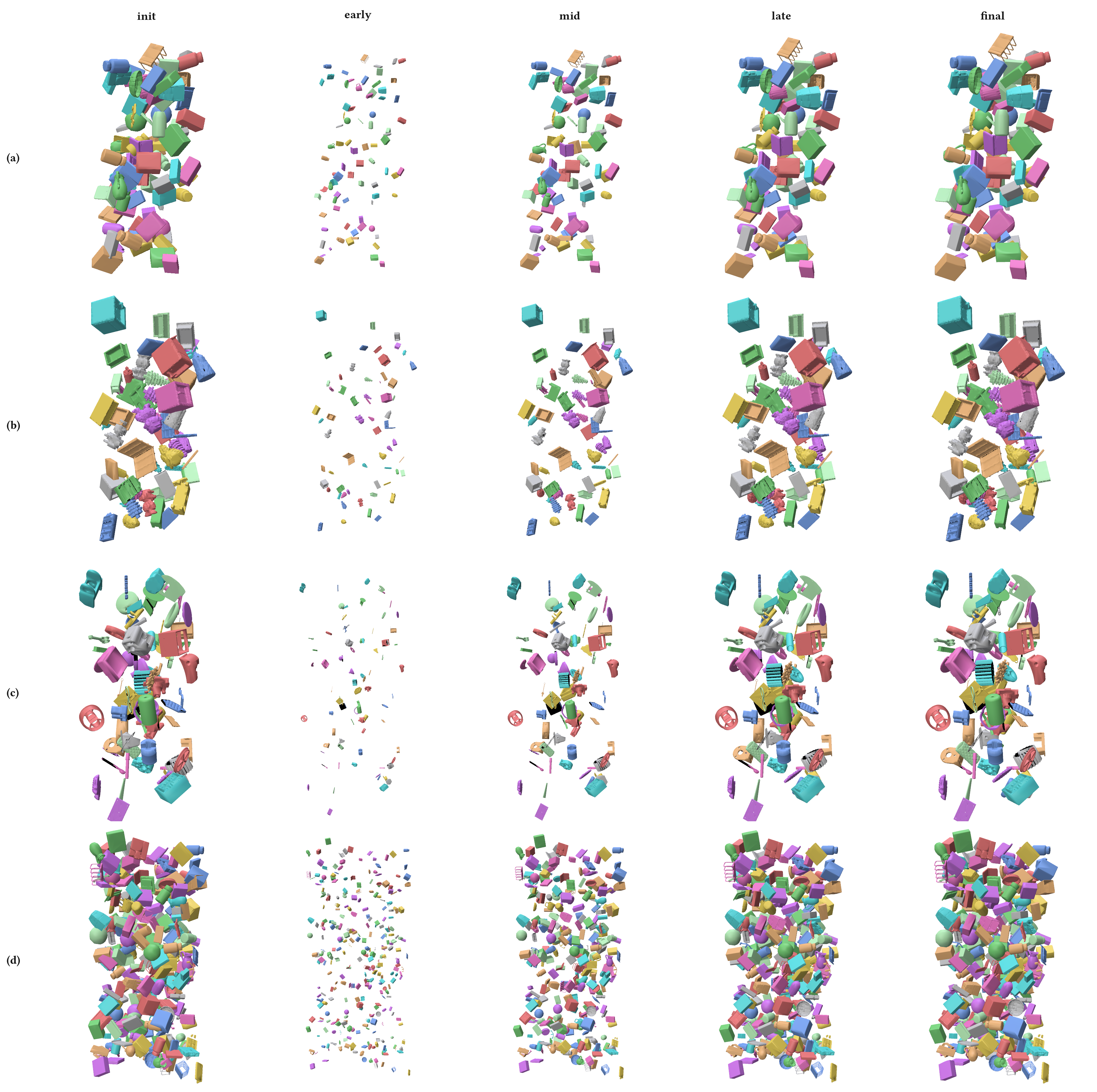}
  \caption{\name\ along the scale path: \textbf{(a)}~Kubric $N{=}80$, \textbf{(b)}~HY3D-Bench $N{=}60$, \textbf{(c)}~Thingi10K $N{=}80$, \textbf{(d)}~Kubric $N{=}300$. Each row runs left to right from the interpenetrating spawn through scaling ratios $0.31$ ($0.26$ for row c), $0.61$, $0.91$ to the full-scale result, under a camera locked per row. Every scene starts with tens to hundreds of penetrating pairs and ends at pen${=}0$ under the shared evaluator.}
  \Description{A four-row by five-column gallery of S4R progressive scaling, with rows labeled (a) through (d): (a) Kubric N=80, (b) HY3D-Bench N=60, (c) Thingi10K N=80, (d) Kubric N=300. Within each row the five columns are the initial spawn, scaling ratios 0.31 (0.26 for row c), 0.61, and 0.91, and the final result. The leftmost cells show densely overlapping bodies and the rightmost cells show the same bodies fully resolved with no penetration.}
  \label{fig:s4r_progression_gallery}
\end{figure*}

\subsection{Jointly Optimizing Positions and Rotations}
\label{sec:exp_rotation}
Our formulation extends naturally to jointly optimizing positions and rotations: each body gains a small-angle rotation increment $\boldsymbol{\omega}_i$, the objective becomes $\|\Delta \vp_i\|^2 + \beta \|\boldsymbol{\omega}_i\|^2$, and the contact rows acquire the corresponding lever-arm terms (App.~\ref{ap:6dof}). The scale schedule and tail logic are unchanged, but an accepted rotation alters the witness geometry, so it triggers fresh nearby-pair queries and invalidates the affected cache entries. App.~\ref{ap:6dof} details the formulation, and App.~\ref{sec:upright_plane} builds the upright-on-plane tabletop variant on the same rotational machinery.
Tab.~\ref{tab:abl_rotation} shows that on typical, loosely spawned Kubric layouts rotation buys a modest displacement reduction at a large wall-time cost; on packing-limited spawns it becomes a packing lever, lowering RMSD by $9$--$10\%$ at $2.4$--$4.8\times$ wall time (Tab.~\ref{tab:rotation_packing}), so the optional 6-DOF path is useful only when the application values the tighter pack enough to pay the extra contact-refresh cost.

\begin{table}[t]
  \caption{6-DOF vs.\ translation-only QP on Kubric (3-seed mean, $N{=}40$).
  Rotation lowers RMSD by about $13\%$ on these layouts but adds a ${\sim}15\times$ wall-time overhead.}
  \label{tab:abl_rotation}
  \centering
  \small
        \setlength{\tabcolsep}{10pt}

  \resizebox{\columnwidth}{!}{
  \begin{tabular}{@{}llccc@{}}
    \toprule
    $N$ & Method                      & Pen. & RMSD & Time \\
    \midrule
    \multirow{2}{*}{40}  & \textbf{3-DOF (translation only)} & 0 & 0.036 & \textbf{0.16 s} \\
                         & 6-DOF (translation + rotation)    & 0 & \textbf{0.031} & 2.4 s \\
    \bottomrule
  \end{tabular}
  }
\end{table}

On Kubric, the rotation DOFs leave the residual penetrating-pair count at zero and lower the translation-only RMSD by about $13\%$ ($0.031$ vs.\ $0.036$).
The wall-time overhead is ${\sim}15\times$ and is dominated by the auxiliary $\mathrm{SO}(3)$ refinement loop that runs after every QP update: each accepted angular step changes the witness points and contact normals, forcing a fresh distance pass over nearby pairs.
The $\mathrm{SO}(3)$ loop and the translation oracle share prebuilt mesh-query structures; rotation updates only modify the rigid transform used for nearby-pair checks.
Rotation DOFs do not pay for themselves on layouts where translation alone can separate the bodies; we therefore use 3-DOF by default and report 6-DOF as an optional extension rather than as part of the main method.

\paragraph{Rotation as a Packing Lever.}
To test whether rotation helps when packing binds, we run a controlled comparison on the same $N{=}40$ Kubric objects, solver, and evaluator over 8 seeds, varying only spawn density: \emph{loose} is the main-benchmark $N{=}40$ spawn, and \emph{tight} shrinks the spawn box $2.5\times$ per axis (${\sim}16\times$ denser, packing-limited). Every cell reaches pen${=}0$.
The 6-DOF variant lowers RMSD on 8/8 seeds in both regimes: $0.0792$ vs.\ $0.0869$ on tight spawns ($-9\%$) and $0.0262$ vs.\ $0.0290$ on loose spawns ($-10\%$).
The gain costs $2.4$--$4.8\times$ wall time, because each accepted rotation invalidates the warm start and the analytical distance cache; this trade-off is what justifies the translation-only default.
\begin{table}[t]
  \caption{Rotation as a packing lever on Kubric $N{=}40$ (8-seed mean; pen${=}0$ in every cell). Rows differ only in spawn density: \emph{loose} is the main-benchmark spawn; \emph{tight} is $2.5\times$ smaller per axis (${\sim}16\times$ denser). When packing binds, 6-DOF beats 3-DOF on 8/8 seeds.}
  \label{tab:rotation_packing}
  \Description{A two-row table comparing translation-only and six-degree-of-freedom variants on loose and tight spawns. On tight spawns rotation lowers displacement on all seeds at higher wall time.}
  \centering\small
  \setlength{\tabcolsep}{4pt}
  \resizebox{\columnwidth}{!}{%
  \begin{tabular}{@{}lccc@{}}
    \toprule
    Spawn & 3-DOF (RMSD / Time) & 6-DOF (RMSD / Time) & 6-DOF : 3-DOF \\
    \midrule
    Tight & 0.0869 / 1.98 s & \textbf{0.0792} / 9.59 s & RMSD $0.91\times$, time $4.8\times$ \\
    Loose & 0.0290 / 1.53 s & \textbf{0.0262} / 3.68 s & RMSD $0.90\times$, time $2.4\times$ \\
    \bottomrule
  \end{tabular}%
  }
\end{table}

\subsection{Application: Penetration-Free Initialization for Barrier-Based Solvers}
\label{sec:exp_engine_frontend}

Barrier-based solvers such as IPC presuppose an intersection-free state and have no mechanism to establish one from overlap; impulse- and penalty-based engines can step from an overlapping state, but their contact response has to undo the overlap, and on a deep one that response is violent. App.~\ref{sec:exp_engine_frontend_app} measures the second case on MuJoCo~\cite{todorov2012mujoco}, PyBullet~\cite{coumans2021pybullet}, and Isaac Gym~\cite{isaacgym} --- peak speeds of $7.8$, $12.5$, and $230.6$\,m/s from the raw layout, and a configuration that stays essentially at rest after \name.

\paragraph{IPC Feasibility.}
For IPC the failure is total rather than merely violent: its barrier acts on unsigned primitive distances and is kept finite by a CCD-filtered line search that assumes an intersection-free start, so an interpenetrating layout violates that invariant before the first Newton step, and IPC provides no procedure for restoring it.
On tightly spawned Kubric $N{=}40$ scenes (3 seeds, $\dhat{=}0.02$, same mesh evaluator as Sec.~\ref{sec:exp_setup}), the raw layouts contain $236$--$333$ penetrating pairs at depths up to $0.073$, and no seed satisfies IPC's precondition.
After \name, every seed reports pen${=}0$ with a strictly positive closest gap ($+6{\times}10^{-5}$ to $+3{\times}10^{-4}$), so the precondition holds (a log-barrier of IPC's form evaluated on the closest gap is finite, $1.6\times10^{-3}$--$2.3\times10^{-3}$; no IPC solve is run in this test), at RMSD $0.086$--$0.093$.
(The closest gap is smaller than the solve-time margin $\dhat{=}0.02$ by design: the tail refinement terminates when the detector reports \emph{zero penetrations}, not when every pair clears the full $\dhat$ margin, so gaps in $(0,\dhat)$ remain --- IPC needs only $d>0$, which is what matters here.)
\name thus establishes, rather than assumes, the intersection-free invariant that IPC maintains.

\subsection{Application: Rigid-Body Interpenetration Resolution for More 3D Assets}
\label{sec:exp_generated_asset}

Beyond the Kubric, HY3D-Bench, and Thingi10K benchmarks, \name is a natural post-process for generated 3D scenes whose coarse physical initialization contains mesh-level interpenetrations.
Together with the generatively produced HY3D-Bench pool (Sec.~\ref{sec:exp_setup}), this covers the main penetration sources we target: text/image-to-3D generator output, packed 3D asset datasets, and coarse scene-layout initialization.
We draw from an 88-asset pool of textured household meshes downloaded from the BlenderKit asset library~\cite{blenderkit} (max-extent--normalized, $12$ to $820$k faces), covering kitchenware, furniture, and everyday objects: the upright qualitative gallery (Fig.~\ref{fig:upright_yaw_qp_blender}) uses 8 scenes sampled from the full pool, while the quantitative Tab.~\ref{tab:generated_asset_repair} uses the five meshes that pass a watertight/manifold filter, screened over a 13-asset subset of the pool.  For the tabletop version of this application, each asset is constrained to a common support plane and an upright final pose using App.~\ref{sec:upright_plane}.  Fig.~\ref{fig:upright_yaw_qp_blender} shows eight representative scenes whose initial layouts span four overlap regimes (severe, heavy, moderate, light; controlled by spawn radius and an initial-penetration floor), with $28^\circ$--$42^\circ$ roll/pitch tilts and 97, 95, 67, 55, 41, 32, 26, and 20 initial mesh-level penetrating pairs.  During progressive inflation, the solve restores the upright pose while optimizing admissible in-plane translation and yaw, and reaches zero penetration at full scale in all eight scenes.  Collision evaluation and rendering use identical triangulated geometry and rigid transforms, and the final state is verified by a strict FCL pass against that same mesh.

\begin{figure*}
  \centering
  \includegraphics[width=\textwidth]{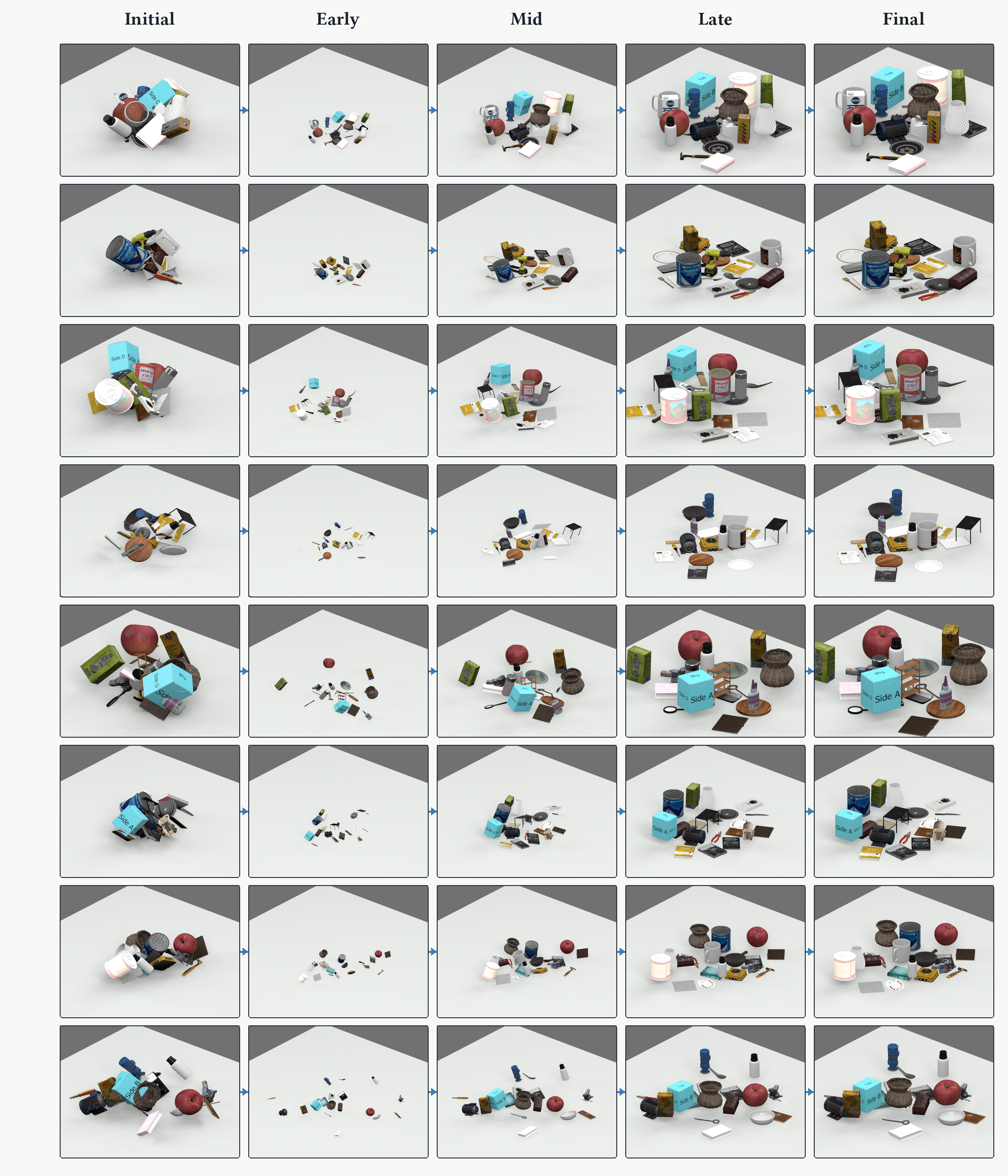}
  \caption{Upright-on-plane generated-asset repair. Eight scenes (rows) progress left to right from the tilted, mesh-interpenetrating initial state through 31\%, 61\%, and 91\% inflation (the Early/Mid/Late columns) to the final upright result.}
  \Description{Eight-row, five-column grid of Blender renders. Each row starts with tilted, interpenetrating objects and ends with the same objects standing upright on a shared support plane with zero pairwise penetration.}
  \label{fig:upright_yaw_qp_blender}
\end{figure*}

\begin{figure}
  \centering
  \includegraphics[width=\columnwidth]{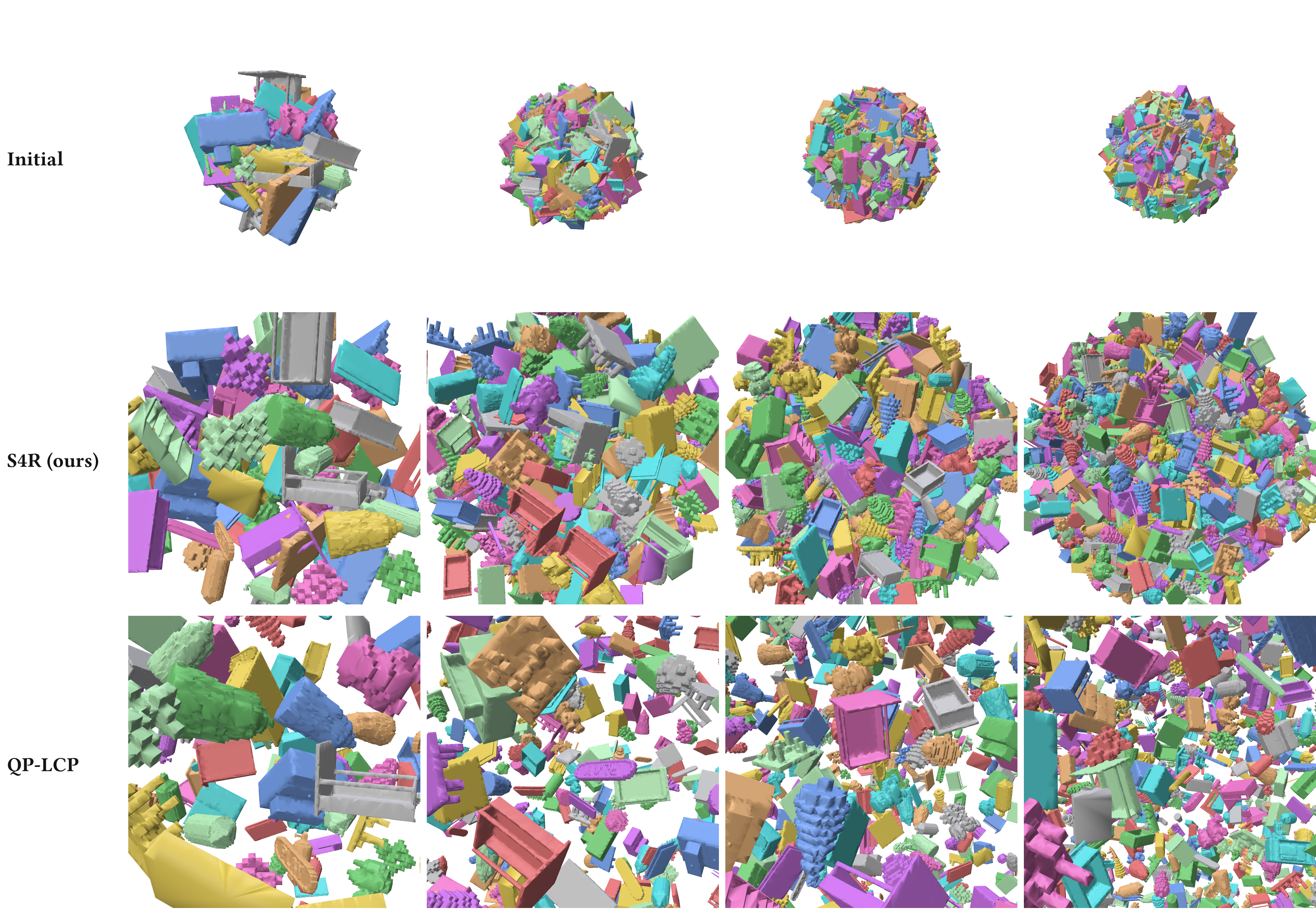}
  \caption{Final states for the spherical-spawn stress test (seed 42; rows: initial spawn, \name, QP/LCP; columns, left to right: $N{=}100$, $500$, $1000$, $2000$). Each column shares one camera anchored to the spawn-ball radius, so visible spread reflects actual displacement: \name\ keeps the pack tightest in every column.}
  \Description{Three-row, four-column gallery of Polyscope renders. Rows are the initial spawn, S4R (ours), and QP/LCP; columns are N=100, 500, 1000, 2000. Each column uses one shared camera frame anchored to the spawn-ball radius, so visual size differences across rows reflect the solver's actual displacement.}
  \label{fig:sphere_baselines_gallery}
\end{figure}

Fig.~\ref{fig:sphere_baselines_gallery} shows the densest regime we test: bodies spawned inside a ball at constant cell density, up to $25{,}000$ initial penetrating pairs at $N{=}2000$. Under a camera anchored to the spawn ball the difference in displacement is visible --- QP/LCP spreads the pack two to three times farther (App.~\ref{sec:exp_sphere}).

Tab.~\ref{tab:generated_asset_repair} further compares \name\ against QP/LCP and PD-PGS on the solver-usable subset of the pool introduced above. That subset is small --- five usable meshes --- so each scene repeats objects many times, which stresses contact reuse and penalizes solvers that rebuild the global contact graph every iteration (QP/LCP, PD-PGS). Face counts are also heterogeneous, $600$--$5476$ across the subset.  \name\ reuses mesh-query structures and cached contact information across scale steps, which is why the wall-time gap grows with $N$ rather than tracking the global QP size alone: \name\ is $2.8\times$ faster than QP/LCP and $11\times$ faster than PD-PGS at $N{=}100$, and $4.0\times$ and $13\times$ faster at $N{=}1000$.

\begin{table}[!tbp]
  \caption{Synthetic asset layout repair on five solver-usable meshes of the 88-asset BlenderKit pool (a watertight/manifold filter over a 13-asset subset; the broader gallery in Fig.~\ref{fig:upright_yaw_qp_blender} uses the full pool).
  3-seed averages. All solvers reach final pen${=}0$; \name\ is the fastest at every $N$, with RMSD below QP/LCP's and above PD-PGS's.}
  \label{tab:generated_asset_repair}
  \centering
  \small
  \setlength{\tabcolsep}{8pt}
  \resizebox{\columnwidth}{!}{%
  \begin{tabular}{@{}llcccc@{}}
    \toprule
    Method & $N$ & Init Pen. & Final Pen. & RMSD & Time \\
    \midrule
    \multirow{3}{*}{\textbf{\name}} & 100  & 87   & 0 & \underline{0.039} & \textbf{1.5 s} \\
                                     & 500  & 491  & 0 & \underline{0.043} & \textbf{11.6 s} \\
                                     & 1000 & 1031 & 0 & \underline{0.047} & \textbf{26.5 s} \\
    \midrule
    \multirow{3}{*}{QP/LCP}  & 100  & 87   & 0 & 0.041 & \underline{4.1 s}   \\
                              & 500  & 491  & 0 & 0.051 & \underline{43.1 s}   \\
                              & 1000 & 1031 & 0 & 0.051 & \underline{105 s}   \\
    \midrule
    \multirow{3}{*}{PD-PGS}  & 100  & 87   & 0 & \textbf{0.028} & 16.8 s   \\
                              & 500  & 491  & 0 & \textbf{0.031} & 123 s   \\
                              & 1000 & 1031 & 0 & \textbf{0.032} & 348 s   \\
    \bottomrule
  \end{tabular}%
  }
  \par
  \raggedright\footnotesize{We evaluate all three methods on identical scenes using the same unconstrained 3-DOF repair task. Bold marks the best entry per column and underline the runner-up; the penetration columns tie everywhere and are left unmarked.}
\end{table}

\subsection{Additional Stress Tests and Ablations}
\label{sec:exp_supplement_pointer}
The supplement reports several further evaluations (App.~\ref{sec:appendix}): a high-density spherical-spawn stress test to $N{=}2000$ (App.~\ref{sec:exp_sphere}); a contact-density sweep (App.~\ref{sec:exp_density}); a full component ablation that isolates each acceleration technique --- progressive scaling versus a direct full-scale QP, analytical scale-of-impact event handling, tail refinement, the analytical distance collision cache, schedule choice, and contact sparsity (App.~\ref{sec:exp_ablation}); the optional 6-DOF extension (evaluated in Sec.~\ref{sec:exp_rotation}, formulated in App.~\ref{ap:6dof}); and a large-scale runtime breakdown (App.~\ref{sec:exp_breakdown}). The acceleration components change wall time without changing feasibility on the tested configurations (one seed of one intermediate stage of the implementation stack leaves one pair, App.~\ref{sec:exp_perf_summary}): scale-of-impact event skipping and the analytical distance cache each remove redundant full detections, and contact sparsity keeps the per-step QP small --- together, this is what holds the measured wall time below the quadratic reference slope over the tested range (Sec.~\ref{sec:exp_scaling}). The progressive-scaling and tail-refinement ablations are different in kind: they disable feasibility-critical stages, and are reported as such.

\section{Discussion and Conclusion}
\label{sec:conclusion}

We have presented \name, a method that, when successful, transforms an infeasible configuration of interpenetrating rigid bodies into a state with zero evaluator-reported penetration, with low displacement from the prescribed layout.
The method shrinks every body to a separated scale, checked with the evaluator, then progressively restores full size, using a small convex QP at each increment to restore the linearized margin.
Across Kubric, HY3D-Bench, and Thingi10K, \name\ reaches zero evaluator-reported penetration in every main-benchmark cell with nearly $N$-invariant displacement over the main-comparison range ($N\leq5000$), and is the fastest solver in both hardware tiers of the main comparisons (Sec.~\ref{sec:exp}).

The scale parameter turns one problem initialized deep inside the infeasible region into a sequence of small local contact QPs: on a separated, stable witness branch the linearization error is second order in the realized per-step motion. Progressive scaling is designed to keep many of these motions small and does so empirically on our benchmarks --- though it does not certify a small correction under a poorly conditioned active set --- with witness switches and stale cached rows handled by later re-detection and tail refinement. Because each step minimizes $\tfrac12\sum_i\|\Delta \mathbf{p}_i\|^2$ over the current active constraints, with no proximal term toward $\mathbf{p}^0$, the per-step objective empirically contributes to low cumulative displacement --- without guaranteeing that every individual correction is small --- and the RMSD at $s{=}1$ is nearly invariant in $N$. Each component carries part of this: progressive scaling supplies the shallow-contact subproblems (a single global linearization leaves residual penetration), the sparse warm-started QP replaces a Newton-and-barrier stack, the analytical cache cuts wall time $1.5\times$ at $M{=}3$ for a small RMSD cost, and tail refinement corrects the full-scale state over the detected near-contact set. App.~\ref{sec:exp_ablation} reports the runtime and displacement effect of each component.

\paragraph{Limitations.}
The method assumes a separated full-scale configuration \emph{exists}: when the bodies collectively exceed the available space, the continuation either stops before full scale and returns an incomplete report, or reaches full scale with the tail refinement reporting residual penetrations. The start is likewise empirical rather than certified: Eq.~\ref{eq:smin_safe} is a bounding-sphere condition, near pairs are perturbed apart (charged to RMSD), and the down-scaled state is verified with the evaluator (Sec.~\ref{sec:progressive_scaling}). The per-step model degrades near medial-axis or feature-switch regions and under poorly conditioned active sets; we certify neither, and the solver carries no per-step feasibility invariant. Rotation is small-angle-linearized ($\|\vomega\|_\infty\leq 0.1$\,rad per step), so scenes requiring large re-orientation would need more steps or a nonlinear rotation subsolve, and the scale-up is quasi-static --- gravity, friction, and inertia are not modeled.

The deeper assumption is a feasible \emph{path}: a collision-free configuration at every intermediate scale, not merely at the two ends. In free space such a path always exists --- the whole layout can be dilated about any point --- but it may require large displacement, and the local minimum-norm continuation follows a low-displacement path instead. Where confinement or fixed geometry closes that path up in between, the monotone continuation stalls even though the full-scale problem has a solution. Tight confinement is the clearest case: App.~\ref{sec:exp_density} locates the onset at packing fraction $\phi\approx0.18$, where no seed reaches zero reported penetration --- consistent with path blockage or severe local conditioning, though not a certificate that no continuous path exists. Fully nested bodies fail at the other end (the deployed monotone path finds no exit from a closed full-scale cavity, although nesting alone does not rule out a collision-free scale path), and near-coincident centers make the separation direction so ill-conditioned that the push side can flip under small input changes. Finally, the upright-on-plane mechanism (App.~\ref{sec:upright_plane}) hard-codes one support plane and one up axis; oriented multi-support equilibrium --- a pile leaning against a wall --- would need the 6-DOF QP with per-support constraints, which we have not evaluated.

\paragraph{Future Work.}
\name\ is most useful to a time-stepping simulator as a repair stage, not a contact-solver replacement: simulators handle shallow contact well, while deep overlap violates their small-step assumptions (Sec.~\ref{sec:exp_engine_frontend}). The per-step program could absorb the physical terms it omits --- mass weighting, friction, velocity variables --- and hand the simulator consistent initial conditions. Shrinking also exposes an editable collision-free state (Fig.~\ref{fig:editing_workflow}): shrink an interpenetrating layout, rearrange it while separated, then restore scale along the same path. Where the monotone schedule is blocked, a discrete \emph{scale jump} across a finite infeasible interval --- the opposite of the SOI jump of Sec.~\ref{sec:soi}, which is guided by the intervals it can rule out --- would reach configurations the deployed continuation cannot. Articulated mechanisms, deformable bodies, exact primitive-level event scheduling, and batched GPU QP solves are natural extensions.

\begin{figure}[t]
  \centering
  \includegraphics[width=\columnwidth]{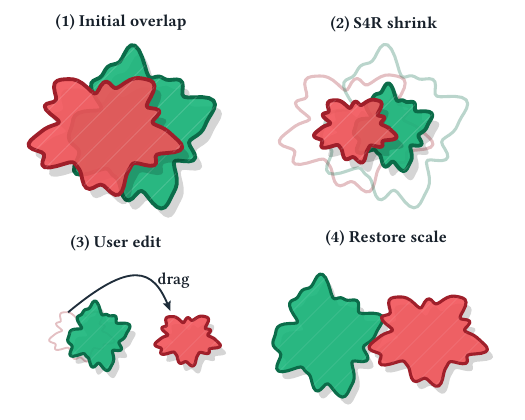}
  \caption{S4R as an editable scale-space. Two complex 2D shapes initially overlap; after shrinking, the user can move the red object to the right and keep the green object on the left, then restore the scale to obtain a separated layout.}
  \Description{A four-panel 2D illustration showing a red and a green complex shape. The first panel shows the shapes overlapping. The second shows both shapes shrunk with clearance. The third shows a user dragging the red shape to the right while the green shape remains on the left. The fourth shows the full-scale separated result.}
  \label{fig:editing_workflow}
\end{figure}

\begin{acks}
The authors sincerely thank all the reviewers for their constructive feedback. Zhiyang Dou is supported in part by a Hasso Plattner Fellowship (2026--2027) of the MIT and HPI AI and Creativity Hub (MHACH). This work was also supported by the Macao Science and Technology Development Fund (FDCT) under Grant 0119/2025/ITP2.
\end{acks}
\bibliographystyle{ACM-Reference-Format}
\bibliography{references.bib}

\clearpage
\appendix
\renewcommand\thefigure{\Alph{section}\arabic{figure}}
\renewcommand\thetable{\Alph{section}\arabic{table}}

\section{Supplementary Material}
\label{sec:appendix}

\emph{Table convention (all supplementary tables).} As in the main text, a bold method name marks our variant (\name/\name-Warp), and a bold value marks the best entry per column within its comparison group (lower error or time, higher speedup; ties at the displayed precision are all bold, purely descriptive tables carry no marks, and columns in which every entry is equal are left unmarked). Underline marks the runner-up where a caption says so, and a caption may restrict the marked columns.

\subsection{Runtime Cost Decomposition}
\label{sec:theory_speed}
\label{sec:theory}
Progressive scaling is not an unconditional speedup: it wins exactly when the mesh-query time it saves exceeds the extra continuation work it adds.
The following is an accounting model for the measured runtime, not a complexity theorem: per-query and per-QP costs are modeled as fixed, whereas real costs vary with pose, traversal, and warm start. Let $q_{ij}>0$ be the modeled wall-clock cost of one exact mesh-level proximity query for pair $(i,j)$.
Write $Q_{\mathrm{saved}}$ for the mesh-query time saved within this accounting model by evaluating contacts only at scale events along the progressive path (rather than re-querying every pair at full scale), and $A_{\mathrm{extra}}$ for the remaining cost of the continuation (per-scale QPs, cache maintenance, event scheduling, and tail refinement including its queries, minus the direct method's QP/overhead); App.~\ref{ap:theory_speed} gives the formal definitions.
Under fixed hardware, collision backend, and QP solver, the two totals differ only in these terms, so progressive scaling is faster than direct full-scale resolution precisely when
\begin{equation}
\label{eq:exact_speed_condition_app}
Q_{\mathrm{saved}} > A_{\mathrm{extra}}.
\end{equation}
Complex geometry helps only insofar as it raises the cost of the exact mesh queries the continuation path avoids: on small scenes, already-sparse contact, or backends with cheap queries, the added scale steps can dominate, which is why the one-shot Direct-QP ablation of Tab.~\ref{tab:abl_scaling} can be faster at small scale while leaving residual penetration.

\subsection{Details for the Conditional Speed Statement}
\label{ap:theory_speed}

We spell out the cost quantities used in the runtime cost decomposition (Eq.~\ref{eq:exact_speed_condition_app}).
Let $N$ be the number of rigid bodies, and let $q_{ij}>0$ denote the
wall-clock cost of one exact mesh-level proximity query for body pair $(i,j)$
under the chosen collision backend.  This quantity absorbs triangle count,
BVH quality, concavity, and implementation details; we do not assume a closed
form such as $q_{ij}=\Theta(V)$ in a mesh-complexity measure $V$ such as the face count.

Consider a direct full-scale method with $R_{\mathrm{dir}}$ outer
re-linearization iterations.  At iteration $r$, let
$\mathcal P_{\mathrm{dir}}^r$ be the candidate pairs queried by the exact
mesh detector, and let $m_{\mathrm{dir}}^r$ be the number of contact
inequalities in the global QP.  Let $\Psi(n,m)$ denote the setup and solve
cost of a QP with $n$ variables and $m$ inequality constraints.  The direct
cost is
\begin{equation}
\label{eq:theory_direct_cost}
\begin{aligned}
T_{\mathrm{dir}} =
\sum_{r=1}^{R_{\mathrm{dir}}}
\bigg[
  \sum_{(i,j)\in\mathcal P_{\mathrm{dir}}^r} q_{ij}
  + \Psi(3N,m_{\mathrm{dir}}^r)
  + h_{\mathrm{dir}}^r
\bigg],
\end{aligned}
\end{equation}
where $h_{\mathrm{dir}}^r\geq0$ collects remaining per-iteration overhead.

For \name, let $K$ be the number of scale steps and
$\mathcal I_{\mathrm{det}}\subseteq\{0,\ldots,K-1\}$ the steps at which the
exact mesh detector is invoked.  At step $k$, let $\mathcal P_k$ be the
candidate pairs passed to the detector, $\mathcal A_k$ the active QP contact
set, and $\mathcal B_k$ the bodies incident to $\mathcal A_k$.  The
progressive cost is
\begin{equation}
\label{eq:theory_s4r_cost}
\begin{aligned}
T_{\mathrm{S4R}} =
&\sum_{k\in\mathcal I_{\mathrm{det}}}
  \sum_{(i,j)\in\mathcal P_k} q_{ij}
 +\sum_{k=0}^{K-1}\Psi(3|\mathcal B_k|,|\mathcal A_k|)\\
&\quad + H_{\mathrm{cache}} + T_{\mathrm{tail}},
\end{aligned}
\end{equation}
where $H_{\mathrm{cache}}\geq0$ is the total cost of analytical distance
updates, scale-event scheduling, and warm-start maintenance, and
$T_{\mathrm{tail}}\geq0$ is the final refinement cost.

Subtracting Eq.~\ref{eq:theory_s4r_cost} from
Eq.~\ref{eq:theory_direct_cost} yields
\begin{equation*}
\begin{aligned}
T_{\mathrm{dir}}-T_{\mathrm{S4R}}
=Q_{\mathrm{saved}}-A_{\mathrm{extra}},
\end{aligned}
\end{equation*}
where
\begin{align}
Q_{\mathrm{saved}}
&=
\sum_{r=1}^{R_{\mathrm{dir}}}
  \sum_{(i,j)\in\mathcal P_{\mathrm{dir}}^r} q_{ij}
-
\sum_{k\in\mathcal I_{\mathrm{det}}}
  \sum_{(i,j)\in\mathcal P_k} q_{ij},\\
A_{\mathrm{extra}}
&=
\sum_{k=0}^{K-1}\Psi(3|\mathcal B_k|,|\mathcal A_k|)
+H_{\mathrm{cache}}+T_{\mathrm{tail}}\notag\\
&\quad -
\sum_{r=1}^{R_{\mathrm{dir}}}
\left[
\Psi(3N,m_{\mathrm{dir}}^r)+h_{\mathrm{dir}}^r
\right].
\end{align}
Therefore $T_{\mathrm{S4R}}<T_{\mathrm{dir}}$ if and only if
$Q_{\mathrm{saved}}>A_{\mathrm{extra}}$, which proves
the cost decomposition of Eq.~\ref{eq:exact_speed_condition_app}.

For interpretation, suppose query costs are comparable at a fixed mesh
complexity scale $V$: there exist $0<c_{\min}\leq c_{\max}$ and a positive
function $G(V)$ such that
\begin{equation}
c_{\min}G(V)\leq q_{ij}\leq c_{\max}G(V)
\end{equation}
for all queried pairs.  Let
\[
C_{\mathrm{dir}}=\sum_{r=1}^{R_{\mathrm{dir}}}
|\mathcal P_{\mathrm{dir}}^r|,
\qquad
C_{\mathrm{S4R}}=\sum_{k\in\mathcal I_{\mathrm{det}}}|\mathcal P_k|,
\]
and let $A_{\mathrm{dir}}$ and $A_{\mathrm{S4R}}$ denote the remaining (non-main-loop-query) parts of
Eqs.~\ref{eq:theory_direct_cost} and~\ref{eq:theory_s4r_cost}.  A useful
sufficient condition for speedup is
\begin{equation}
c_{\min}C_{\mathrm{dir}}G(V)-c_{\max}C_{\mathrm{S4R}}G(V)
>
A_{\mathrm{S4R}}-A_{\mathrm{dir}}.
\end{equation}
If $c_{\min}C_{\mathrm{dir}}>c_{\max}C_{\mathrm{S4R}}$, this becomes the
mesh-complexity threshold
\begin{equation}
G(V)>
\frac{[A_{\mathrm{S4R}}-A_{\mathrm{dir}}]_+}
{c_{\min}C_{\mathrm{dir}}-c_{\max}C_{\mathrm{S4R}}},
\qquad [x]_+=\max(x,0).
\end{equation}
Thus complex meshes favor progressive scaling only when the progressive path
reduces enough exact mesh queries for the saved query time to dominate its
extra continuation work.

\paragraph{Penalty-cleanup fallback.}\phantomsection\label{ap:penalty_cleanup} The released solver also offers an optional pairwise penalty pass after the correction-QP tail: each remaining penetrating pair is nudged apart along its current normal by its penetration depth plus a small margin, split equally between the two bodies, for up to $200$ iterations. It is switched off in every reported run; all \name\ numbers in this paper come from the progressive-scaling QP and the bounded tail refinement alone.

\subsection{Proof of Lemma~\ref{lem:taylor_residual} (Linearization Residual)}
\label{ap:theory_linearization}

\textbf{Linearization residual.}
Linearizing a distance function discards the second-order remainder
\begin{equation*}
  R(\Delta \mathbf{x}) \;=\; d(\mathbf{x}+\Delta \mathbf{x})-d(\mathbf{x})-\nabla d(\mathbf{x})^{\top}\Delta \mathbf{x} .
\end{equation*}
We bound it on a fixed smooth branch and assume the Hessian bound directly, rather than deriving it from the reach of the underlying surfaces.

  \begin{lemma}[First-order residual on a smooth local branch]
  \label{lem:taylor_residual}
  Let $f:\Omega\to\mathbb R$ be $C^2$ on an open convex set $\Omega\subset\mathbb R^m$ with $\|\nabla^2 f(\mathbf{x})\|\leq L_{\mathrm{seg}}$ for all $\mathbf{x}\in\Omega$, where $\|\cdot\|$ denotes the Euclidean norm and the spectral norm it induces.
  If $\mathbf{x},\mathbf{x}+\Delta \mathbf{x}\in\Omega$, then
  \begin{equation}
    \bigl|f(\mathbf{x}+\Delta \mathbf{x})-f(\mathbf{x})-\nabla f(\mathbf{x})^{\top}\Delta \mathbf{x}\bigr|\;\leq\;\frac{L_{\mathrm{seg}}}{2}\|\Delta \mathbf{x}\|^2 .
  \end{equation}
  \end{lemma}

\begin{proof}
Taylor's theorem with integral remainder gives
\[
f(\mathbf{x}+\Delta \mathbf{x})
=
f(\mathbf{x})+\nabla f(\mathbf{x})^{\top}\Delta \mathbf{x}
+
\int_0^1
(1-t)\,
\Delta \mathbf{x}^{\top}\nabla^2 f(\mathbf{x}+t\Delta \mathbf{x})\Delta \mathbf{x}\,\mathrm{d}t .
\]
Taking absolute values and using $\|\nabla^2 f(\mathbf{x})\|\leq L_{\mathrm{seg}}$ on $\Omega$ gives
\[
\left|
f(\mathbf{x}+\Delta \mathbf{x})-f(\mathbf{x})-\nabla f(\mathbf{x})^{\top}\Delta \mathbf{x}
\right|
\leq
\int_0^1(1-t)L_{\mathrm{seg}}\|\Delta \mathbf{x}\|^2\,\mathrm{d}t
=
\frac{L_{\mathrm{seg}}}{2}\|\Delta \mathbf{x}\|^2 .
\]
\end{proof}

The residual is thus quadratic in the step magnitude on a fixed smooth branch. This is a local conditioning statement: it says nothing across witness switches or nonsmooth mesh features, and the implementation never estimates $L_{\mathrm{seg}}$.

Where does such an $L_{\mathrm{seg}}$ come from, and why does it degrade with depth? For intuition only, consider the signed distance to a single smooth surface of \emph{reach} $\rho$~\cite{federer1959}, the largest distance within which every query point has a unique closest surface point. Positive reach alone does not justify assuming a $C^2$ distance branch~\cite{federer1959}; for a $C^2$ embedded surface the signed distance is $C^2$ in a sufficiently small tubular neighborhood~\cite[Theorem~1 and Remark~2]{foote1984regularity}, which is why we assume the $C^2$ regularity of the pairwise distance branch separately; we call this the \emph{local smooth-distance assumption}, and it underlies every local bound in this paper. Reach does not provide a uniform Hessian bound over the entire reach tube; on a restricted band $|d|\leq\delta<\rho$ it does: writing a query point as $\mathbf{x}=\mathbf{y}+d\,\mathbf{n}(\mathbf{y})$, with $\mathbf{y}$ its closest surface point and $\mathbf{n}$ the surface normal there, the Hessian of $d$ has the tangential eigenvalues $\kappa_a/(1+d\,\kappa_a)$, where $\kappa_a$ are the principal curvatures of the surface at $\mathbf{y}$, and a zero normal eigenvalue; since $|\kappa_a|\leq1/\rho$ on a surface of reach $\rho$ (the reach of the surface itself, not of the solid it bounds), $\|\nabla^2 d\|\leq 1/(\rho-|d|)$, which degrades toward the medial axis and gives no bound at all once $|d|\geq\rho$. If the whole update segment stays inside that band,
\[
  \sup_{t\in[0,1]}|d(\mathbf{x}+t\Delta \mathbf{x})|\;\leq\;\delta\;<\;\rho,
  \qquad \|\Delta \mathbf{x}\|\leq\delta,
\]
then one may take $L_{\mathrm{seg}}=1/(\rho-\delta)$ along the segment, and Lemma~\ref{lem:taylor_residual} gives $|R(\Delta \mathbf{x})|\leq\delta^2/[2(\rho-\delta)]$, asymptotic to $\delta^2/(2\rho)$ only in the shallow limit $\delta\ll\rho$. Note that the depth bound alone does not imply the band condition: it has to be assumed. This bound concerns the distance from a point to one surface; the constant $L_{\mathrm{seg}}$ of Sec.~\ref{sec:direct} refers to the pairwise distance branch in the coordinates used there and is assumed separately.

Fig.~\ref{fig:truncation_error} makes the same point on the one case where the residual is exact rather than bounded. Take the signed distance to a circle of radius $\rho$,
\begin{equation*}
  d(\mathbf{x})=\|\mathbf{x}\|-\rho ,
\end{equation*}
and an interior query point $\mathbf{x}_0=(a,0)$ with $a=\rho-\delta_0>0$, where $\delta_0$ is its penetration depth, so $a$ is the remaining distance to the medial axis (here the center). For a tangential update $\Delta \mathbf{x}=(0,\eta)$ we have $\nabla d(\mathbf{x}_0)^{\top}\Delta \mathbf{x}=0$, and the exact first-order residual is
\begin{equation}
\label{eq:circle_residual}
  \begin{aligned}
    \mathcal{E}(\eta;a) &\;:=\; d(\mathbf{x}_0+\Delta \mathbf{x})-d(\mathbf{x}_0)-\nabla d(\mathbf{x}_0)^{\top}\Delta \mathbf{x} \\
              &\;=\; \sqrt{a^{2}+\eta^{2}}-a \;=\; \frac{\eta^{2}}{\sqrt{a^{2}+\eta^{2}}+a},
  \end{aligned}
\end{equation}
which equals $\eta^{2}/(2a)+\mathcal O(\eta^{4}/a^{3})$ for $|\eta|\ll a$. The local quadratic coefficient is therefore $1/(2a)$: it degrades as the query point approaches the medial axis, and the quadratic regime itself holds only while $|\eta|\ll a$, crossing over to linear growth for $|\eta|\gg a$. The example is illustrative only: the deployed solver never estimates this scale.

For a direct method, the displacement required to escape a deep penetration
can be comparable to the penetration depth $\delta$, giving a residual
$\mathcal{O}(L_{\mathrm{seg}}\delta^2)$.  In progressive scaling, a scale step produces a per-step relative
motion $h_{ij}=\|\Delta \mathbf{p}_j^\star-\Delta \mathbf{p}_i^\star\|+\mathit{ds}\,(r_i+r_j)$, giving
a residual $\mathcal{O}(L_{\mathrm{seg}}h_{ij}^2)$; the inflation part of $h_{ij}$ is $\mathcal O(\mathit{ds}\,r_{\max})$
by construction, while the QP-correction part is small only when the active
constraints are well conditioned (Sec.~\ref{sec:objective}).  This argument assumes a local $C^2$ signed-distance function.
Triangle meshes violate this assumption globally at edges and vertices, so
the lemma is a local conditioning argument rather than a
global correctness proof; the mesh-level detector and tail refinement provide
the final feasibility check.

\subsection{Piecewise-Affine Property of the Frozen-Normal Translation QP}
\label{ap:proof_piecewise}

This appendix justifies the warm-start remark of Sec.~\ref{sec:schedule}. It concerns the \emph{translation-only} QP with frozen witnesses and normals; with rotational variables enabled the constraint matrix itself varies through the lever arms and no affine characterization is claimed.

Fix the current accepted state and freeze the row set, witnesses, normals, and closure coefficients, and consider the one-parameter trial QP
\[
\vq^\star(\lambda) \;=\; \argmin_{\vq}\; \tfrac{1}{2}\|\vq\|^2 \quad\text{s.t.}\quad \mA\,\vq \;\geq\; \vb(\lambda):=\vb_0+\lambda\,\vb_1,
\]
where $\mA$, $\vb_0$ and $\vb_1$ are constant and $\lambda$ parameterizes a trial inflation; $\vq$ stacks the translational corrections $\Delta\mathbf{p}_i$, $\lambda$ plays the role of $\mathit{ds}$, and the rows of $\vb_1$ are the closure coefficients $E_{ij}$. Let $D=\{\lambda:\ \mA\vq\geq\vb(\lambda)\ \text{for some}\ \vq\}$ be the feasible parameter domain; for $\lambda\in D$ the strictly convex objective has a unique minimizer.

\begin{proposition}
On every open interval of $D$ over which the binding row set $\mathcal R$ of $\vq^\star(\lambda)$ is constant and $\mA_{\mathcal R}$ has full row rank, $\vq^\star(\lambda)$ is affine in $\lambda$.
\end{proposition}
\begin{proof}
On such an interval the binding constraints satisfy $\mA_{\mathcal R}\vq^\star=\vb_{0,\mathcal R}+\lambda\vb_{1,\mathcal R}$, where the subscript $\mathcal R$ selects rows, and Karush--Kuhn--Tucker (KKT) stationarity for $\mA\vq\geq\vb(\lambda)$ reads $\vq^\star=\mA_{\mathcal R}^\top\bm{\mu}$ with multipliers $\bm{\mu}\geq0$ (written $\bm{\mu}$ to keep $\lambda$ for the trial parameter). Hence, while primal and dual feasibility hold,
\[
\begin{aligned}
\vq^\star(\lambda) &\;=\; \mA_{\mathcal R}^{\top}\bigl(\mA_{\mathcal R}\mA_{\mathcal R}^{\top}\bigr)^{-1}\bigl(\vb_{0,\mathcal R}+\lambda\,\vb_{1,\mathcal R}\bigr), \\[2pt]
\frac{\mathrm{d}\vq^\star}{\mathrm{d}\lambda} &\;=\; \mA_{\mathcal R}^{\top}\bigl(\mA_{\mathcal R}\mA_{\mathcal R}^{\top}\bigr)^{-1}\vb_{1,\mathcal R},
\end{aligned}
\]
so the derivative is constant on that interval and $\vq^\star(\lambda)$ is affine there.
\end{proof}
Across such intervals $\vq^\star(\lambda)$ is therefore piecewise affine. At rank-deficient binding sets the inverse above is not used; the piecewise-affine characterization of the unique primal solution on $D$ follows from the parametric-QP result in~\cite{tondel2003mpqp}.

The statement concerns only this frozen local trial QP. A change of row set --- whether from cached activation/deactivation or from a full refresh --- or of witness or normal replaces $(\mA,\vb_0,\vb_1)$ and falls outside the fixed-QP characterization. We use the property to seed a warm start --- it never replaces a QP solve, and the analytical cache propagates frozen-witness gap predictions (Eq.~\ref{eq:cache_update}), not $\vq^\star$.

\subsection{Solver--Evaluator Consistency}
\label{ap:solver_evaluator_consistency}

In the CPU-FCL implementation, at full-detection steps and during tail refinement \name\ uses the same piecewise FCL convention as Algorithm~\ref{alg:evaluator}: candidate pairs come from an AABB pass, and each surviving pair is queried with FCL BVH--BVH \texttt{distance} when separated or \texttt{collide} with per-contact depths when intersecting. Cached steps instead use the frozen-witness prediction $\breve d_{ij}$, which is never treated as a measured score. Every returned pose is then scored by Algorithm~\ref{alg:evaluator}, so the solver's exact queries and the final metrics share one sign convention, while cached predictions stay explicitly separate. \name-Warp uses its own contact oracle during the solve (Sec.~\ref{sec:exp_setup}) and shares the final evaluator. QP/LCP, PD-PGS, and Soft-Penalty use the same proximity primitive.

\subsection{Default Hyperparameters}
\label{ap:hyperparams}

Unless stated otherwise, every \name run in Sec.~\ref{sec:exp} uses the following defaults (Tab.~\ref{tab:hyperparams}).

\begin{table}[t]
\centering
\small
\setlength{\tabcolsep}{5pt}
\renewcommand{\arraystretch}{1.15}
\begin{tabular}{@{}lll p{0.40\columnwidth}@{}}
\toprule
Parameter & Sym. & Value & Notes \\
\midrule
Clearance            & $\dhat$              & $0.02$       & $0.02$ for \name/QP/LCP/Soft-Penalty; {PD-PGS uses $10^{-3}$;} Drake uses a $10^{-4}$ lower bound \\
\rowcolor{rowgray} Scale lower bound    & $s_{\min}$           & $0.01$       & checked per scene; Sec.~\ref{sec:progressive_scaling}, App.~\ref{ap:scene_gen} \\
Base scale stride    & $\mathit{ds}_{\max}$ & $0.05$       & fixed default; not derived from any certified bound \\
\rowcolor{rowgray} Cache refresh        & $M$                  & $3$          & Sec.~\ref{sec:cache}; ablated at $M\in\{1,2,3,5,10\}$, no certified maximum \\
Contact backend      &                      & FCL          & CPU default: prebuilt-BVH FCL; \name-Warp: Warp-native oracle \\
Initial separation padding & $\varepsilon_{\mathrm{sep}}$ & $10^{-6}$ & added to the perturbation target \\
Minimum retry increment & $\mathit{ds}_{\min}$ & $10^{-6}$ & Alg.~\ref{alg:s4r} step-halving floor \\
\rowcolor{rowgray} Scale-loop attempt budget & $K_{\max}$ & $200$ & Alg.~\ref{alg:s4r}; not reached in any reported run \\
\rowcolor{rowgray} OSQP abs/rel tol     & $\varepsilon$        & $10^{-6}$    & $10^{-7}$ in the tail loop \\
OSQP max iters       &                      & $4000$       & $8000$ in the tail loop \\
\rowcolor{rowgray} Rotation DOFs        &                      & off          & 6-DOF via \texttt{-{}-rotation} \\
6-DOF weight         & $\beta$              & $r_{\max}^2$ & surface-motion normalization \\
\rowcolor{rowgray} 6-DOF clamp          & $\omega_{\max}$      & $0.1$\,rad   & per component ($\|\boldsymbol{\omega}_i\|_\infty$) \\
Schedule             &                      & uniform      & $\mathit{ds}_{\max}{=}0.05$, adaptive skip \\
Tail max iters       & $K_{\mathrm{tail}}$  & $20$         & 3-iter stagnation guard \\
\bottomrule
\end{tabular}
\caption{Default hyperparameters of \name.}
\label{tab:hyperparams}
\end{table}

Baseline hyperparameters follow the reference implementations: Drake-Ipopt
uses $10^4$ Ipopt maximum iterations (an explicit exception to the reference settings; App.~\ref{ap:baseline_geometry}), and Drake-SNOPT sets SNOPT's iteration limits high enough that they never bind; QP/LCP uses 50 outer iterations; PD-PGS uses
relaxation 0.8 and iterates to feasibility within its wall-clock budget (its sweep cap is non-binding); Soft-Penalty uses $\mu_0=10$, $4\times$
growth, and 6 outer loops; \textsc{AVBD-OBB} (tuned) uses its own $\alpha=0.97$, $\beta=10^5$,
damping $0.95$, and $100$ iterations; the official configuration keeps $\alpha=0.99$, $\beta=10^4$, damping $0$, and $10$ iterations.

\subsection{Scene Generation and Units}
\label{ap:scene_gen}
Where metric lengths are quoted (e.g.\ centimeters in the stress tests), we take one scene unit $=1$\,m.

Every benchmark scene is generated from a fixed seed by the rules below.

\emph{Per-object normalization.}
Each source mesh is recentered at the center of its normalization box --- the object bounds shipped with each Kubric and HY3D-Bench asset, and the mesh's own bounding box for Thingi10K --- and this center is the body ``reference center'' $\mathbf{c}_i$ used everywhere in the paper; the mesh is uniformly scaled so the box's \emph{longest} axis equals $\texttt{target\_size}=0.1$ world units: $\mathrm{nf}_i = 0.1/\max(\mathrm{bbox}_i)$, $\bar{\mathbf{v}} \leftarrow \mathrm{nf}_i\,(\mathbf{v}_{\mathrm{raw}} - \text{bbox center})$, and world-space vertices are $\mathbf{v}_i(1) = \mathbf{R}_i\bar{\mathbf{v}} + \mathbf{c}_i$. All three datasets therefore share the same normalized units, in which every spatial quantity of the paper ($\hat d{=}0.02$, RMSD, penetration depths) is expressed.

\emph{Scene assembly (constant-density spawn box).}
With $f=(N/40)^{1/3}$, body centers are drawn i.i.d.\ uniformly from the box $[-a_{\mathrm{xz}},a_{\mathrm{xz}}]\times[-a_{\mathrm{y}},a_{\mathrm{y}}]\times[-a_{\mathrm{xz}},a_{\mathrm{xz}}]$ with $a_{\mathrm{xz}}=0.15f$ and $a_{\mathrm{y}}=0.35f$ (full sides $0.30\times0.70\times0.30$ at $N{=}40$), so the expected density is $N$-invariant. Orientations use three independent uniform Euler angles $\vartheta_x,\vartheta_y,\vartheta_z\sim \mathcal U(0,2\pi)$ composed as $\mathbf{R}=\mathbf{R}_x(\vartheta_x)\,\mathbf{R}_y(\vartheta_y)\,\mathbf{R}_z(\vartheta_z)$ from the rotations $\mathbf{R}_x,\mathbf{R}_y,\mathbf{R}_z$ about the coordinate axes (not Haar-uniform on $\mathrm{SO}(3)$; the same convention for every method). Mesh templates are sampled from the dataset pool with replacement when $N$ exceeds the pool size, with the same seeds for every method.

\emph{Dataset pools.} Kubric: $40$ templates selected with a $500$--$5000$-vertex metadata filter; their collision meshes have $64$--$1170$ vertices. HY3D-Bench: watertight volume meshes decimated toward $1500$ faces ($1438$--$4254$ after repair) (\textsc{HY3D-Decimated}; the $\leq\!5000$-face source pool is \textsc{HY3D-Full}, and the high-concavity subset \textsc{HY3D-Concave}); the benchmark filters retain $13$ templates in the decimated pool and $12$ in the full pool, which the scenes sample with repetition. Thingi10K: pre-validated pool with a $1000$--$1500$-facet band. The spherical-spawn stress variant replaces the box by a ball of radius $0.10\,(N/100)^{1/3}$.

\paragraph{Initial-scale perturbation.}
At the fixed $s_{\min}$, any pair with $\|\mathbf{c}_i-\mathbf{c}_j\| < \hat d + s_{\min}(r_i+r_j)$ violates Eq.~\ref{eq:smin_safe}, including exactly coincident centers, which no positive $s_{\min}$ admits. Such centers are pushed apart deterministically along their pair axis (a fixed axis for exact coincidence) to $\|\mathbf{c}_i-\mathbf{c}_j\|\geq\hat d + s_{\min}(r_i+r_j) + \varepsilon_{\mathrm{sep}}$; the offending pairs are visited once, and the displacement is charged to the reported RMSD (at most $4\%$ of it across the scene sizes we measured, $8\%$ for the GPU pipeline, whose sweep pushes each visited pair by a fixed $\hat d$). Since Eq.~\ref{eq:smin_safe} is only a bounding-sphere condition, we check the mesh-level statement directly: shrinking every body to $s_{\min}$ and scoring the layout with the evaluator of App.~\ref{ap:evaluator} returns zero penetrating pairs on every scene reported in this paper, with a closest gap of at least $0.014$ on the main-benchmark scenes and $0.009$ on the stress scenes.

\subsection{Baselines and Protocol}
\label{ap:benchmark_limits}
\paragraph{Baselines.}
\begin{itemize}[leftmargin=1.2em,itemsep=0.2em,topsep=0.2em,parsep=0pt,partopsep=0pt]
  \item[(i)] \textbf{AVBD}~\cite{Giles2025AVBD} --- the official 3D implementation of Augmented Vertex Block Descent (\textsc{AVBD-OBB}), which represents bodies as OBBs and uses OBB--OBB SAT collision.
  \item[(ii)] \textbf{Rigid-ISIR}~\cite{Jang2025ISIR} --- our rigid-pose adaptation of Instant Self-Intersection Repair: the original method flows mesh vertices along local signed tangent-point energies; we keep its energy, CUDA mesh-intersection BVH and optimizer schedule, and aggregate the per-vertex gradients onto per-body 6-DOF rigid poses (one global GPU step per outer iteration), with the step size selected on the validation seeds.
  \item[(iii)] \textbf{Drake}~\cite{drake,wachter2006ipopt,gill2005snopt} --- inverse kinematics with minimum-distance lower-bound constraints on convex hulls of the meshes, solved by Ipopt (\textsc{Drake-Ipopt}) or by Drake's bundled SNOPT (\textsc{Drake-SNOPT}).
  \item[(iv)] \textbf{Global QP/LCP} --- an iteratively re-linearized global contact QP solved by OSQP~\cite{stellato2020osqp} until feasibility.
  \item[(v)] \textbf{PD-PGS} --- a position-based projected Gauss--Seidel baseline: local sweeps over the overlap graph with connected-component decomposition and mass-weighted displacement splitting; an implementation of the standard projection sweep rather than a published method.
  \item[(vi)] \textbf{Soft-Penalty} --- a static penalty-continuation baseline after TrajOpt's hinge collision loss and penalty-strength continuation~\cite{schulman2014trajopt}, with a squared-hinge penetration penalty minimized by L-BFGS-B.
\end{itemize}

\paragraph{Baseline configurations.}
\label{ap:baseline_geometry}
Drake uses its default collision geometry for these meshes, one convex hull per body, a $10^{-4}$ minimum-distance lower bound and an influence distance of $0.05$; its Ipopt iteration budget is raised to $10^4$ so that the solver's own convergence test, not the stock iteration cap, ends the solve, and \textsc{Drake-SNOPT} solves the same program with non-binding iteration limits. Because the hull hides non-convex cavities, Drake separates highly non-convex bodies farther than the meshes require, which raises its RMSD under the mesh-level evaluator; Drake is reported on the Kubric setting up to $N{=}100$.
\textsc{AVBD-OBB (tuned)} is the official demo under our quasi-static adaptation with a convergence-binding step budget; \textsc{AVBD-OBB (official)} keeps the released defaults, which have no static-repair termination, so we report the first stable pen$_{\mathrm{OBB}}{=}0$ pose within a $20{,}000$-step cap (reached by $2/3$ seeds at $N{=}40$), while its Time covers the full capped run of those seeds. The OBB proxy overstates the overlap of non-convex meshes, so the solver separates bodies farther than mesh-level contact would require, consistent with its RMSD of $1.15$--$4.46$ at pen${=}0$. \textsc{ISIR}'s iteration cap and collision buffers are enlarged so that they never truncate a run prematurely; its state at $N{=}5000$ is limited by the wall-time budget.

\paragraph{Timing and software.}
Each method's reported time is $\text{setup}+\text{solve}$: setup covers the method's own per-scene construction (oracle or BVH build, OBB system, GPU framework initialization and JIT, with synchronized timers on GPU) and solve covers the solver loop including its own contact detection; scene generation, mesh loading and the shared final evaluation are excluded. We use Python 3.13 with NumPy 2.4, SciPy 1.17, OSQP 1.1.1, python-fcl 0.7.0, trimesh 4.11 and Warp 1.14; \textsc{ISIR} and \textsc{Drake} run in a Python 3.10 environment with PyTorch 2.6 (CUDA 12.4) and Drake 1.51.

\paragraph{Evaluation conventions.}
\emph{(a)~Common final target.} Every method is scored against the same target $\tau{=}0$ by the evaluator of Algorithm~\ref{alg:evaluator} (strict zero threshold, no tolerance band), complemented for watertight meshes by an offline containment audit of the retained final states; internal margins ($\dhat$ for \name, QP/LCP and Soft-Penalty, $10^{-3}$ for PD-PGS, $10^{-4}$ for Drake) are solver parameters, not the target.
\emph{(b)~Proxy solvers are scored on meshes.} \textsc{AVBD-OBB}'s final poses are applied back to the triangle meshes before evaluation.
\emph{(c)~Per-seed budget.} $1800$\,s per (method, $N$, seed) cell, counted on the method's own setup and solve; Rigid-ISIR stops itself at $1770$\,s so that it can return its best pose; \textsc{Drake} is allowed $7200$\,s and \textsc{AVBD-OBB} $3600$\,s so that their step budgets can complete.
\emph{(d)~Incomplete seeds.} Unless a study states otherwise, every cell is attempted on all three seeds; a cell with an incomplete seed is annotated with its completion count (e.g.\ ``$2/3$, 1 T/O''), its mean is taken over the completed seeds, and no speedup ratio is quoted against it; a cell with no completed seed is reported as T/O (time-out) or OOM (out of memory).

\paragraph{Scenes.}
\emph{Kubric}: our primary pool of $40$ watertight household meshes from Google Scanned Objects~\cite{downs2022gso}, the collection the Kubric generator~\cite{greff2022kubric} draws on, with mild-to-moderate non-convexity ($\kappa$ from $0$ for the convex templates to $0.54$ under the concavity of App.~\ref{sec:exp_objaverse}); we place them with the procedural spawn of App.~\ref{ap:scene_gen} and write \emph{Kubric} for this pool throughout.
\emph{HY3D-Bench}~\cite{tencent2026hy3dbench}: a $300$-asset collection of synthesized meshes from which the benchmark filters retain $13$ templates in the decimated pool and $12$ in the full-resolution pool, plus a $57$-mesh high-concavity sub-pool with $\kappa\in[0.3,0.95]$ for the non-convex stress test (App.~\ref{sec:exp_objaverse}).
\emph{Thingi10K}~\cite{zhou2016thingi10k}: a $430$-mesh subset of artist-authored printable meshes, closed, manifold and non-self-intersecting, with $500$--$5000$ vertices and $1000$--$1500$ facets.
Three seeds ($42$, $123$, $456$) are used unless a study states otherwise (five for the high-concavity test, eight for the packing study); the seed spread is small (\name's RMSD standard deviation is at most $0.006$ on CPU and $0.0074$ on GPU across scene sizes), so the tables quote means and Fig.~\ref{fig:scaling} shows $\pm$s.d.\ error bars.

\paragraph{Metrics.}
Given a solver's output poses, the same triangle-mesh evaluator (App.~\ref{ap:evaluator}) computes
(i)~\emph{Pen.}, the number of body pairs with $\widetilde d_{ij}<0$ under the FCL convention of Algorithm~\ref{alg:evaluator};
(ii)~\emph{maxPen}, reported in the stress tests where residual penetration occurs: the maximum reported local contact depth over those pairs (zero when none remains);
(iii)~the root-mean-square reference-center displacement from the initial configuration, $\mathrm{RMSD}=\smash{\sqrt{\tfrac1N\sum_i\|\mathbf{p}_i-\mathbf{p}_i^0\|^2}}$ over the body reference centers (Sec.~\ref{sec:progressive_scaling}); and
(iv)~the wall-clock time under the timing protocol above.
RMSD measures translational layout deviation: for the translation-only solvers every surface point of a body moves by its center displacement, so center RMSD and surface RMSD coincide, while for AVBD and ISIR, which also rotate bodies, it can understate the total surface motion. Because all methods share the evaluator, differences in (i)--(iii) arise from the solver output rather than from method-specific scoring.

\subsection{Mesh-Level Penetration Evaluator}
\label{ap:evaluator}

All seven solvers in Sec.~\ref{sec:exp_main} are scored by the same triangle-mesh evaluator. It applies one strict zero threshold to every method and counts the pairs whose piecewise score is negative; we give its pseudocode so the reported \emph{Pen.} counts can be reproduced exactly.

\begin{algorithm}[t]
\SetKwInOut{Input}{input}\SetKwInOut{Output}{output}
\Input{pose-updated world-space meshes $\{\mathcal M_i = (\mathcal V_i, \mathcal F_i)\}_{i=1}^N$ (one FCL BVH per body)}
\Output{negative-score pair count $\mathrm{Pen.}$; local-depth diagnostic $\mathrm{maxPen}$}
$\mathrm{Pen.} \leftarrow 0;\; \mathrm{maxPen}\leftarrow 0$\;
\ForEach{pair $(i,j)$, $i<j$ (vectorized upper-triangle AABB pass over all $N(N{-}1)/2$ pairs)}{
  \If{the AABBs are disjoint}{
    $g_{ij}^{\mathrm{AABB}} \leftarrow$ AABB axis-gap norm $\geq 0$ \tcp*{box-gap lower bound; feeds only the min-gap diagnostic; the pair is not penetrating}
  }
  \Else{
    $g \leftarrow \texttt{FCL.distance}(\mathcal M_i,\mathcal M_j)$ \tcp*{BVH--BVH; returns a non-negative gap}
    \If{$g > 0$}{$\widetilde d_{ij} \leftarrow g$ \tcp*{separated: FCL-reported boundary gap}}
    \Else{
      $\mathcal C \leftarrow \texttt{FCL.collide}(\mathcal M_i,\mathcal M_j,$
        \Indp
        $\texttt{enable\_contact},\ \texttt{num\_max\_contacts}{=}16)$\;
        \Indm
      $\widetilde d_{ij} \leftarrow \begin{cases}0 & \mathcal C=\emptyset\ \text{(by convention)}\\[2pt] -\max\limits_{c\in\mathcal C} c.\mathrm{pen\_depth} & \text{otherwise}\end{cases}$\;
    }
    \If{$\widetilde d_{ij} < 0$}{
      $\mathrm{Pen.}\leftarrow\mathrm{Pen.}+1$\;
      $\mathrm{maxPen}\!\leftarrow\!\max(\mathrm{maxPen},\,-\widetilde d_{ij})$\;
    }
  }
}
\caption{Mesh-level penetration evaluator.\label{alg:evaluator}}
\end{algorithm}

The threshold is exactly zero: a pair is counted only when $\widetilde d_{ij}<0$, with no tolerance band; a pair whose \texttt{collide} call reports contact but returns an empty contact set is scored as zero by convention ($\widetilde d_{ij}=0$), which counts it as non-penetrating. The online score packs two different quantities into one signed variable: a positive boundary distance from \texttt{distance}, and a negative local contact-depth summary from \texttt{collide}. The sign alone partitions them, and the positive branch never contributes to either penetration metric. The reported online metric is \emph{Pen.}\ $=|\{(i,j): \widetilde d_{ij}<0\}|$; the companion diagnostic is \emph{maxPen} $=\max(-\widetilde d_{ij})$ over those pairs, zero when there are none. The depth is the maximum over the at most $16$ contacts FCL returns and is a local triangle-pair intersection depth, not a global minimum-translation distance, so \emph{maxPen} is a cap- and tessellation-dependent local diagnostic; \emph{Pen.}\ depends only on the sign and not on this estimate.

The online evaluator has no containment branch, so a fully nested watertight pair can be assigned a positive boundary gap. We therefore re-score every retained final state of the main, scaling, cross-dataset and spherical-spawn runs offline with a separate containment test; that audit reproduces the tabulated \emph{Pen.}\ values and finds no missed nested pair; the first-feasible \textsc{AVBD-OBB} poses lie outside it. FCL is invoked only for the $P$ pairs whose boxes overlap ($P \ll N^2$ in practice).

\subsection{Optional 6-DOF Extension}
\label{ap:6dof}

The default \name solver uses the translation-only QP of Eq.~\ref{eq:qp_translation}. For scenes whose initial layout is packing-limited, where re-orientation can reduce the displacement needed to separate the bodies (in free space, bounded bodies can always be separated by translation alone), the QP can be extended with a per-body small-angle rotation increment $\boldsymbol{\omega}_i\in\mathbb{R}^3$, exposed via a \texttt{-{}-rotation} command-line flag. Let $\vw_i,\vw_j$ be the witness points of a contact pair and $\mathbf{y}_i = \vw_i-\mathbf{p}_i$, $\mathbf{y}_j = \vw_j-\mathbf{p}_j$ the lever arms; the rotational contribution to the first-order increment of the projected gap $\mathbf{n}_{ij}^{\top}(\vw_j-\vw_i)$ is $(\mathbf{y}_j\times \mathbf{n}_{ij})^{\top}\boldsymbol{\omega}_j - (\mathbf{y}_i\times \mathbf{n}_{ij})^{\top}\boldsymbol{\omega}_i$, with $\boldsymbol{\omega}_i$ a world-frame increment applied as $\mathbf{R}_i^{+}=\exp([\boldsymbol{\omega}_i]_\times)\mathbf{R}_i$, where $[\boldsymbol{\omega}]_\times$ is the cross-product matrix, $[\boldsymbol{\omega}]_\times\mathbf{a}=\boldsymbol{\omega}\times\mathbf{a}$. Substituting into the per-step constraint gives
\begin{equation*}
\begin{aligned}
  \min_{\Delta \mathbf{p}^\star, \boldsymbol{\omega}^\star} \quad & \frac{1}{2}\sum_{i=1}^N \left(\|\Delta \mathbf{p}_i^\star\|^2 + \beta\|\boldsymbol{\omega}_i^\star\|^2\right)\\
  \text{s.t.}\quad & \mathbf{n}_{ij}^{\top}(\Delta \mathbf{p}_j^\star-\Delta \mathbf{p}_i^\star) + (\mathbf{y}_j\times\mathbf{n}_{ij})^{\top}\boldsymbol{\omega}_j^\star\\
  & \quad - (\mathbf{y}_i\times\mathbf{n}_{ij})^{\top}\boldsymbol{\omega}_i^\star \geq b_{ij}, \quad \forall (i,j)\in\mathcal{A}_k,\\
  & \|\boldsymbol{\omega}_i^\star\|_\infty \leq \omega_{\max}, \quad \forall i.
\end{aligned}
\end{equation*}
The weight $\beta=r_{\max}^2$ matches the rotational term to the translational one in length$^2$ units: the surface-arc displacement of a small-angle rotation $\boldsymbol{\omega}$ on a body of radius $r$ is at most $r\|\boldsymbol{\omega}\|$, so $\beta=r_{\max}^2$ weighs the worst-case rotational surface motion like a translation. The box bound $\|\boldsymbol{\omega}_i\|_\infty\leq\omega_{\max}=0.1$\,rad per step limits the small-angle linearization error of the lever arm.
Each accepted rotation changes the witness geometry, so it invalidates the affected cache and warm-start entries and triggers fresh nearby-pair queries.
Sec.~\ref{sec:exp_rotation} evaluates this variant: cost on typical layouts (Tab.~\ref{tab:abl_rotation}), and the packing-limited regime where it lowers displacement (Tab.~\ref{tab:rotation_packing}).

\subsection{Upright-on-Plane Layout Constraint}
\label{sec:upright_plane}

Generated-asset layout repair often has a stronger scene prior than the fully unconstrained 3D problem: objects are meant to stand upright on a tabletop or floor, and penetration should be removed without letting the solver lift objects into the air or tilt them sideways.  We encode this prior as an optional hard constraint in the S4R optimization, rather than as a Blender post-process.

Let $\vu$ be the unit up vector and let the support plane be $\Pi=\{\vx\in\mathbb R^3:\vu^\top\vx=h_0\}$.  For each body we first choose an upright canonical orientation $\mathbf{R}_i^{\mathrm{up}}$ whose semantic up axis aligns with $\vu$.  Its admissible orientation is then parameterized by a single yaw angle about $\vu$:
\begin{equation}
  \mathbf{R}_i(\theta_i) = \mathbf{R}_{\vu}(\theta_i) \mathbf{R}_i^{\mathrm{up}},
  \qquad
  \theta_i^{k+1}=\theta_i^k+\Delta\theta_i,
  \end{equation}
where $\mathbf{R}_{\vu}(\theta_i)$ is a rotation about $\vu$.  Roll and pitch are not decision variables because they would violate the upright prior; yaw is optimized simultaneously with translation during progressive inflation.

If a generated initial layout contains tilted objects, as in the Blender application of Sec.~\ref{sec:exp_generated_asset}, we treat the tilt as part of the continuation path rather than as an admissible final degree of freedom.  Let $\mathbf{T}_i^0$ denote the initial roll/pitch tilt and let $\zeta(s)$ be a smooth monotone tilt schedule with $\zeta(s_{\min})=1$ and $\zeta(1)=0$ (written $\zeta$ to keep $\alpha$ for the tail damping of Sec.~\ref{sec:tail}). Our implementation uses the smoothstep complement $\zeta(s)=1-\sigma^2(3-2\sigma)$ with $\sigma=\operatorname{clamp}\bigl((s-s_{\min})/(1-s_{\min}),0,1\bigr)$, and interpolates the tilt by scaling the initial roll/pitch Euler angles, $\mathbf{T}_i(\zeta)=\mathbf{R}_y(\zeta\,\varphi_i^{\mathrm{pitch}})\,\mathbf{R}_x(\zeta\,\varphi_i^{\mathrm{roll}})$, with $\mathbf{R}_x,\mathbf{R}_y$ the rotations about the two horizontal axes, so that $\mathbf{T}_i(1)=\mathbf{T}_i^0$. The initial orientation is decomposed as $\mathbf{R}_i^0=\mathbf{R}_{\vu}(\theta_i^0)\,\mathbf{T}_i^0\,\mathbf{R}_i^{\mathrm{up}}$ (initial yaw, tilt, and canonical upright orientation), and during inflation we use $\mathbf{R}_i(s,\theta_i)=\mathbf{R}_{\vu}(\theta_i)\mathbf{T}_i(\zeta(s))\mathbf{R}_i^{\mathrm{up}}$, so the path starts from the tilted layout and ends upright.  The QP still optimizes the rotation term that is physically admissible on the support plane---yaw---while the roll/pitch component is driven to zero by the upright scene prior. This variant re-detects at every scale step, without event skipping or caching, runs the tail pass once the tilt schedule reaches $(s,\zeta)=(1,0)$, and enters no timing comparison (Tab.~\ref{tab:generated_asset_repair} times the unconstrained 3-DOF repair).

The support constraint is imposed at every scale step.  For any prescribed orientation $\mathbf{R}$, let $\underline h_i(\mathbf{R})=\min_{\bar{\mathbf{v}}}\vu^\top \mathbf{R}\bar{\mathbf{v}}$, with $\bar{\mathbf{v}}$ ranging over the body's vertex offsets from its reference center as in Sec.~\ref{sec:linearized_constraints}, be the lowest local vertex height of body $i$; during the tilt transition this quantity is evaluated at the deterministic target roll/pitch state $\mathbf{R}_i(s_{k+1},\theta_i)$, and at the upright orientation once $\zeta{=}0$.  Since the scaled world-space surface is $\vp_i+s\,\mathbf{R}_i\bar{\mathbf{v}}$, body $i$ touches the support plane at scale $s$ when
\begin{equation}
  \vu^\top \vp_i(s) + s\,\underline h_i(\mathbf{R}_i) = h_0 .
  \end{equation}
During the transition $s_k\!\to\!s_{k+1}$, we append the linear equality
\begin{equation}
  \vu^\top(\vp_i^k+\Delta\vp_i) + s_{k+1}\,\underline h_i(\mathbf{R}_i) = h_0,
  \qquad i=1,\ldots,N,
  \label{eq:support_qp_constraint}
\end{equation}
to the contact QP of Eq.~\ref{eq:qp_translation}.  The equality is imposed by parameterizing $\vp_i$ with two in-plane coordinates and the deterministic height $\vu^\top\vp_i(s)=h_0-s\,\underline h_i(\mathbf{R}_i)$.  Because yaw is around $\vu$, $\underline h_i$ is invariant to $\theta_i$ for any fixed roll/pitch homotopy state, and in particular for the final upright state.

The resulting tabletop QP uses variables $(\Delta\vp_i,\Delta\theta_i)$. We split the center motion as $\vp_i(s)=\vp_i^{\parallel}(s)+\bigl(h_0-s\,\underline h_i(\mathbf{R}_i(s))\bigr)\vu$: the QP optimizes only the in-plane part $\Delta\vp_i^{\parallel}$, while the normal component is the deterministic height that keeps the body on the plane. That deterministic vertical motion, and the tilt-induced motion when the roll/pitch homotopy is active, are not part of the QP's contact linearization: the contact rows of Eq.~\ref{eq:upright_yaw_qp_constraint} form a simplified model that keeps the right-hand side $b_{ij}$, whereas eliminating Eq.~\ref{eq:support_qp_constraint} exactly from the full contact row would replace it by $b_{ij}-(\vn_{ij}^\top\vu)(\Delta z_j-\Delta z_i)$, with $\Delta z_i$ the prescribed height change of body $i$; any residual effect that persists after the accepted update is observed and re-linearized at the next per-step exact detection.  For an active contact $(i,j)$ with witness lever arms $\vy_i=\vw_i-\vp_i$ and $\vy_j=\vw_j-\vp_j$ (where $\vw_i,\vw_j$ are the witness points), the yaw-linearized constraint is
\begin{equation}
\begin{aligned}
  \vn_{ij}^{\top}(\Delta\vp_j^{\parallel}-\Delta\vp_i^{\parallel})
  &+ \Delta\theta_j\,(\vy_j\times \vn_{ij})^{\top}\vu
   - \Delta\theta_i\,(\vy_i\times \vn_{ij})^{\top}\vu \\
  &\ge b_{ij}.
\end{aligned}
\label{eq:upright_yaw_qp_constraint}
\end{equation}
Our implementation evaluates $b_{ij}$ with the projected width $W_i(\vn)=\max_{\bar{\mathbf{v}}}\vn^\top \mathbf{R}_i\bar{\mathbf{v}}-\min_{\bar{\mathbf{v}}}\vn^\top \mathbf{R}_i\bar{\mathbf{v}}$ of each body in place of the one-sided support $e_i(\vn)$, and clips each in-plane translation to $0.03$ ($0.018$ in the tail pass) after the solve; yaw is bounded by additional QP rows.
The objective adds a yaw regularizer to the minimum-norm displacement objective,
\begin{equation}
  \min_{\Delta\vp,\Delta\boldsymbol{\theta}}
  \frac{1}{2}\sum_i\|\Delta\vp_i\|^2
  +\frac{\beta_{\theta}}{2}\sum_i(\Delta\theta_i)^2,
  \end{equation}
with $\beta_{\theta}$ in length$^2$ units; our implementation uses $\beta_{\theta}=t^2$ for bodies normalized to longest extent $t$, and clamps $|\Delta\theta_i|\leq0.08$\,rad per step ($0.04$ in the tail pass).  This is the formulation intended for Blender tabletop generated-asset repair: \name resolves penetration during inflation through coupled in-plane translation and yaw, while the support-plane equality prevents lifting and the upright parameterization prevents tipping.

This constraint changes the feasible set and should therefore be evaluated separately from the unconstrained volumetric benchmarks.  It is appropriate for tabletop layout repair, shelf organization, and generated-scene cleanup, but not for scenes where valid resolution requires stacking, lifting, or tipping objects, nor for multi-support configurations (e.g.\ a floor plus a wall) whose simultaneous constraints from different directions remove the yaw freedom this parameterization relies on; Sec.~\ref{sec:conclusion} discusses this limitation.

\begin{figure}[htbp]
  \centering
  \includegraphics[width=\columnwidth]{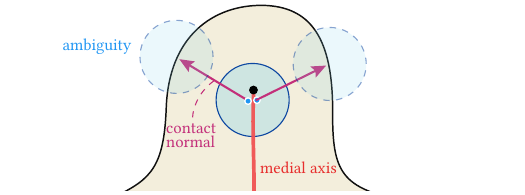}
  \caption{Medial-axis ambiguity under deep penetration.}
  \Description{A schematic of a deeply penetrated configuration near the medial axis, where the closest point and contact normal can switch discontinuously.}
  \label{fig:medial_axis}
\end{figure}

\begin{figure*}[t]
  \centering
  \includegraphics[width=0.9\textwidth]{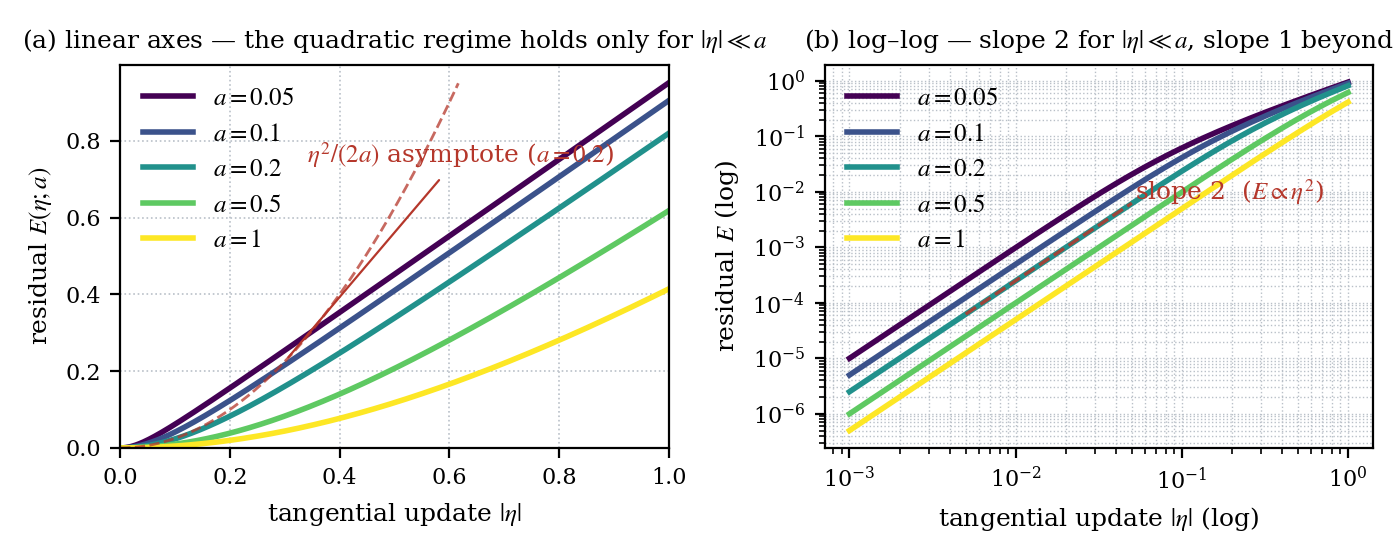}
  \caption{Exact first-order residual for the signed distance to a circle (Eq.~\ref{eq:circle_residual}). For an interior query point whose remaining distance to the medial axis is $a>0$, a tangential update of magnitude $|\eta|$ leaves the residual $\mathcal{E}(\eta;a)=\sqrt{a^{2}+\eta^{2}}-a$. The quadratic approximation $\mathcal{E}\approx\eta^{2}/(2a)$ holds only for $|\eta|\ll a$, and the coefficient $1/(2a)$ degrades as $a\downarrow0$. The plot illustrates the local smooth-branch argument.}
  \Description{Two plots of the exact circle-distance residual versus the tangential update magnitude, for several positive values of the remaining distance to the medial axis, showing the local quadratic regime and the linear growth for large updates.}
  \label{fig:truncation_error}
\end{figure*}

\subsection{Non-Convex Test (HY3D-Bench High-Concavity)}
\label{sec:exp_objaverse}

Kubric objects are at most moderately non-convex ($\kappa\leq0.54$; App.~\ref{ap:benchmark_limits}).
To stress the non-convex regime, we draw $34$ distinct meshes, one per body, from the HY3D-Bench~\cite{tencent2026hy3dbench} high-concavity sub-pool of $57$ eligible meshes ($\kappa \in [0.3, 0.95]$ after the winding-consistency and face-count filters; the largest $\kappa$ in the pool is $0.548$).
Concavity is defined as $\kappa = 1 - \operatorname{area}(\text{hull})/\operatorname{area}(\text{mesh})$ (distinct from the local-reach parameter $\rho$ of App.~\ref{ap:theory_linearization}); objects with $\kappa > 0.5$ have over $2\times$ more surface area than their convex hulls.

\begin{table}[t]
  \caption{Non-convex stress test on \textsc{HY3D-Concave}, the HY3D-Bench high-concavity subset (concavity $\kappa \in [0.3, 0.95]$).
  Averages over 5 random seeds. Both methods reach pen${=}0$; the separating axis is displacement: \textsc{AVBD-OBB} resolves phantom OBB overlaps of the concave meshes and displaces bodies $126\times$ farther than \name.}
  \label{tab:objaverse}
  \centering
  \small
  \setlength{\tabcolsep}{17pt}
  \resizebox{\columnwidth}{!}{
  \begin{tabular}{@{}lccc@{}}
    \toprule
    Method & Pen. (avg) & RMSD (avg) & Time \\
    \midrule
    AVBD-OBB           & 0.0 & 2.530 & \textbf{3.6 s}  \\
    \textbf{\name (ours)}& 0.0 & \textbf{0.020} & 5.8 s \\
    \bottomrule
  \end{tabular}
  }
\end{table}

This sub-pool contains non-watertight source meshes, so it is scored with the negative-score metric of Algorithm~\ref{alg:evaluator} alone; the containment audit applies to watertight inputs.
On highly non-convex meshes (Tab.~\ref{tab:objaverse}) both methods reach pen${=}0$ in all $5$ seeds, but \name's displacement is $126\times$ lower ($0.020$ vs.\ $2.530$): the concave meshes occupy a small fraction of their bounding boxes, so the OBB proxy reports overlaps the meshes do not have and \textsc{AVBD-OBB} separates bodies much farther than mesh-level contact requires. This is consistent with a mesh-level formulation being essential when geometry does not approximate its bounding box.

\begin{figure}[t]
  \centering
  \includegraphics[width=\columnwidth]{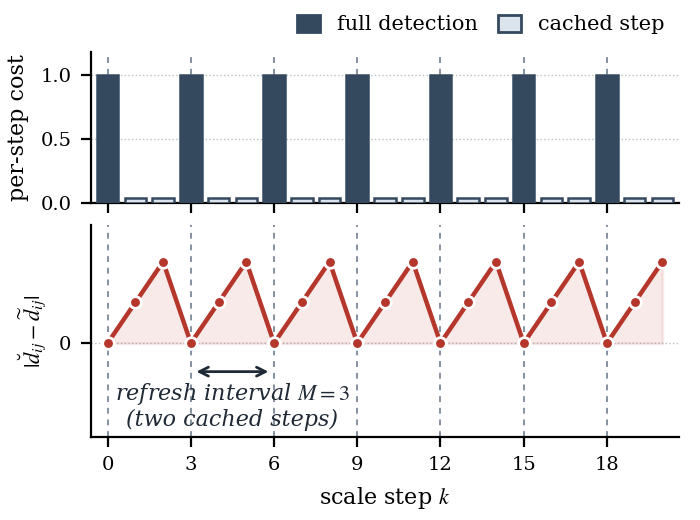}
  \caption{Schematic of the cache rhythm with refresh interval $M{=}3$: full BVH/narrow-phase detection every third scale step (top, dark bars), two cached steps in between; the residual axis is $|\breve d_{ij}-\widetilde d_{ij}|$; the cached-score residual (bottom) accumulates within each window and is reset to zero at every full detection. The panel illustrates the cadence rather than a measured trace; App.~\ref{ap:cache_audit} reports the final feasibility at each cache interval.}
  \Description{A two-part diagram showing full collision-detection calls every three scale steps and the corresponding cached-distance residual rising between refreshes before resetting to zero.}
  \label{fig:cache_rhythm}
\end{figure}

\subsection{Spherical-Spawn Stress Test}
\label{sec:exp_sphere}

To probe how each solver behaves when crowding outranks scene size, we place bodies uniformly inside a ball of radius $r_{\mathrm{ball}}(N){=}0.10\,(N/100)^{1/3}$ so cell density stays roughly invariant across $N{\in}\{100, 500, 1000, 2000\}$.  Initial penetrating pairs grow from about $825$ to $\sim\!25\,000$ and the deepest pair reaches about $9.5$\,cm.

Under extreme density, \name's QP update minimizes the total squared incremental displacement subject to the current linearized margin rows, which keeps RMSD small (${\sim}0.13$--$0.31$ here); at its default tail budget it leaves under $1\%$ of the initial pairs from $N{=}500$ on ($1.3$, $32$, and $192$ pairs at $N{=}500$, $1000$, $2000$), and a $200$-iteration tail clears them at $1.1$--$1.7\times$ the wall time ($20$, $64$, $215$\,s) with RMSD at most $5.1\%$ higher. QP/LCP reaches pen${=}0$ at every $N$, at two to three times the displacement ($0.24$--$0.92$) and, at $N{=}2000$, three times the wall time of the longer-tail \name\ run; \textsc{ISIR} reaches pen${=}0$ up to $N{=}1000$ at four to six times \name's displacement and stops at its budget on two of three seeds at $N{=}2000$; \textsc{PD-PGS} reaches pen${=}0$ up to $N{=}500$ and returns a budget-limited state at $N{=}1000$; \textsc{Soft-Penalty} times out from $N{=}500$, and \textsc{Drake} completes two of three seeds at $N{=}100$ (RMSD $0.09$--$0.14$ in $1600$--$1800$\,s) and none at $N{=}500$ within its $7200$\,s allowance.  With the longer tail, the methods that reach zero penetration differ mainly in displacement and wall time.
Fig.~\ref{fig:sphere_baselines_gallery} renders \name\ and QP/LCP, the two methods that reach zero penetration at every $N$, under a common camera per $N$.

\subsection{Rigid-Body Engines on an Interpenetrating Scene}
\label{sec:exp_engine_frontend_app}

The IPC front-end of Sec.~\ref{sec:exp_engine_frontend} checks, under the shared evaluator, the intersection-free start that a barrier solver presupposes. Impulse- and penalty-based engines can step from an overlapping state, but the contact
response has to remove the overlap, and on a deeply interpenetrating layout that response is violent.
We measure this on three independent engines (Tab.~\ref{tab:engine_frontend}).

\begin{table}[t]
  \caption{Peak body speed over a one-second rollout from an interpenetrating YCB layout and from the same
  layout after \name\ (3-seed means, $N{=}50$, gravity disabled so the motion is contact response alone).}
  \label{tab:engine_frontend}
  \centering
  \small
  \setlength{\tabcolsep}{9pt}
  \renewcommand{\arraystretch}{1.35}
  \begin{tabular}{@{}lcc@{}}
    \toprule
    Engine & Raw scene & After \name \\
    \midrule
    MuJoCo 3.10              & $7.8$\,m/s   & $0.000$\,m/s \\
    PyBullet 3.2.7           & $12.5$\,m/s  & $0.001$\,m/s \\
    Isaac Gym (PhysX)          & $230.6$\,m/s & $0.000$\,m/s \\
    \bottomrule
  \end{tabular}
\end{table}

We drop $50$ YCB~\cite{calli2015ycb} instances, resampled from a $16$-object pool, into a tight spawn; the raw
layouts carry $773$ overlapping pairs on average (up to $0.11$ deep). Every engine ejects the bodies, mildly
in MuJoCo and some $30\times$ more in PhysX~\cite{physx5}, and the maximum body displacement over the one-second
rollout ($7.2$, $10.0$, and $230.8$ units) is close to the peak speed times the rollout duration. \name\ resolves the same layouts to
pen${=}0$ in $8.2$\,s with a closest gap of $2.6\times10^{-4}$--$6.4\times10^{-4}$, after which all three engines hold
the configuration essentially at rest. The collision geometry is each object's convex hull, used identically
by \name\ and by the engines, so the comparison does not turn on a mesh-versus-proxy mismatch.

\subsection{Contact-Density Stress Test}
\label{sec:exp_density}

The scaling experiments above keep contact density approximately constant by expanding the spawn box with $N$.
We additionally test fixed-size scenes while varying the spawn-box multiplier, which directly changes the number of initial penetrating pairs and the density of the active-contact graph.
Tab.~\ref{tab:density_sweep} reports how crowding changes feasibility, displacement, and wall time at fixed scene size.

\begin{table}[t]
  \caption{Contact-density sweep on Kubric and HY3D-Bench at $N{=}1000$, 3-seed averages. Smaller spawn multipliers create denser initial interpenetration; displacement and wall time grow as packing tightens, and only the densest spawn leaves residual pairs at the default tail budget. The multiplier-$1.00$ Kubric row is the scene of Tab.~\ref{tab:scaling}.}
  \label{tab:density_sweep}
  \centering
  \small
  \setlength{\tabcolsep}{9pt}
  \resizebox{\columnwidth}{!}{%
  \begin{tabular}{@{}lccccc@{}}
    \toprule
    Dataset / $N$ & Spawn mult. & Init. Pen. & Final Pen. & RMSD & Time \\
    \midrule
    Kubric / 1000     & 0.70 & 2160 & 9.3 & 0.106 & 26 s \\
    Kubric / 1000     & 0.85 & 1230 & 0  & 0.054 & 16 s \\
    Kubric / 1000     & 1.00 & 772  & 0 & 0.035 & 8 s \\
    HY3D-Full / 1000 & 0.70 & 1793 & 21.7 & 0.088 & 66 s \\
    HY3D-Full / 1000 & 1.00 & 641  & 0  & 0.032 & 29 s \\
    \bottomrule
  \end{tabular}%
  }
\end{table}

This study separates scene-size effects from crowding effects: at fixed $N$, denser initial layouts create larger active-contact graphs and increase RMSD and wall time; only the densest spawn leaves residual pairs at the default tail budget, and a $200$-iteration tail clears them at about $1.3\times$ the wall time.
The confinement failure mode discussed in Sec.~\ref{sec:conclusion} appears one regime tighter, when a container bounds the free volume: with wall constraints at packing fraction $\phi\approx0.14$ (total mesh volume over the volume of the cubic container) all three seeds still separate (pen${=}0$), while at $\phi\approx0.18$ none does ($12$--$35$ residual pairs per seed) --- as free volume shrinks, the continuation first becomes expensive and then fails. This locates the empirical onset of failure for the current solver.

\subsection{Ablation of Technical Contributions}
\label{sec:exp_ablation}

We isolate each technical contribution of \name.
Each ablation switches one component on or off with everything else at the main-table setting, so the default rows agree with the $N{=}40$ entry of Tab.~\ref{tab:main} to within $10\%$. All ablations use Kubric with $N{=}40$, $\dhat{=}0.02$, $\mathit{ds}_{\max}{=}0.05$, 3-seed means, and the $M{=}3$ cache setting of the main tables unless stated otherwise.

\subsubsection{Progressive Scaling vs.\ Direct QP}
Progressive scaling decomposes a problem with a deeply interpenetrating initialization into a sequence of locally convex QPs.
To isolate what progressive scaling contributes, we compare \name\ against two full-scale \emph{Direct-QP} variants --- a one-shot linearization, and an iterative variant that re-detects and re-linearizes for up to $50$ outer iterations --- both running their global QP on the full-size configuration without progressive scaling.

\begin{table}[t]
  \caption{Progressive-scaling ablation on Kubric ($N{=}40$, 3-seed mean). Both feasible variants reach pen${=}0$; progressive scaling does so through shallow-contact QPs and in less time. Among the feasible runs, bold marks the best entry per column and underline the runner-up (lower is better).}
  \label{tab:abl_scaling}
  \centering
  \small
    \setlength{\tabcolsep}{9pt}
  \resizebox{\columnwidth}{!}{
  \begin{tabular}{@{}lcccc@{}}
    \toprule
    Method                       & Pen.  & RMSD   & Feasible?  & Time \\
    \midrule
    Direct-QP (1 linearization)  & 13.7  & 0.028  & No         & 0.12 s \\
    Direct-QP (50 re-linearizations)  & \textbf{0}     & \underline{0.038}  & Yes        & \underline{0.41 s} \\
    \textbf{\name (progressive)} & \textbf{0} & \textbf{0.036} & Yes & \textbf{0.16 s} \\
    \bottomrule
  \end{tabular}
  }
\end{table}

Progressive scaling keeps the modeled contacts shallower, and on a fixed smooth branch Lemma~\ref{lem:taylor_residual} bounds the Taylor remainder by the square of the realized update. The Direct-QP~(50-iter) variant relies on outer iterations to walk out of the initial linearization; at larger $N$ the corresponding QP/LCP baseline becomes more expensive (Tab.~\ref{tab:scaling}: $3.0$--$4.6\times$ \name's wall time through $N{=}2000$).

\subsubsection{Tail Refinement.}
\label{sec:exp_tail_ablation}
The final tail-refinement phase rechecks the full-scale configuration and solves correction QPs on the near-contact set.
This phase is designed to remove residual linearization and stale-witness errors after the scale path reaches $s{=}1$.
Tab.~\ref{tab:abl_tail} quantifies the feasibility/runtime trade-off of this final correction.

\begin{table}[t]
  \caption{Tail-refinement ablation on three large-scale cells;
  3-seed averages. Tail refinement removes the residual penetrations
	  the main loop leaves behind, at the cost of an additional QP pass
	  over the near-contact set. The Off rows disable the tail so the main loop's own residual is visible; the On rows are the main configuration.}
  \label{tab:abl_tail}
  \centering
  \small
  \setlength{\tabcolsep}{8pt}
  \resizebox{\columnwidth}{!}{%
  \begin{tabular}{@{}llcccc@{}}
    \toprule
    Dataset / $N$ & Tail & Pen. & maxPen & RMSD & Time \\
    \midrule
    Kubric / 1000      & Off & 85.3 & 0.052 & \textbf{0.033} & \textbf{3 s} \\
    Kubric / 1000      & On  & \textbf{0} & \textbf{0.000} & 0.035 & 7 s \\
    HY3D-Full / 2000  & Off & 212.3 & 0.015 & \textbf{0.030} & \textbf{44 s} \\
    HY3D-Full / 2000  & On  & \textbf{0} & \textbf{0.000} & 0.034 & 66 s \\
    Thingi10K / 2000   & Off & 129.7 & 0.055 & \textbf{0.024} & \textbf{12 s} \\
    Thingi10K / 2000   & On  & \textbf{0} & \textbf{0.000} & 0.026 & 21 s \\
    \bottomrule
  \end{tabular}%
  }
\end{table}

\subsubsection{Scale of Impact (SOI) Analytical Event Handling.}
\label{sec:exp_soi}

We ablate the SOI rule of Sec.~\ref{sec:soi}, which advances toward a conservative lower bound on the earliest margin-entry scale under pure inflation at fixed poses (capped at $3\,\mathit{ds}_{\max}$), against two fixed-stride schedules: the coarse $\mathit{ds}{=}0.05$ stride of the default configuration and a deliberately fine $\mathit{ds}{=}0.01$ reference that takes about $100$ steps, most of them contact-free.

\begin{table}[t]
  \caption{Component ablation of the scale of impact (SOI) event rule on Kubric (3-seed average, $N{=}40$).
  Naïve $\mathit{ds}{=}0.01$ requires $\approx100$ scale steps; coarsening to $\mathit{ds}{=}0.05$ already removes most of that cost. SOI then advances toward a conservative lower bound on the earliest possible margin-entry scale (pure inflation, fixed poses) and skips the residual event-free steps. Bold marks the best Time and RMSD, underline the runner-up.}
  \label{tab:abl_soi}
  \centering
  \small
\setlength{\tabcolsep}{14pt}
  \resizebox{\columnwidth}{!}{
  \begin{tabular}{@{}lccc@{}}
    \toprule
    Method                         & Scale steps & Time   & RMSD \\
    \midrule
    Fixed $\mathit{ds}{=}0.01$     & 99.0        & 0.34 s & \textbf{0.0342} \\
    Fixed $\mathit{ds}{=}0.05$     & 20.0        & \textbf{0.15 s} & \underline{0.0348} \\
    \textbf{SOI jump (ours)}       & 17.7 & \underline{0.16 s} & 0.0359 \\
    \bottomrule
  \end{tabular}
  }
\end{table}

Relative to the coarse fixed stride ($20.0$ steps, $0.15$\,s, RMSD $0.0348$), SOI removes the remaining contact-free steps ($17.7$ steps) at essentially unchanged wall time and RMSD at this scene size ($0.16$\,s, $0.0359$); most of the cost of the fine schedule ($99.0$ steps, $0.34$\,s) is removed by coarsening the stride itself (Tab.~\ref{tab:abl_schedule}), so the gap to it is not a speedup attributable to SOI. What SOI changes is where the QPs are solved: at or before the potential contact events of the fixed-pose model rather than at arbitrary stride points, while contacts induced by translation are caught at the next refresh or in the tail.

\subsubsection{Collision Cache with Analytical Distance Update.}
\label{sec:exp_cache}
Full closest-point re-detection can approach a dense $\mathcal{O}(N^2)$ pair sweep, with each exact pair query carrying a backend- and mesh-dependent cost (the $q_{ij}$ term in App.~\ref{sec:theory_speed}).
Between events, we advance the cached predictions $\breve d_{ij}$ analytically by Eq.~\ref{eq:cache_update} and only re-run the full detector every $M$ scaling iterations.

\begin{table}[t]
  \caption{Ablation of the collision cache (refresh interval $M$) on Kubric (3-seed average, $N{=}40$). Bold marks the best Time and RMSD, underline the runner-up.}
  \label{tab:abl_cache}
  \centering
  \small
      \setlength{\tabcolsep}{14pt}
  \resizebox{\columnwidth}{!}{
  \begin{tabular}{@{}lccc@{}}
    \toprule
    $M$ & Main-loop detections & Time & RMSD \\
    \midrule
    1 (no cache)       & 19.0 & 0.234 s & \textbf{0.033} \\
    2                  &  9.3 & 0.186 s & \underline{0.034} \\
    \textbf{3 (ours)}  & 6.3 & \underline{0.157 s} & 0.036 \\
    5                  &  3.3 & \textbf{0.149 s} & 0.036 \\
    10                 &  2.0 & 0.160 s & 0.038 \\
    \bottomrule
  \end{tabular}
  }
\end{table}

At $N{=}40$, $M{=}3$ yields a $1.5\times$ speedup over $M{=}1$ at a small RMSD cost ($0.033\to0.036$). Setting $M{=}5$ is slightly faster still at the same RMSD, and $M{=}10$ gains nothing further while RMSD drifts to $0.038$ (cached distances become stale and the QP over-pushes). 

\paragraph{Cache Correctness Audit.}
\phantomsection\label{ap:cache_audit}
At every refresh interval $M{\in}\{1,2,3,5,10\}$ on Kubric ($N{=}40$, three seeds), the continuation ends with zero evaluator-reported penetration after the full-scale check and tail refinement. This certifies the final outcome at each interval, not step-wise agreement between the cached rows and a fully re-detected active set; a pair that enters the activation margin between two refreshes is recovered at the next refresh or in the tail.

\subsubsection{Schedule Type and Adaptive Skip.}
\label{sec:exp_schedule}

\begin{table}[t]
  \caption{Step size and adaptive skip on Kubric ($N{=}40$, 3-seed mean). Coarser steps dominate the speedup; adaptive skip removes the remaining contact-free steps at near-neutral cost at this scene size. Bold marks the best Time and RMSD, underline the runner-up.}
  \label{tab:abl_schedule}
  \centering
  \small
        \setlength{\tabcolsep}{7pt}
  \resizebox{\columnwidth}{!}{
  \begin{tabular}{@{}lllccc@{}}
    \toprule
    $\mathit{ds}_{\max}$ & Adaptive skip & Steps & Time & RMSD & Speedup \\
    \midrule
    $0.01$ (fine)        & No  & 99 & 0.33 s & \underline{0.034} & $1.0\times$ \\
    $0.01$ (fine)        & Yes & 76 & 0.31 s & \textbf{0.032} & $1.1\times$ \\
    $0.05$ (coarse)      & No  & 20 & \textbf{0.15 s} & 0.035 & $2.2\times$ \\
    \textbf{0.05 (coarse)} & Yes & 18 & \underline{0.16 s} & 0.036 & $2.1\times$ \\
    \bottomrule
  \end{tabular}
  }
\end{table}

Adaptive skip buys $1.1\times$ on the fine schedule and is cost-neutral on the coarse default because a fine schedule runs many inexpensive steps. The step-size choice is the bigger lever: with $\mathit{ds}_{\max}{=}0.05$ the solver takes $18$--$20$ steps with or without skip, and each step clears a meaningful fraction of the scale range.

\subsubsection{Contact Sparsity.}
\label{sec:exp_sparsity}
At each scale step only the bodies incident to an active contact enter the QP, about half of the bodies on average over a run (the per-step fraction ranges from a few percent to about nine tenths at $N{=}40$, $100$, and $1000$); variables are introduced only for those bodies and the constraint matrix is block-sparse.

\begin{table}[t]
  \caption{Contact sparsity exploitation on Kubric (3-seed average, $N{=}40$). Without sparsity the QP dimension is $3N$; with sparsity it is $3|\mathcal{B}|$ where $\mathcal{B}$ is the set of bodies in the active-contact graph.}
  \label{tab:abl_sparsity}
  \centering
  \small
  \resizebox{\columnwidth}{!}{
  \begin{tabular}{@{}lcccc@{}}
    \toprule
    Method & QP dim (avg) & QP iters & Time & RMSD \\
    \midrule
    Dense QP (no sparsity)      & 120          & 63          & 0.16 s          & 0.036 \\
    \textbf{Sparse QP (ours)}   & \textbf{56}  & 63 & 0.16 s & 0.036 \\
    \bottomrule
  \end{tabular}
  }
\end{table}

Sparsity roughly halves the QP dimension ($120\to56$ at $N{=}40$; about $1.6$k of $3000$ variables at $N{=}1000$) at an unchanged iteration count. Because contact detection dominates the wall time and candidate-pair generation remains a separate worst-case $\mathcal{O}(N^2)$ stage, this is a scalability feature rather than a measurable constant-factor speedup.

\subsection{Cumulative Optimization Stack}
\label{sec:exp_perf_summary}

\name's practical performance comes from the cumulative effect of several targeted optimizations.
Tab.~\ref{tab:perf_stack} summarizes the wall-time contribution of each on Kubric.

\begin{table}[t]
  \caption{Cumulative implementation stack on Kubric ($N{=}40$, 3-seed mean). Adaptive skip is cost-neutral at this scale under the Python backend (L3 to L4; Tab.~\ref{tab:abl_schedule} shows the same under FCL), and the controlled per-component effects are in Tab.~\ref{tab:abl_cache} and Tab.~\ref{tab:abl_sparsity}. The largest gain comes from replacing the Python BVH with FCL. Every stage reaches zero reported penetration on all three seeds except L3, which leaves one negative-score pair on one seed.}
  \label{tab:perf_stack}
  \centering
  \small
  \resizebox{\columnwidth}{!}{
  \begin{tabular}{@{}lcc@{}}
    \toprule
    Optimization (cumulative) & Time & Speedup \\
    \midrule
    L1: na\"ive (trimesh, fine $\mathit{ds}$, dense QP, no cache, no skip) & 68.0 s & --- \\
    L2: + contact sparsity (Sec.~\ref{sec:exp_sparsity})            & 67.8 s & $1.0\times$ \\
    L3: + coarse $\mathit{ds}_{\max}{=}0.05$                      & 21.2 s & $3.2\times$ \\
    L4: + adaptive skip (Sec.~\ref{sec:exp_schedule})               & 21.9 s & $3.1\times$ \\
    L5: + collision cache $M{=}3$ (Sec.~\ref{sec:exp_cache})        & 12.8 s & $5.3\times$ \\
    \textbf{L6: + FCL backend (ours)}                             & \textbf{0.16 s} & {\boldmath$428\times$} \\
    \bottomrule
  \end{tabular}
  }
\end{table}

\subsection{Large-Scale Runtime Breakdown}
\label{sec:exp_breakdown}

To expose the bottlenecks in the regime where \name is intended to be used, we also measure wall-time breakdowns at large scene sizes.
The measurement separates full mesh-proximity queries from QP setup/solve, SOI/cache maintenance, tail refinement, and final evaluation.
Tab.~\ref{tab:runtime_breakdown} reports the breakdown over three datasets at $N{=}2000$ and over Kubric at $N{\in}\{1000, 2000, 5000\}$.

\begin{table}[H]
  \caption{Large-scale runtime breakdown (per-phase seconds and \% of total; 3-seed means under the main tables' configuration; totals match Tab.~\ref{tab:scaling} and Tab.~\ref{tab:datasets} within $3\%$, and the \textsc{HY3D-Full} row corresponds to the full-resolution-pool comparison reported in Sec.~\ref{sec:exp_datasets}). Detection and tail refinement account for $73$--$97\%$ of every cell and the QP solve never exceeds $1\%$: savings come from fewer mesh-proximity queries, not a faster QP.}
  \label{tab:runtime_breakdown}
  \centering
  \small
  \setlength{\tabcolsep}{4pt}
  \resizebox{\columnwidth}{!}{%
  \begin{tabular}{@{}lcccccc@{}}
    \toprule
    Dataset / $N$ & Setup & Contact detection & QP solve & Tail refinement & Other & Total \\
    \midrule
    Kubric / 1000 & 0.1 s (1\%) & 3.0 s (39\%) & 0.1 s (1\%) & 4.1 s (54\%) & 0.4 s (5\%) & 7.6 s \\
    Kubric / 2000 & 0.3 s (1\%) & 5.9 s (34\%) & 0.1 s (1\%) & 9.3 s (53\%) & 1.8 s (10\%) & 17.4 s \\
    Kubric / 5000 & 1.5 s (2\%) & 15.1 s (23\%) & 0.3 s (1\%) & 33.4 s (50\%) & 16.3 s (24\%) & 66.6 s \\
    HY3D-Full / 2000 & 0.6 s (1\%) & 42.6 s (64\%) & 0.1 s (0\%) & 22.1 s (33\%) & 1.0 s (2\%) & 66.4 s \\
    Thingi10K / 2000 & 0.3 s (1\%) & 10.8 s (50\%) & 0.1 s (0\%) & 9.9 s (46\%) & 0.6 s (3\%) & 21.7 s \\
    \bottomrule
  \end{tabular}%
  }
\end{table}

Since the main-loop QP solve is essentially free at these problem sizes (the tail column combines detection with its own QPs), replacing OSQP with a faster QP backend (e.g., active-set or batched dense) would not meaningfully change the main-loop wall-time picture; future acceleration work should target the mesh-proximity calls and the tail-refinement loop instead.
Moreover, tail refinement is the largest phase on Kubric, on par with detection on Thingi10K, and second largest on HY3D-Full; its budget is the conservatively set iteration cap (Sec.~\ref{sec:tail}). A smarter tail --- warm-started from the main loop, or re-solving only the still-penetrating connected components --- would target the largest line of the Kubric rows directly.

\end{document}